\documentclass{statsoc}

\usepackage[a4paper]{geometry}

\usepackage[utf8]{inputenc}
\usepackage{lipsum}
\usepackage{subfiles}
\usepackage{booktabs}
\usepackage{chngcntr}
\usepackage{hhline}
\usepackage[colorlinks=true,linkcolor=blue,citecolor=blue]{hyperref}

\usepackage [english]{babel}
\usepackage [autostyle, english = american]{csquotes}
\MakeOuterQuote{"}

\usepackage{amsmath,amsthm,amssymb}
\usepackage{amsfonts} 

\usepackage{natbib}
\usepackage{tikz} 
\usepackage{caption,subcaption}
\usepackage{graphicx}
\usepackage{amsmath, amsthm, amsfonts, microtype,booktabs}
\usepackage{wrapfig}
\usepackage{lscape}
\usepackage{rotating}

\usepackage{xcolor}
\usepackage{algorithm}
\usepackage{algpseudocode}

\DeclareMathOperator*{\argmin}{arg\,min}

\newcommand{\ba}{\mathbf{a}}

\newcommand{\bbX}{\mathbb{X}}
\newcommand{\bbS}{\mathbb{S}}
\newcommand{\bbU}{\mathbb{U}}
\newcommand{\bbQ}{\mathbb{Q}}

\newcommand{\bbH}{\mathbb{H}}
\newcommand{\dagset}{\mathfrak{G}}

\newcommand{\indep}{\perp \!\!\! \perp}
\newcommand{\R}{\mathbb{R}}
\newcommand{\E}{\mathbb{E}}

\newcommand{\pr}{\mathbb{Q}}
\newcommand{\pa}{\textnormal{Pa}}
\newcommand{\Parents}[2]{\textnormal{Pa}_{#1}^{#2}}
\newcommand{\Ancestors}[2]{\textnormal{An}_{#1}^{#2}}

\newcommand{\cG}{\mathcal{G}}

\newcommand{\Gtrue}{\mathcal{G}^*}
\newcommand{\Gnat}{\bar{\mathcal{G}}}
\newcommand{\Pobs}{\mathbb{P}}

\newcommand{\noise}{\epsilon}
\newcommand{\cov}{\text{Cov}}
\newcommand{\Var}{\text{Var}}
\newcommand{\hilb}{\mathcal{H}}

\newcommand{\raj}[1]{\textcolor{black}{#1}}

\newcommand{\barH}{\bar{H}}
\newcommand{\barK}{\bar{K}}
\newcommand{\fspace}{\mathcal{F}}
\newcommand{\fkspace}{\tilde{\mathcal{F}}}

\newcommand{\rkspace}{\tilde{\mathcal{R}}}

\newcommand{\overbar}[1]{\mkern 1.5mu\overline{\mkern-1.5mu#1\mkern-1.5mu}\mkern 1.5mu}

\newcommand{\camobs}{\texttt{CAM-OBS}}
\newcommand{\pcssbic}{\texttt{PCSS-BIC}}
\newcommand{\lrpsbic}{\texttt{LRPC+BIC}}
\newcommand{\decamfound}{\texttt{DeCAMFound}}
\newcommand{\cam}{\texttt{CAM}}
\newcommand{\vanillabic}{\texttt{Vanilla-BIC}}

\usepackage{harpoon}

\newcommand{\edit}[1]{\textcolor{black}{#1}}

\newcommand{\raedit}[1]{\textcolor{black}{#1}}

\newcommand{\rvline}{\hspace*{-\arraycolsep}\vline\hspace*{-\arraycolsep}}

\usepackage{mathtools}
\usepackage{lipsum,graphicx,multicol}

\usepackage{mathtools}
\usepackage[capitalize, sort]{cleveref}

\usepackage{thmtools}
\theoremstyle{plain}
\newtheorem{nthm}{Theorem}

\newtheorem{nprop}{Proposition}
\newtheorem{nlem}{Lemma}
\newtheorem{ncor}{Corollary}

\theoremstyle{definition}
\newtheorem{ndefn}{Definition}
\newtheorem{nexa}{Example}
\newtheorem{nassum}{Assumption}

\crefformat{footnote}{#1\footnotemark[#2]#3}

\newif\ifdraft 
\draftfalse 

\newcommand{\myfigure}[2]{%
  \ifdraft
    \includegraphics[draft,#1]{#2} 
  \else
    \includegraphics[#1]{#2} 
  \fi
}

\title[The DeCAMFounder Score]{The DeCAMFounder: Non-Linear Causal Discovery in the Presence of Hidden Variables}
\author[]{Raj Agrawal} \coaddress{Raj Agrawal, Massachusetts Institute of Technology, Cambridge, MA, United States of America. Email: r.agrawal@csail.mit.edu.} 
\address{Laboratory for Information \& Decision Systems, Massachusetts Institute of Technology, Cambridge, USA}
\author[]{Chandler Squires}
\address{Broad Institute of MIT and Harvard, and Laboratory for Information \& Decision Systems, Massachusetts Institute of Technology, Cambridge, USA.}
\author[]{Neha Prasad}
\address{Laboratory for Information \& Decision Systems, Massachusetts Institute of Technology, Cambridge, USA.}
\author[]{Caroline Uhler}
\address{Broad Institute of MIT and Harvard, and Laboratory for Information \& Decision Systems, Massachusetts Institute of Technology, Cambridge, USA.}

\begin{document}

\begin{abstract}
Many real-world decision-making tasks require learning causal relationships between a set of variables. 
Traditional causal discovery methods, however, require that all variables are observed, which is often not feasible in practical scenarios.
Without additional assumptions about the unobserved variables, it is not possible to recover any causal relationships from observational data. Fortunately, in many applied settings, additional structure among the confounders can be expected.
In particular, \emph{pervasive confounding} is commonly encountered and has been utilized for consistent causal estimation in linear causal models. %
In this paper, we present a provably consistent method to estimate causal relationships in the non-linear, pervasive confounding setting.
The core of our procedure relies on the ability to estimate the confounding variation through a simple spectral decomposition of the observed data matrix.
We derive a DAG score function based on this insight, prove its consistency in recovering a correct ordering of the DAG, and empirically compare it to previous approaches.
We demonstrate improved performance on both simulated and real datasets by explicitly accounting for both confounders and non-linear effects.
\end{abstract}


\section{Introduction} 

Many decision-making and scientific tasks require learning causal relationships between observed variables. 
For example, biologists seek to identify causal pathways in gene-regulatory networks,  epidemiologists search for the causes of diseases in complex social networks, and data scientists assess the biases of their machine-learning models by investigating their causal structure \citep{causality_book, gene_expr_anal, pearl_causality, epidemiology_models, counter_fair}. 
However, estimating causal effects can be challenging when the underlying models lack sufficient flexibility. For example, an individual's health may generally improve with increasing exercise, but excessive exercise might lead to a decline. Additionally, unobserved confounders can pose difficulties in recovering causal relationships as they induce spurious correlations between observed variables. To accurately estimate causal relationships, it is essential to develop methods that account for both non-linear effects and unobserved confounders. Unfortunately, existing methods often fall short in at least one of these two aspects.

Some existing causal discovery methods assume linear effects and no unobserved confounders \citep{chick2002, greedy_sp, high_dim_PC}.
More recent works focus on recovering causal relationships when each variable is a non-linear function of its direct causes plus additive noise. However, these methods still assume the absence of unobserved confounders \citep{hoyer_ident, mooij_ident, additive_noise, cam}. 
In seminal work, \citet{richardson2002} introduced ancestral graphical models to account for unobserved confounding.
Several methods have been proposed to estimate these models \citep{fast_fci, gspo}. However, when confounders have an effect on many observed variables (referred to as "pervasive" confounding), ancestral graphs exhibit poor identifiability compared to methods that explicitly model pervasive confounders \citep{causal_lrps}. Since pervasive confounding arises in numerous economic and biological applications, it is crucial to develop methods that exploit this additional structural information. Unfortunately, existing methods that provably recover the true causal structure under pervasive confounding assume linear causal effects \citep{causal_lrps, causal_pca}.

In this paper, we consider causal discovery in the non-linear additive noise and pervasive confounding setting. We show that the true graph is still recoverable (\cref{thm:superdag}) 
and we provide a method which asymptotically recovers the correct partial ordering of the graph in \cref{sec:decam}. 
Our approach relies on estimating the confounding variation through a simple spectral decomposition of the observed data matrix.

The remainder of the paper is structured as follows: In \cref{sec:problem_setup}, we formalize the target of inference and the assumptions about the data generating process. We propose our method in \cref{sec:method} and present related work in \cref{sec:existing_work}. Finally, we compare our method against existing approaches on real and synthetic datasets in \cref{sec:experiments}. In \cref{A:notation}, we summarize the main notations used. \label{sec:intro}

\section{Problem Statement: Causal Discovery in the Presence of Confounding} 

\textbf{Preliminaries.} We would like to recover the causal relationships between the components of a $p$-dimensional random vector $X \sim \Pobs(X)$.
We are given $N$ datapoints $\{ x^{(n)} \}_{n=1}^N$, sampled i.i.d. from $\Pobs(X)$.
The aim is to understand causal dependencies, such as whether one variable (e.g., medical treatment) is a cause of another variable (e.g., health outcome), and how this effect is mediated by other variables (e.g., gene-expression levels).
We use a \emph{causal directed acyclic graph} (DAG) to model causal relationships, where directed edges capture cause-and-effect relationships \citep{causality_book}.
In this context, we consider a DAG $\cG$ with nodes $X_1, X_2, \ldots, X_p$, where each node $X_i$ is associated to the random variable $X_i$ for $i = 1, 2, \ldots, p$, and the distinction between random variables and nodes is made obvious from context \citep{koller2009probabilistic}.
We say $\Pobs(X)$ \textit{factorizes} according to a DAG $\cG$ if we have $\Pobs(X) = \prod_{i=1}^p \Pobs(X_i \mid X_{\Parents{\cG}{i}}),$ where $X_{\Parents{\cG}{i}}$ denotes the variables which are \textit{parents} of the variables $X_i$ in $\cG$.

Unfortunately, in the presence of unobserved confounders, there generally does not exist a DAG over observed variables that  preserves, in its factorization, all the conditional and marginal independences that we expect to hold in the margin of the full DAG over the observed and unobserved variables (i.e., DAGs are not closed under marginalization) \citep{richardson2002}. 
This factorization or \emph{Markov} property is crucial for many causal learning methods to accurately recover causal relationships \citep{chick2002, high_dim_PC, greedy_sp, additive_noise}.

To address this limitation, instead of assuming $\Pobs(X)$ factorizes according to a DAG, we assume the joint distribution over the complete set of variables $(X, \bar{H})$ is Markov with respect to a DAG $\Gnat$, where $\bar{H}$ represents a vector of $\barK \in \mathbb{N}$ unobserved variables. This assumption implies that, in the structural causal model (SCM) associated with $\Gnat$, each variable can be expressed as a function of its parents in $\Gnat$ and independent noise.\\

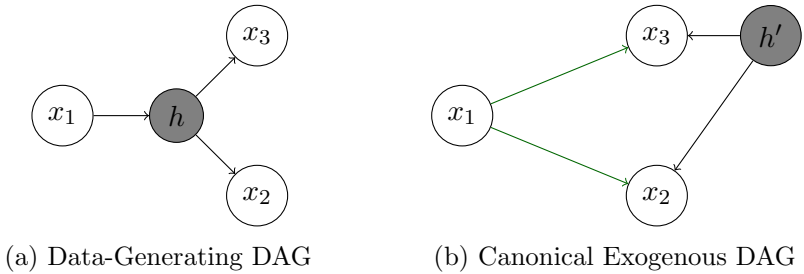
\begin{figure}
\centering
\begin{subfigure}[b]{0.4\textwidth}
\centering 
\begin{tikzpicture}[node distance={15mm}, main/.style = {draw, circle}] 
\node[main] (1) {$X_1$};
\node[main] (2) [fill=gray, right of=1] {$\bar{H}$};
\node[main] (3) [below right of=2] {$X_2$};
\node[main] (4) [above right of=2] {$X_3$}; 
\draw[->] (1) to (2); 
\draw[->] (2) to (3); 
\draw[->] (2) to (4); 
\end{tikzpicture} 
\caption{Data-Generating DAG}
\label{fig:simple_dag}
\end{subfigure}
\begin{subfigure}[b]{0.4\textwidth}
\centering 
\begin{tikzpicture}[node distance={15mm}, main/.style = {draw, circle}] 
\node[main] (1) {$X_1$};
\node[main] (3) [below right of=2] {$X_2$};
\node[main] (4) [above right of=2] {$X_3$}; 
\node[main] (2) [fill=gray, right of=4] {$H$}; 

\draw[->, line width=.75mm] (1) to (3); 
\draw[->, line width=.75mm] (1) to (4); 
\draw[->] (2) to (3); 
\draw[->] (2) to (4); 

\end{tikzpicture} 
\caption{Canonical Exogenous DAG}
\label{fig:canon_dag}
\end{subfigure}
\caption{
The right-hand figure reparameterizes the model on the left such that the unobserved variable $H$ is a source in the graph.
The bold arrows represent the DAG $\Gtrue$ corresponding to the conditional distribution $\pr(X \mid H)$.
} \label{fig:exog_repar}
\end{figure}

    \noindent 
    \textbf{Reduction to Exogenous Confounders.} Since $\bar{H}$ is unobserved, recovering $\Gnat$ from observational data alone is not possible without further assumptions about $\bar{H}$ \citep{richardson2002}.
    This work instead focuses on learning the "canonical exogenous DAG" induced by $\Gnat$ using the concepts of \textit{latent projection} and the \textit{canonical DAG} introduced in \citet{exog_dag}. 
    The canonical exogenous DAG transforms $\Gnat$ into another DAG $\cG^{\prime}$ over the same set of observed variables $X$ and a new set of unobserved variables $H$, which satisfies three properties: (1) all unobserved variables are sources (i.e., they have no parents) in $\cG^{\prime}$, (2) the partial ordering of observed variables in $\cG^{\prime}$ matches the partial ordering in $\Gnat$, and (3) there exists a joint distribution $\mathbb{Q}(X, H)$, Markov with respect to $\cG^{\prime}$, such that the marginal distribution $\mathbb{Q}(X)$ is equal to $\mathbb{P}(X)$ almost everywhere.
    In the canonical exogenous DAG, an edge between two observed nodes corresponds to the existence of a \textit{completely hidden path} in $\Gnat$, while edges from a hidden node to a set of observed nodes correspond to the existence of a \textit{hidden common cause} for those nodes in $\Gnat$.
    These structures are defined as follows:

\edit{
\begin{ndefn}
$X_i$ has a \textit{completely hidden path} to $X_j$ in $\Gnat$ if there exists a path $X_i \to \barH_{k_1} \to \ldots \to \barH_{k_m} \to X_j$ in $\Gnat$, where $1 \leq i,j \leq p$, and $1 \leq k_1, \ldots, k_m \leq \barK$.
\end{ndefn}
\begin{ndefn}
Let $X_C = \{ X_c : c \in C \}$, where $C \subseteq \{ 1, \ldots, p \}$.
$X_C$ shares a \emph{hidden common cause} in $\Gnat$ if there exists an unobserved node $\barH_j, 1 \leq j \leq \barK$, such that for each $c \in C$ there is a directed path from $\barH_j$ to $X_c$ with all vertices in $\{ \barH_k \}_{k=1}^{\barK}$.
We denote the collection of all maximal sets $C$ such that $X_C$ share a hidden common cause in $\Gnat$ by $\{ C_1, \ldots, C_K \}$ for $K \in \mathbb{N}$.
\end{ndefn}
We define the canonical exogenous DAG such that, in the terminology of \citet{exog_dag}, it is the canonical DAG associated with the latent projection of $\Gnat$ onto the observed nodes. 
We prove this equivalence in \cref{A:canon_exog_dag}.
\begin{ndefn} (canonical exogenous DAG) \label{def:can_exog_dag}
The \textit{canonical exogenous DAG} of $\Gnat$ is a DAG $\cG^{\prime}$ with: (1) an edge $X_i \to X_j$ if $X_i$ has a completely hidden path to $X_j$ in $\Gnat$, $1 \leq i,j \leq p$, and (2) a node $H_k$ and edges $H_k \to X_c$ for all $X_c \in X_{C_k}$, $1 \leq k \leq K$.
\end{ndefn}
}

By \cref{def:can_exog_dag}, the presence of an edge between two observed nodes  $X_i \to X_j$,  $1 \leq i,j \leq p$, in the canonical exogenous DAG implies that $X_i$ is a cause and not an effect of $X_j$. 
Hence, the canonical exogenous DAG preserves the partial ordering of the observed nodes in $\Gnat$. 
When $\barK=1$, the canonical exogenous DAG equals the DAG formed by first adding edges from each of the parents of $\barH_1$ to each of the children of $\barH_1$ in $\Gnat$, second removing all directed edges pointing into $\barH_1$, and third removing $\barH_1$ if it has only one child.
See \cref{fig:exog_repar} for an illustration. 
If the unobserved variables are all sources in $\Gnat$ to begin with (as assumed in \cite{causal_lrps, mult_causes}), then the canonical exogenous DAG equals $\Gnat$, with the only exception being the removal of any unobserved variables with one or zero children. \\

\noindent \textbf{Target of Inference.} \edit{\citet{exog_dag} proves that there exists a distribution $\mathbb{Q}(X, H)$ such that $\mathbb{Q}(X, H)$ is Markov with respect to the canonical exogenous DAG, and the marginal distribution $\mathbb{Q}(X)$ is equal to $\mathbb{P}(X)$ almost everywhere.
Since we cannot distinguish between the distributions $\mathbb{P}(X, \barH)$ and $\mathbb{Q}(X, H)$ from observational data, we treat $\mathbb{Q}(X, H)$ as the canonical representation of the joint distribution over observed variables and confounders. 
Hence, we make technical assumptions about $\mathbb{Q}(X, H)$ instead of $\mathbb{P}(X, \barH)$ to prove when causal recovery is possible.}

Let $\Gtrue$ denote the DAG formed by removing the nodes $\{H_k\}_{k=1}^K$ and all arrows pointing out of $\{H_k\}_{k=1}^K$ from the canonical exogenous DAG (see, for example, \cref{fig:canon_dag}). 
Since $H$ is exogenous in the canonical exogenous DAG, $\Gtrue$ is the DAG associated with the conditional distribution $\pr(X \mid H)$. 
Hence, $\Gtrue$ captures the causal relationships among the observed variables. 
Our goal is  to somehow learn $\Gtrue$ from $N$ observational datapoints $\{ x^{(n)} \}_{n=1}^N$ without observing the hidden confounders $\bar{H}$. \\

\noindent \textbf{Model.}  The SCM associated with $\Gtrue$ equals 
\begin{equation*} \label{eq:general_sem}
	\begin{split}
	X_j &= f_j(X_{\Parents{\cG}{j}}, H, \epsilon_j), \quad 1 \leq j \leq p,
	\end{split}
\end{equation*}
where \edit{$\epsilon_j$ is the noise term}, $H \sim \prod_{k=1}^K \pr(H_k)$, and $H \indep \epsilon$. 
Unfortunately, due to the curse of dimensionality, the number of datapoints required to estimate $f_j$ to some level of precision depends exponentially on the number of parents \citep{non_para_bounds}. 
Given this statistical hardness result, we assume that each $f_j$ has low-dimensional structure. 

Causal additive models (CAMs) are widely used to introduce low-dimensional structure while balancing flexibility and statistical efficiency \citep{cam}.
In a CAM, it is assumed that each node can be expressed as an additive function of its parents plus independent noise. In this paper, we adopt this assumption and consider the following model: for all $1 \leq j \leq p$,
\begin{equation} \label{eq:cam_gen_model}
	\begin{split}
	X_j 
    &= 
    \sum_{i \in \Parents{\Gtrue}{j}} f_{ij}(X_i) + \sum_{k=1}^K g_{kj}(H_k) +  \noise_j 
    \\
	&=  
    \sum_{i \in \Parents{\Gtrue}{j}} f_{ij}(X_i) + D_j + \noise_j, \quad D_j =  \sum_{k=1}^K g_{kj}(H_k),
	\end{split}
\end{equation}
where $f_{ij}$ and $g_{kj}$ are unknown functions, $f_{ij}$ does not equal zero almost everywhere, and $\epsilon_j \sim \mathcal{N}(0, \sigma^2_j)$ \raedit{for some unknown noise variances $0 < \sigma_j^2 < \infty$, $1 \leq j \leq p$}. 
Without further constraints, \cref{eq:cam_gen_model} is not identifiable. 
To ensure the uniqueness of the expansion in \cref{eq:cam_gen_model} as described in \citet{cam}, we assume $\E[f_{ij}(X_i)] = 0$ and $\E[g_{kj}(H_k)] = 0$.

Directed Gaussian graphical models are a specific instance of the SCM described in \cref{eq:cam_gen_model}, where all the unknown functions are assumed to be linear. In this particular case, \cref{eq:cam_gen_model} can be simplified to a more commonly recognized form:
\begin{equation*}
	X = B X + \Theta H + \epsilon \quad \text{s.t.} \quad  \epsilon \sim N(0, \Lambda), \  H \sim N(0, I_K),
\end{equation*}
where $B \in \R^{p\times p}$ and $\Theta \in \R^{p \times K}$ consist of \raedit{unknown} edge weights, $B$ is upper triangular, $\Lambda$ is a diagonal matrix such that $\Lambda_{jj} = \sigma^2_j$, $1 \leq j \leq p$, and $I_K$ is the $K \times K$ identity matrix. 
Solving for $X$ in the equation above,  
\begin{equation} \label{eq:linear_sem}
	X = (I_p - B)^{-1} \epsilon + (I_p - B)^{-1} \Theta H.
\end{equation}
 \label{sec:problem_setup}

\section{Our Method} \label{sec:method} 
Our objective is to learn $\Gtrue$, the DAG induced by the conditional distribution $\pr(X \mid H)$ under the model in \cref{eq:cam_gen_model}, where $\pr(X, H)$ is faithful to the canonical exogenous DAG. 
To achieve this, we propose a three-part solution. 
First, we show in \cref{sec:suff_stats} that there exists a $p$-dimensional vector of sufficient statistics $S$ of $H$ such that $\pr(X \mid S)$ factorizes according to $\Gtrue$. 
Next, in Section \ref{sec:poet_reduct}, we provide a method to accurately estimate $S$ when $H$ has an impact on many observed variables. This estimation is achieved through principal components analysis applied to the observed data.
Finally, we introduce a novel DAG score function that utilizes these estimated values of $S$. We establish in \cref{sec:decam} that this score function consistently recovers the partial order of $\Gtrue$.

Throughout the paper, we adopt the assumption, without loss of generality, that $\pi^* \coloneqq (1 \cdots p)$ represents a consistent topological ordering of $\Gtrue$ to simplify notation. This ordering ensures that variable $X_i$ can only serve as a cause, not an effect, of variable $X_j$ when $i < j$.

\subsection{Sufficient Statistics for Recovering the Exogenized DAG} \label{sec:suff_stats}

Our main theorem below, which we prove in \cref{A:proofs}, says that $\pr(X \mid S)$ is Markov relative to $\Gtrue$ when $S$ equals the conditional expectation $\E[X \mid H]$.
\begin{nthm} \label{thm:superdag}
There exist functions $\{r_j\}_{j=1}^p$ such that the SCM in \cref{eq:cam_gen_model} can be re-written as:
\begin{equation*}
\begin{split}
	X_j 
	&=  
	\sum_{i \in \Parents{\Gtrue}{j}} f_{ij}(X_i) + [S_j - r_j(S_1, \ldots, S_{j-1})] + \epsilon_j, \quad 1 \leq j \leq p,
\end{split}
\end{equation*}
where $S_j  = \E[X_j \mid H].$ 
\end{nthm}
Since $\epsilon_j \indep \{(X_i, S_i)\}_{i=1}^{j-1} \cup \{ S_j \}$, the SCM in \cref{thm:superdag} is an additive noise model. In the context of additive noise models, \citet{additive_noise} discuss settings when $\Gtrue$ is identifiable, i.e., settings for which all edges in the graph can be learned and oriented correctly with an infinite amount of observational data.
They demonstrate cases where $\Gtrue$ is fully identifiable,  and situations where $\Gtrue$ is identifiable only up to its Markov equivalence class. An example of such limited identifiability occurs when the $f_{ij}$ functions are linear and the errors follow a Gaussian distribution.  
For the sake of clarity, we highlight a specific instance in which $\Gtrue$ is fully identified, which can be directly derived from Corollary 31 of \citet{additive_noise}. However, we refer readers to \citet{additive_noise} for a comprehensive exploration of identifiability results in additive noise models.
\begin{ncor} \label{cor:ident} $\Gtrue$ is identifiable from the conditional distribution $\pr(X \mid S)$ when all functions in $\{ f_{ij} \}_{i \in \Parents{\Gtrue}{j}, j \in [p]}$ are non-linear and three-times differentiable, and when the errors $\{\epsilon_j\}_{j=1}^p$ are Gaussian.
\end{ncor}

Based on the above discussion, identifying $\Gtrue$ (or recovering its Markov equivalence class) requires computing or estimating $S$. 
Since we do not know $H$, however, we cannot estimate $S$ by regressing $X$ on $H$.
In the next section, we show how to estimate $S$ from just the observed data. 
We conclude this subsection by describing the special linear case to build intuition for \cref{thm:superdag}.

\begin{nexa} \label{ex:linear_s_form}
In the linear setting of \cref{eq:linear_sem}, 
	\begin{equation*}
		\begin{split}
			S &= \E[X \mid H] \\
				&= \E[(I - B)^{-1} \epsilon + (I - B)^{-1} \Theta H \mid H] \\
				&= (I - B)^{-1} \Theta H, \quad \text{since $\epsilon \indep H$ and $\E[\epsilon] = 0$}.
		\end{split}
	\end{equation*}
\cref{lem:linear_s} below, which we prove in \cref{A:proofs}, allows us \edit{to write $r_j$ in \cref{thm:superdag} analytically.}
\begin{nlem} \label{lem:linear_s} In the linear setting,
\begin{equation} \label{eq:linear_preproc}
	(X_j - S_j) 
	= 
	\sum_{i \in \Parents{\Gtrue}{j}} B_{ij} (X_i - S_i) + \epsilon_j, \quad 1 \leq j \leq p,
\end{equation}
where $B_{ij}$ denotes the $(i,j)$-th element of the matrix $B$.
\end{nlem}
Suppose $\{ (x^{(n)}, h^{(n)}) \}_{n=1}^N$ are sampled i.i.d. from $\pr(X, H)$, but a data analyst observes only $\{ x^{(n)} \}_{n=1}^N$.
Let $s^{(n)} = \E[X \mid H = h^{(n)}]$ for $1 \leq n \leq N$.
If we had $\{ s^{(n)} \}_{n=1}^N$, then \cref{eq:linear_preproc} suggests a natural data pre-processing algorithm: give the analyst a processed dataset $\{(x^{(n)} - s^{(n)})\}_{n=1}^N$. 
The analyst can then feed the processed dataset into a linear causal learning algorithm and \raj{estimate} $\Gtrue$ up to its Markov equivalence class.
As we will see in \cref{sec:decam}, even in the nonlinear setting, the values $\{ s^{(n)} \}_{n=1}^N$ play a fundamental role in computing the score function that we use to learn $\Gtrue$.
\end{nexa}

Motivated by these observations, in the next section, we provide a setting under which $\{s^{(n)}\}_{n=1}^N$ can be accurately estimated.

\subsection{Asymptotically Exact Estimates of the Sufficient Statistics} \label{sec:poet_reduct}
To estimate $\{ s^{(n)} \}_{n=1}^N$, we reduce our problem into the linear latent factor model proposed in \citet{poet}. \\ 

\noindent \textbf{Reduction to \citet{poet}.} Let $\mu^*$ denote the joint distribution of $(\epsilon, H)$. 
Assume that each random variable $X_j \in \hilb^* \subset \L^2(\mu^*)$, where $\hilb^*$ is a complete Hilbert space and  $ \L^2(\mu^*)$ consists of all square-integrable functions of $(\epsilon, H)$ with respect to the measure $\mu^*$. 
Then, there exists a set of $M$ basis functions $\{\phi_m\}_{m=1}^M$, $1 \leq M \leq \infty$, ~such~that $S_j = \E[X_j \mid H] \in \hilb_M  \coloneqq \text{span}\{\phi_1, \ldots, \phi_M \}$ since $\hilb^*$ is separable \citep{rudin1974functional}. 
Hence, there exists a vector of basis coefficients $\psi_j \in \R^M$ such that $S_j = \psi_j^T \Phi(H)$ for $1 \leq j \leq p$, where $\Phi(h) \coloneqq [\phi_1(h), \ldots, \phi_M(h)]^T \in \R^M$ and $h \in \R^K$.
Through this basis expansion, we can re-express the observed random vector $X$ as,
\begin{equation} \label{eq:poet_reduction}
	\begin{split}
		X &= \E[X \mid H] + [X - \E[S \mid H]] \\
		&= S + [X - \E[X \mid H]] \\
		&= \Psi \Phi(H) + U,
	\end{split}
\end{equation}
where $\Psi \in \R^{p \times M}$, the $j$th row of $\Psi$ equals $\psi_j^T$, $1 \leq j \leq p$, and $U = X -  \Psi \Phi(H)$. 
By properties of conditional expectation, $U$ belongs to the orthogonal complement of $\hilb_M$. 
As a result, the covariance between each component of $\Psi \Phi(H)$ and $U$ is zero so \cref{eq:poet_reduction} is a special instance of the linear latent factor model considered in \citet{poet}. 
Hence, we use the general purpose procedure developed in \citet{poet} to estimate $S = \Psi \Phi(H)$ below. \\

\noindent \textbf{Derivation of the Estimator.} Let $\bbX \in \R^{N \times p}$ denote the data matrix, where the $n$th row of $\bbX$ equals $[x^{(n)}]^T$.
By \cref{eq:poet_reduction}, we can decompose $\bbX$ as,
\begin{equation*} \label{eq:mat_decomp}
	\begin{split}
	& \bbX  = \bbS + \bbU \quad \text{s.t.} \quad \bbX = [x^{(1)} \cdots x^{(N)}]^T, \quad \bbS = [s^{(1)} \cdots s^{(N)}]^T, \\
	& \qquad \qquad \qquad \qquad \ \bbU = [u^{(1)} \cdots u^{(N)}]^T,
	\end{split}
\end{equation*}
where each $(x^{(n)}, h^{(n)})$ is sampled i.i.d. from $\pr(X, H)$, $s^{(n)} =  \Psi \Phi(h^{(n)})$, and $u^{(n)} = x^{(n)} - \Psi \Phi(h^{(n)})$  for $1 \leq n \leq  N$. 
To estimate $\bbS$ from only the observed data matrix $\bbX$, we use principal components analysis (PCA) as in \citet{poet}.
In particular, let $V \in \R^{p \times p}$ and $\Lambda \in \R^{p \times p}$ consist of eigenvectors and eigenvalues (in decreasing order) from the eigendecomposition of the sample covariance matrix $\hat{\Sigma}$, respectively. 
For $J \in \mathbb{N}$, let $V_J \in \R^{p\times J}$ denote the matrix consisting of the $J$ eigenvectors in $V$ with largest eigenvalues.
Then we estimate $\bbS$ by
\begin{equation}  \label{eq:pca_estimator}
	 \hat{\bbS} = (V_J V_J^T \bbX^T)^T \quad \text{s.t.} \quad \hat{\Sigma} := \frac{1}{N} \bbX^T \bbX = V \Lambda V^T,
\end{equation}
Below we prove when \cref{eq:pca_estimator} yields an accurate estimate of $\bbS$, and describe how to pick $J$. \\

\noindent \textbf{Theoretical Guarantees.} In order for \cref{eq:pca_estimator} to estimate $\bbS$ well, we need the singular values of $\bbS$ to be much larger than that of $\bbU$. 
If the singular values of $\bbS$ dominate, then most of the low-dimensional structure in the data matrix $\bbX$ is captured by $\bbS$. 
Hence, we might expect that the PCA estimator in \cref{eq:pca_estimator} will output an estimate close to $\bbS$. 
Below we make this basic intuition precise.

\begin{nassum} \label{assum:finite_dime} \edit{There exists a constant $M < \infty$ such that} for all $1 \leq j \leq p$, $S_j  \in \hilb_M$.
\end{nassum}

\begin{nassum} \label{assum:div_evals} 
Let $\lambda_{\min}(A)$ and $\lambda_{\max}(A)$ denote the smallest and largest eigenvalues of a matrix $A$, respectively.
We assume that there exists constants $0 < K_1, K_2 < \infty$ that do not depend on $p$ such that $\lambda_{\min} \left( \frac{1}{p}  \Psi^T \Psi \right) > K_1$ and $\lambda_{\max} \left(  \frac{1}{p}  \Psi^T \Psi \right) < K_2$.
Furthermore, $\cov(\Phi(H), \Phi(H))$ is non-singular.
\end{nassum}
\begin{nassum} \label{asum:3.2ii} (Assumption 2(b) from \citet{poet})
There exists constants $c_1$ and $c_2$ such that $\lambda_{\min}(\Sigma_u) > c_1$, $\|\Sigma_u\|_1 < c_2$, and $\min_{1 \leq i, j \leq p} \cov(U_i, U_j) > c_1$, where $\Sigma_u = \E[U U^T]$.
\end{nassum}

\noindent If $p$ is greater than $M$, the number of basis in coefficients in \cref{eq:poet_reduction}, then the rank of $\Sigma_s = \E[S S^T]$ is less than or equal to $M$. 
Hence, \cref{assum:finite_dime} ensures that $\bbS$ is low-rank as $p \rightarrow \infty$. 
Since $\lambda_{\min}(\Psi^T \Psi) \geq K_1 p$ and $\cov(\Phi(H), \Phi(H))$ is non-singular by \cref{assum:div_evals}, together this implies that the top $M$ singular values of $\bbS$ grow at rate $O(p)$, and the remaining $p-M$ singular values of $\bbS$ equal 0.
Hence, the spectrum of $\bbS$ is "spiked" in the sense that there are only a small number of very large singular values.
Graphically, this property of $\bbS$ implies that each component $H_k$, $1 \leq k \leq K$, must have an effect on a non-vanishing number of observed variables as $p \to \infty$, which we call pervasive confounding.
In particular, since all diagonal entries of a positive-definite matrix must be positive, \cref{assum:div_evals} implies that,
 \begin{equation*}
   \lim_{p \rightarrow \infty} \sum_{j=1}^p \psi_{jm}^2 \rightarrow \infty,
   \quad\quad
   \textrm{for all}~1 \leq m \leq M
 \end{equation*}
with probability 1. 
Hence, each $\phi_m(H)$, $1 \leq m \leq M$, must have an effect on a non-vanishing number of observed variables.
Thus each $H_k$, $1 \leq k \leq K$, is a pervasive confounder, since there exists at least one $1 \leq m' \leq M$ such that $\phi_{m'}$ is a non-constant function of $H_k$ by construction of the canonical exogenous DAG.

The final assumption (\cref{asum:3.2ii}) ensures that $\Sigma_s$ dominates the spectrum of $\E[X X^T]$ as $p \rightarrow \infty$. Specifically, by properties of matrix norms, $ \lambda_{\max}(\Sigma_u) \leq \sqrt{p} \|\Sigma_u\|_1 <  c_2 \sqrt{p}$. 
Hence, all the eigenvalues of $\Sigma_u$ grow at a slower rate than that of $\Sigma_s$.
The remaining assumptions needed to prove the convergence of $\hat{S}$ are technical regularity assumptions from \citet{poet} which we defer to \cref{A:proof_spect_deconfound}.
\begin{nthm} \label{thm:pcss_convg} \raj{Under \cref{assum:finite_dime}, \cref{assum:div_evals}, \raedit{\cref{asum:3.2ii}}, and the technical regularity assumptions provided in \cref{A:proof_spect_deconfound}} \edit{(\cref{asum:3.2iii} \raedit{and} \cref{assum:3.4})},
\begin{equation*}
	\max_{1 \leq j \leq p, 1 \leq n \leq N} 
	\left\| \hat{s}_j^{(n)} - s_j^{(n)} \right\|_2 
	= 
	O_p \left(\log(N)^{1/c} \sqrt{\frac{\log p}{N}} + \frac{N^{\frac{1}{4}}}{\sqrt{p}} \right),
\end{equation*}
for some constant $c > 0$, where for $1 \leq n \leq N$, $\hat{s}_j^{(n)}$ and $s_j^{(n)}$ denote the $n$th rows of matrices $\hat{\bbS}$ and $\bbS$, respectively, and $J=M$ in \cref{eq:pca_estimator}.
\end{nthm}
We prove \cref{thm:pcss_convg} in \cref{A:proof_spect_deconfound}, which follows almost directly from Corollary 1 of \citet{poet}.
By \cref{thm:pcss_convg}, we can estimate $s_j^{(n)}$ at the rate $\tilde{O}\left(\frac{1}{\sqrt{N}} +\frac{N^\frac{1}{4}}{\sqrt{p}} \right)$ up to log factors for root mean-squared error loss. Therefore, we need both $p, N \rightarrow \infty$ and $N = o(p^2)$ for consistent estimation of $s_j^{(n)}$. 

\cref{algo:pcss_est} summarizes our approach for estimating the matrix of sufficient statistics $\bbS$. 
In practice, we might not know what value of $J$ to use in \cref{algo:pcss_est}. 
To this end, we can use the estimator provided in Section 2.4 of \citet{poet} to estimate $J$. 
As we show in \cref{sec:experiments}, it is often easier to instead visually inspect the spectrum of the sample covariance matrix (referred to as the "scree plot") to pick $J$. 

In the next section, we propose a new score function and prove that it can recover a correct ordering of $\Gtrue$ in \cref{thm:decam_consis}. 
We discuss the extension of our consistency result to the $M=\infty$ setting in \cref{A:m_inf_behav}.

\subsection{The DeCAMFounder Score Function} \label{sec:decam}

We conclude by combining our results so far to derive an estimator for $\Gtrue$. We prove that our approach can recover a correct partial ordering of $\Gtrue$ in the high-dimensional setting when $p, N \rightarrow \infty$. To prove this result, we extend the main theorem in \citet{cam} to the setting when there are pervasive confounders. \\

\noindent \textbf{Derivation of the DeCAMFounder Score Function.} We propose to estimate $\Gtrue$ via a \emph{score-based} approach.
In general, score-based approaches consist of two parts: (1) a score-function to measure the quality of a DAG based on the observed data $\bbX$, and (2) a combinatorial optimization procedure to optimize this score-function over the space of DAGs. 
Since there are many existing general-purpose procedures for (2), such as greedy equivalence search, we focus our attention on (1).

A natural starting point (based on its popularity in the linear setting) might be to use the \emph{Bayesian information criterion} (BIC) as the score function, which penalizes DAGs based on the number of parameters \citep{chick2002}. 
For linear models, computing this penalty is straightforward because the number of parameters simply equals the number of edges in the graph (corresponding to the number of edge weights to be estimated) and $p$ (the number of node noise variances to be estimated). 
For \edit{semi-parametric} models, however, the number of parameters can be infinite.  
Hence, BIC cannot directly be applied. 
Instead, we use a Bayesian score, namely the marginal log-likelihood of the DAG (which the BIC approximates with a second-order Taylor series) conditional on the matrix of confounder sufficient statistics $\bbS$.
Specifically, let $\pr(\Omega_\cG \mid \cG)$ be a prior over the unknown set of functions $\Omega_\cG \coloneqq \{ \{f_{ij}^\cG \}_{i \in \Parents{\cG}{j}} , r_j^\cG \}_{j=1}^p$ that define the model from \cref{thm:superdag}.
Assume that $\pr(\Omega_\cG \mid \cG) = \prod_{j=1}^p \bbQ(\{ f_{ij}^\cG \}_{i \in \Parents{\cG}{j}}, r_j^\cG \mid \cG)$, also known as \emph{parameter modularity} \citep{gp_nets}.
Given this prior, we score a DAG $\cG$ by computing
\begin{equation} \label{eq:gp_marg_like}
	\begin{split}
	\max_{\sigma^2 \in \mathbb{R}^p_+} \log \pr(\bbX \mid \cG, \bbS, \sigma^2) 
	&= 
	\max_{\sigma^2 \in \mathbb{R}^p_+} \log \int_{\Omega_\cG} \pr(\bbX \mid \cG, \bbS, \Omega_\cG, \sigma^2) d\pr(\Omega_\cG \mid \cG) 
	\\
	&= 
	\sum_{j=1}^p  \max_{\sigma^2_j \in \mathbb{R}_+} \log \int_{\Omega_\cG} \pr(\bbX_{:, j} \mid \cG, \bbS, \Omega_\cG, \sigma_j^2) d\pr(\Omega_\cG \mid \cG)
	\end{split}
\end{equation}
where $\bbX_{:, j} \in \mathbb{R}^N$ denotes the $j$th column of matrix $\bbX$ and $\sigma^2$ denotes a $p$-dimensional vector of unknown noise variances. 
In the following, we consider the special setting where the priors over the unknown functions in $\Omega_\cG$ are given by  Gaussian processes. 
This choice on prior allows us to model non-linear relationships, penalize for model complexity, and compute \cref{eq:gp_marg_like} analytically. 

\begin{nprop} \label{eq:gp_gen_formula}
Let $\cG$ denote a DAG, and $\pi$ a topological ordering of $\cG$. 
Let $GP(0, k)$ denote a zero-mean Gaussian process with positive definite kernel function $k(\cdot, \cdot)$.
Assume the following priors on the unknown functions in $\Omega_\cG$:
\begin{equation*}
\begin{split}
	\ f_{ij}^\cG(x_i)
	&\sim 
	GP(0, k_{ij}) 
	\quad \forall~i \in \Parents{\cG}{j}, 1 \leq j \leq p,
	\\
	r_j^\cG(s_{\pi_1}, \cdots, s_{\pi_{j-1}}) &\sim GP(0, k_{j}) \quad \forall~1 \leq j \leq p,
\end{split}
\end{equation*}
where $k_{ij}(\cdot, \cdot)$ and $k_{j}(\cdot, \cdot)$ are positive definite kernel functions from $\mathbb{R} \times \mathbb{R} \mapsto \mathbb{R}$ and $\mathbb{R}^{j-1} \times \mathbb{R}^{j-1} \mapsto \mathbb{R}$, respectively. 
Then, $\log \int_{\Omega_\cG} \pr(\bbX_{:, j} \mid \cG, \bbS, \Omega_\cG, \sigma_j^2) d\pr(\Omega_\cG \mid \cG)$ equals
\begin{equation} \label{eq:marg_like}
 -\frac{1}{2} 
 \left( \tilde{\bbX}_{:,j}^T L_{j}^{-1} \tilde{\bbX}_{:,j} - \log\det L_{j} - N\log(2\pi)
 \right), 
 \quad 
 \tilde{\bbX}_{:,j}  \coloneqq \bbX_{:,j} - \bbS_{:,j},
\end{equation}
where $L_j =  K_{j} + \sigma_j^2 I_{N}$, and the $(n, m)$th entry of the $N \times N$ kernel matrix $K_j$  equals
\begin{equation*}
    \sum_{i \in \Parents{\cG}{j}} k_{ij}\left(\tilde{x}_i^{(n)}, \tilde{x}_i^{(m)}\right) + k_{j}\left(\left(s_{\pi_1}^{(n)}, \cdots, s_{\pi_{j-1}}^{(n)}\right), \left(s_{\pi_1}^{(m)}, \cdots, s_{\pi_{j-1}}^{(m)}\right)\right), \ 1 \leq n, m \leq N.
\end{equation*}
\end{nprop}
We prove \cref{eq:gp_gen_formula} in \cref{A:proofs}. 
Below we bound the computational complexity of computing \cref{eq:gp_marg_like}.
\begin{ncor} \label{eq:decam_runtime}
Under the assumptions in \cref{eq:gp_gen_formula}, $\pr(\bbX \mid \cG, \bbS, \sigma^2)$ takes $O(p^2 N^2 + pN^3)$ time to compute.
\end{ncor}
\begin{proof}
It suffices to show that $\int_{\Omega_\cG} \pr(\bbX_{:, j} \mid \cG, \bbS, \Omega_\cG, \sigma_j^2) d\pr(\Omega_\cG \mid \cG)$ takes $O(pN^2 + N^3)$ time to compute by \cref{eq:gp_gen_formula}. 
$k_{ij}$ and $k_{j}$ take $O(1)$ and $O(p)$ time to evaluate on a pair of datapoints, respectively. 
Hence, the $N \times N$ kernel matrix $K_j$ takes at most $O(p N^2)$ time to compute. 
Computing the determinant and inverse of $L_j$ takes $O(N^3)$ time. 
Finally, the matrix vector multiplication $\tilde{\bbX}_{:,j}^T L_{j}^{-1} \tilde{\bbX}_{:,j}$ takes $O(N^2)$ time.  
\end{proof}

\begin{algorithm}[]
    \caption{Principal Confounding Sufficient Statistics}
    \label{algo:pcss_est}
    \begin{algorithmic}[1] 
        \Procedure{pcss}{$\bbX, J$}
        	\State $\hat{\Sigma} = \frac{1}{N} \bbX^T \bbX$
        	\State $V, \Lambda, V^T  = \text{eigen}(\hat{\Sigma} )$  \Comment{Eigendecomposition of $\hat{\Sigma}$}
        	\State $\bbS = (V_J V_J^T \bbX^T)^T \in \R^{N \times p}$ \Comment{Sufficient statistics derived in \cref{eq:pca_estimator}}        	
            \State \textbf{return} $\bbS$
        \EndProcedure
    \end{algorithmic}
\end{algorithm}
\begin{algorithm}[]
    \caption{The DeCAMFounder Score}
    \label{algo:decam_dag_score}
    \begin{algorithmic}[1] 
        \Procedure{decam}{$\bbX, \cG, J$} \Comment{$\cG$ a DAG with vertex set $[p]$}
        	\State $\bbS = \text{PCSS}(\bbX, J)$ \Comment{Computed from \cref{algo:pcss_est} } 
        	\State $\pi$ is a topological order of $\cG$
        	\State $mll = 0$
        	\For{$1 \leq j \leq p$}
			\State $mll = mll + \max_{\sigma_j^2}$[ \cref{eq:marg_like} ] \Comment{Maximize via gradient ascent}
		\EndFor      	

            \State \textbf{return} $mll$\Comment{The log marginal likelihood of $G$.}
        \EndProcedure
    \end{algorithmic}
\end{algorithm}
We summarize our final derived score function in \cref{algo:decam_dag_score}. In our implementation of the DeCAMFounder score, we use \emph{GPyTorch} \citep{black_box} to fit the noise variances and kernels via gradient ascent. 
In \cref{A:kernel_details}, we detail the particular kernel functions used in our experiments. 
We now turn our attention to proving when maximizing the DeCAMFounder score recovers a consistent ordering of $\Gtrue$. \\

\noindent \textbf{Theoretical Guarantees.} Assume that $\Gtrue \in \dagset$ for some set of DAGs $\dagset$ with vertices $\{X_j\}_{j=1}^p$.
Let $\Pi_* \subset \dagset$ consist of all DAGs with partial orderings consistent with that of $\Gtrue$, namely $\Pi^* = \{\cG \in \dagset: i \in \Ancestors{\Gtrue}{j} \implies i \in \Ancestors{\cG}{j}   \}$, where $\Ancestors{\cG}{j}$ denotes the set of ancestors of $X_j$ in $\cG$. 
To prove that an estimator recovers a true ordering of $\Gtrue$, it suffices to show that it selects a DAG in  $\Pi_*$ with probability tending towards 1 as $N \rightarrow \infty$. 
In \cite{cam}, the authors proved consistency of the MLE for causal additive models. 
\cref{thm:decam_consis} extends the main result in \citet{cam} to the more general setting when there are pervasive confounders. To prove this result we need to make a number of different assumptions, which we now discuss further.
\begin{nassum} \label{assum:decompose}
For all $1 \leq j \leq p$, there exists some function $q_j$ such that $\raedit{r_j(S_1, \ldots, S_{j-1})} = q_j(S_{\Parents{\Gtrue}{j}})$ in \cref{thm:superdag}, \raedit{where $S_{\Parents{\Gtrue}{j}}$ restricts the components of $S$ to the set of indices that belong to the parent set of $X_j$ in $\Gtrue$.}
\end{nassum}
In general, $r_{j}(S_1, \ldots, S_{j-1})$ depends on $j-1$ variables, which makes estimation challenging when $j$ is large.
\cref{assum:decompose} ensures that each $r_{j}$ depends on a sparse number of components in $S$ whenever the true DAG $\Gtrue$ is sparse, making estimation significantly easier.
When $f_{ij}$ is a linear function, then \cref{assum:decompose} always holds, see \cref{eq:linear_confound_parents}.
With slight abuse of notation, in the remainder of this section we use $r_j(S_{\Parents{\cG}{j}})$ to denote $q_j(S_{\Parents{\cG}{j}})$.

The next assumption concerns sufficient "model separation", measured in terms of the negative log-likelihood scores of different DAGs. 
By \cref{thm:superdag}, Gaussianity of the errors, and \cref{assum:decompose}, the expected negative log-likelihood score of a DAG $\cG$ equals $\sum_{j=1}^p \log(\sigma_j^{\cG}) + p\sqrt{2\pi} + .5p$ \citep{cam}, where 
\begin{equation} \label{eq:proj_func}
    \begin{split}
    & \left( \sigma_j^{\cG} \right)^2 \coloneqq \E_{\pr} \left[ \left(X_j  - \sum_{i \in \Parents{\cG}{j}} f_{ij}^{\cG}(X_i) - S_j + r_j^{\cG}(S_{\Parents{\cG}{j}}) \right)^2 \right] 
    \\
    \left\{ f_{ij}^\cG \right\}_{i \in \Parents{\cG}{j}}, r_j^{\cG} 
    = 
    &\argmin_{g_{ij} \in \fkspace_{ij}, \ w_j \in \rkspace_j^\cG}  \E_{\pr} \left[ \left( X_j - \sum_{i \in \Parents{\cG}{j}} g_{ij}(X_i) - S_j - w_j( S_{\Parents{\cG}{j}}) \right)^2 \right]
    \end{split}
\end{equation}
and $\fkspace_{ij}$ and $\rkspace_j^\cG$ denote the reproducing kernel Hilbert spaces induced by the kernels $k_{ij}$ and $k_{j}$ in \cref{eq:gp_gen_formula} for $1 \leq j \leq p$, respectively. 
We define the \emph{model-separation} parameter $\xi_p$ as the difference in the normalized negative log-likelihood score of the best incorrect DAG and correct DAG:
\begin{equation*}
    \xi_p = p^{-1} \min_{\cG \notin \Pi_*} \sum_{j=1}^p 
	\left( \log \left( \sigma_j^{\cG} \right) - \log \left( \sigma_j^{\Gtrue} \right) \right).
\end{equation*}
If $\cG \not \in \Pi_*$, $\fkspace_{ij}$ and $\rkspace_j^{\Gtrue}$ contain the true unknown functions generating the SCM in \cref{thm:superdag} for \raedit{$1 \leq j \leq p$}, and the requirements in \cref{cor:ident} are satisfied, then $\xi_p > 0$.
This non-zero gap occurs because the function class is restricted to additive functions of the observed variables; regression on the wrong order induces a different density than $\pr(X \mid S)$.
Hence, the KL divergence is negative which implies $\xi_p > 0$. 
Consequently, the partial ordering is identifiable.

\edit{To prove \cref{thm:decam_consis}, we require that $\xi_p$ is large enough to offset the estimation error of having a finite number of datapoints, and noisy estimates of $\{s^{(n)}\}_{n=1}^N$.
\begin{nassum} \label{assum:gap}
(model-separation gap) $\xi_p > 0$, and for the constant $c$ in \cref{thm:pcss_convg},
\begin{equation*}
     \left(\log(N)^{1/c} \sqrt{\frac{\log p}{N}} + \frac{N^{\frac{1}{4}}}{\sqrt{p}} \right) = o(\xi_p)
\end{equation*}
\end{nassum}
In \citet{cam}, the authors require that $\xi_p$ grows faster than $\sqrt{\frac{\log p}{N}}$ as $p, N \rightarrow \infty$. 
We require a stronger gap due to the additional  error from estimating $\{s^{(n)}\}_{n=1}^N$. 
When $p = \Omega(N^{3/2})$, the model-separation gap requirement is $O\left(\log(N)^{1/c} \sqrt{\frac{\log p}{N}}\right)$, and only differs from the requirement in \citet{cam} by a log factor}.

The remaining assumptions used to prove \cref{thm:decam_consis} closely resemble those in \citet{cam}, where we make minor modifications to handle the $r_j$ component in \cref{thm:superdag}. 
As in \citet{cam}, we assume that $\Gtrue$ is sparse, and that the function classes $\fkspace_{ij}$ and $\rkspace_j^\cG$ are finite-dimensional with bounded log-bracketing numbers; see \cref{assum:sparsity} through \cref{assum:var_pos} in \cref{A:proof_consis}. 
Our proof of \cref{thm:decam_consis}, which shows that the DeCAMFounder asymptotically recovers the true partial ordering of $\Gtrue$, is provided in \cref{A:proof_consis}.
\edit{\begin{nthm} \label{thm:decam_consis}
Let $\widehat{\cG}$ be the DAG that maximizes the score in \cref{algo:decam_dag_score} over $\dagset$. 
Suppose that the assumptions in \cref{thm:pcss_convg} hold, and additionally that \cref{assum:decompose}, and \cref{assum:sparsity} through  \cref{assum:var_pos} in \cref{A:proof_consis} hold. If $N = o(p^2)$ and $p = o(\exp(N))$, then  
\begin{equation*}
    \mathbb{P}(\widehat{\cG} \in \Pi_*) \rightarrow 1 \quad \text{as} \quad p,N \rightarrow \infty.
\end{equation*}
\end{nthm}
The rate constraint in \cref{thm:decam_consis} is satisfied, for example, when $p = \Theta(N)$. In \cref{A:m_inf_behav}, we extend \cref{thm:decam_consis} to handle the case when  \cref{assum:finite_dime} does not hold by using a similar idea in \citet{cam} to handle model misspecification.}

\section{\raj{Related Work}} 

Two main approaches have been proposed for learning causal structures in the presence of confounders. The first approach focuses on modeling the conditional independence structure of $\pr(X)$ using an ancestral graph and subsequently recovering the edges in this graph. Examples of algorithms for ancestral graph estimation include FCI \citep{causality_book}, RFCI \citep{fast_fci}, and GSPo \citep{gspo}. However, these methods recover fewer causal relationships since they make minimal assumptions about the confounders, as illustrated in Figure 1 of \citet{causal_lrps}.

The second approach makes more assumptions about the confounders, and tries to exploit these structural assumptions for better estimation. 
One commonly used assumption is that the confounders are sources in the graph ("exogenous") and have an effect on many observed variables ("pervasive"). 
For linear Gaussian graphical models with exogenous pervasive confounders, several methods have been proposed to accurately estimate the covariance matrix of the conditional distribution $\pr(X \mid H)$ using only data points drawn from $\pr(X)$ for the purposes of conditional independence testing \citep{lrps_graph, poet, causal_lrps, causal_pca}. 
The methods for estimating the covariance matrix associated with $\pr(X \mid H)$, which in the linear Gaussian case is the same for all $H$, fall into two classes: the \emph{optimization approach} and the \emph{spectral approach}. \\

\noindent \textit{Optimization Approach: Deconfounding the Precision Matrix.} \raj{\citet{causal_lrps} decompose} the precision matrix of $(X, H)$ into the blocks
\begin{equation*}
[\cov((X, H))]^{-1} = \begin{pmatrix}
  \begin{matrix}
  J_O
  \end{matrix}
  & \rvline & J_{OH} \\
\hline
  J_{HO} & \rvline &
  \begin{matrix}
  J_H
  \end{matrix}
\end{pmatrix}. 
\end{equation*}
By standard theory for multivariate Gaussian distributions, $J_O \in \R^{p\times p}$ is the precision matrix associated with the conditional distribution $\pr(X \mid H)$. 
The key idea for estimating $J_O$ is to relate the observed precision matrix $[\cov(X)]^{-1} $ and the target quantity $J_O$ via Schur complements:
\begin{equation*}
[\cov(X)]^{-1} = J_O - J_{OH} J_{HH}^{-1} J_{HO}.
\end{equation*}
For sparse DAGs,  $J_O$ is a sparse matrix, and when $K \ll p$, $J_{OH} J_{HH}^{-1} J_{HO}$ is low-rank. 
Hence, when the confounders are pervasive, it is possible to estimate each component through a low-rank plus sparse matrix decomposition of $[\cov(X)]^{-1}$ \citep{causal_lrps, lrps, lrps_graph}. \\

\noindent \textit{Spectral Approach: Deconfounding the Covariance Matrix.} In contrast to the optimization approach, \raj{\citet{causal_pca}} directly estimates the covariance matrix. Under the SCM in \cref{eq:linear_sem}, 
\begin{equation} \label{eq:poet}
	\begin{split}
	\cov(X) &= (I - B)^{-1} \Lambda (I - B)^{-T} + (I - B)^{-1}\Theta \Theta^T (I - B)^{-T} \\
		&= (I - B)^{-1} \Lambda (I - B)^{-T} + \tilde{\Theta} \tilde{\Theta}^T \quad \text{s.t.} \quad \tilde{\Theta} =  (I - B)^{-1}\Theta \\
		&= \underbrace{ J_O^{-1}}_{\text{sparse inverse}} + \underbrace{\tilde{\Theta} \tilde{\Theta}^T}_{\text{low-rank}}.
	\end{split}
\end{equation}
If the confounders are pervasive, then the eigenvalues of $\tilde{\Theta} \tilde{\Theta}^T$ grow linearly with the dimension $p$, while the eigenvalues of $J_O^{-1}$ are assumed to grow much slower.
As a result, the spectrum of $\cov(X)$ is dominated by the spectrum of $\tilde{\Theta} \tilde{\Theta}^T$.
Hence, the top principal components of $\cov(X)$ approximate $\tilde{\Theta} \tilde{\Theta}^T$ well, and consequently provide a way to estimate $J_O^{-1}$ \citep{poet, spoet}. 
In settings when the gap between the eigenvalues of $J_{O}^{-1}$ and $\tilde{\Theta} \tilde{\Theta}^T$ in \cref{eq:poet} is small, \citet{causal_pca} propose a modified covariance estimator.
Outside of the graphical models literature, such spectral approaches are similar to the approach proposed by \citet{mult_causes} to estimate average treatment effects in the potential outcomes framework.
 \label{sec:existing_work}

\section{Experiments} \label{sec:experiments} 

We begin our empirical evaluation by examining the performance of the \emph{Principal Confounding Sufficient Statistics} (PCSS) procedure, outlined in \cref{algo:pcss_est}, in estimating the confounder sufficient statistics $S = \E[X \mid H]$. This analysis is presented in \cref{sec:exp_pcss_est}. We investigate the impact of estimation errors in PCSS on the performance of our DeCAMFounder score function, which is discussed in \cref{sec:decam_perform} and \cref{sec:real_data_decam}. To provide a comprehensive comparison, we benchmark our DeCAMFounder score function against other commonly used score functions in the fully-observed or pervasive confounding settings. Detailed descriptions of these benchmark methods are provided in \cref{sec:decam_perform}.

To establish a known ground truth DAG, we conduct our evaluation using simulated data. This allows us to assess the accuracy of our method. In addition to the simulated data, we evaluate our approach on a real ovarian cancer dataset in \cref{sec:real_data}. The use of this dataset enables partial validation based on prior biological knowledge of the underlying system. We express our gratitude to the authors of \citet{causal_lrps} for providing a pre-processed and easily reproducible version of this dataset. All the results presented in this section can be reproduced using the provided data and code, which are accessible at \texttt{https://github.com/uhlerlab/decamfound}.

\subsection{Simulated Data}
In this set of simulations, we examine the robustness of our method to different data characteristics. To achieve this, we vary the following parameters: the strength of confounding, linear or non-linear SCM, and dimensionality. \\

\noindent{\textbf{Controlling the Strength of Confounding.}}
Generating non-linear (and even  linear) data can be challenging. 
For example, without proper normalization, the variance of downstream nodes can explode (e.g., consider a line graph $X_1 \rightarrow X_2 \cdots \rightarrow X_p$ with edge weights larger than 1). 
Additionally, if the noise variance is the same across nodes, downstream nodes typically exhibit a higher signal-to-noise ratio than upstream nodes. To mitigate any biases introduced by such artifacts in the data simulation process, we apply the following normalization (note that specific edge cases arise for source nodes or those independent of $H$; refer to \cref{A:sim_data_generate} for further details):
\begin{itemize}
	\item \textbf{Unit variance nodes}: $\Var(X_j) = 1$ for all $1 \leq j \leq p$.
	\item \textbf{Zero mean nodes}: $\E[X_j] = 0$ for all $1 \leq j \leq p$.
	\item \textbf{Fixed signal, confounding, and noise variances}: for all $1 \leq j \leq p$:
	\begin{itemize}
	    \item $\cov(X_j, \epsilon_j) = \sigma^2_{\text{noise}}$,
	    \item $\cov(\sum_{k=1}^K g_{kj}(H_k)) = \sigma_h^2$, and
	    \item $\cov(\sum_{i \in \Parents{\Gtrue}{j}} f_{ij}(X_i) ) = \sigma^2_{\text{signal}}$.
	\end{itemize}
\end{itemize}
Since all nodes have unit variance, $\sigma_h^2$ equals the fraction of the variation in $X_j$ explained by $H$, and $\sigma_{\text{signal}}^{2}$ equals the fraction of the variation in $X_j$ explained by its parents.
In our experiments, we fix $\sigma^2_{\text{noise}} = 0.2$. 
In other words, if we could actually observe both $X$ and $H$, then the noise in the problem is small. 
To assess how the strength of confounding affects DAG recovery, we vary  $\sigma_h^2$. 
Note that $\sigma_{\text{signal}}^{2}$ is a deterministic function of $\sigma_h^2$: $\sigma_{\text{signal}}^{2} = \Var(X_{j}) - \sigma_{\text{noise}}^{2} - \sigma_{h}^{2} = 0.8 -  \sigma_{h}^{2}$. Therefore, varying $\sigma_h^2$ implicitly varies the observed signal variance $\sigma^2_{\text{signal}}$. \\

\noindent{\textbf{Linear / Non-Linear SCM.}} We randomly draw each $f_{ij}$ and $g_{kj}$ in \cref{eq:cam_gen_model} from the set of linear trends $\{ \theta x: \theta \in \R \}$, seasonal trends $\{ \theta \sin(\pi x): \theta \in \R \}$, or quadratic trends $\{ \theta x^2: \theta \in \R \}$ when generating non-linear data. 
For data generated from a linear SCM, we assume that each $f_{ij}$ and $g_{kj}$ belongs strictly to the set of linear trends. \\

\noindent{\textbf{Data Generation.}} We randomly draw $\Gtrue$ from an Erd{\"o}s-R{\'e}nyi random graph model with expected neighborhood size of $5$ and consider graphs with number of observed nodes $p \in \{ 250, 500, 1000\}$. 
As in \citet{causal_lrps}, we assume that each confounder is a direct cause of an observed variables with a $70\%$ chance. 
Given the graph, we randomly select a trend type for each edge (i.e., the $f_{ij}$ and $g_{kj}$ in \cref{eq:cam_gen_model}) with equal probability. 
We then randomly draw a weight $\theta \sim \text{Uniform}([-1, -.25] \cup [.25, 1])$ and appropriately scale weights of all parents simultaneously to satisfy the normalization constraints above. %
We finally add $N(0, \sigma_{\text{noise}}^2)$ noise to each node; see  \cref{A:sim_data_generate} for more details and references to our \texttt{python} code. 
Unless otherwise stated, we draw 25 random datasets for each specific configuration of confounding strength and dimensionality.  
We consider $K=1$ pervasive confounders, which allows us to plot how each node varies as a function of this confounder. See \cref{A:sensitivity} for our empirical results for other choices of $K$.

\subsubsection{Confounder Sufficient Statistics Estimation Performance} \label{sec:exp_pcss_est}

Based on \cref{thm:pcss_convg}, our estimate of the confounder sufficient statistics should decrease as $N$ and $p$ increase. To test this empirically, we keep the ratio of $\frac{p}{N}$ fixed at 2, and  vary the confounding strength and dimensionality. \cref{fig:pcss_mse} plots the maximum mean-squared error, which equals $\max_{1 \leq j \leq p} \frac{1}{N}\sum_{n=1}^N \left(s_j^{(n)} - \hat{s}_j^{(n)} \right)^2$.
In the linear case, we can compute $S$ analytically as a function of $H$ (see \cref{ex:linear_s_form}). 
However, in the non-linear case, there is no analytical expression for $S$ in term of $H$. 
Nonetheless, if $X_j$ only has confounders as parents, then $S_j =  \sum_{k=1}^K g_{kj}(H_k)$. 
%
%
Thus, in \cref{fig:pcss_mse}, the max MSEs for the non-linear SCM case are only with respect to the $S_j$ associated with nodes that only have confounders as parents. 
From \cref{fig:pcss_mse}, we see that the MSE decreases as $p$ increases or when the strength of the confounders increases.
\begin{figure}[t]
\centering
\myfigure{width=.8\linewidth}{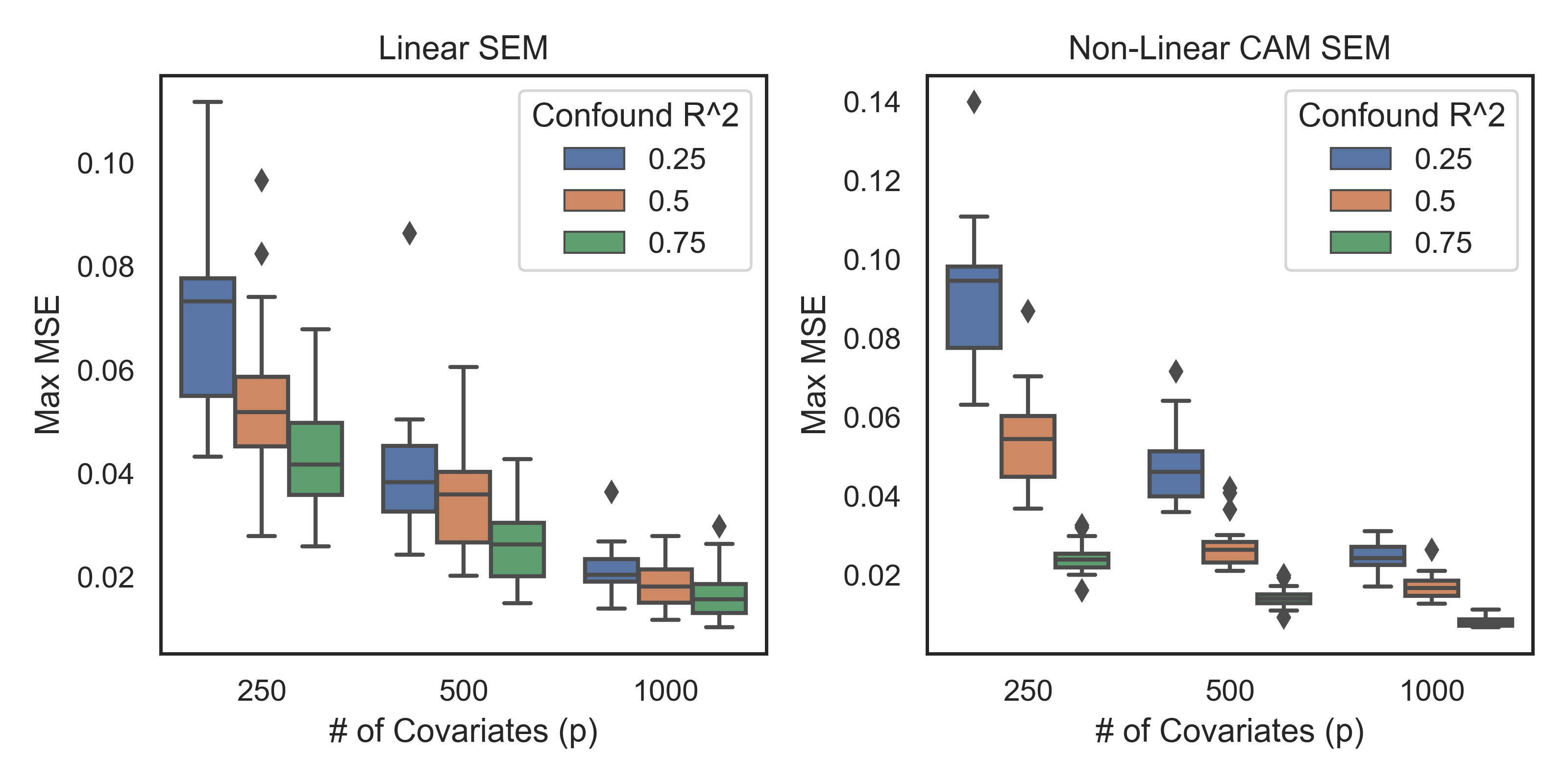}
\caption{Maximum Mean-Squared Error (MSE) across all dimensions for estimating $\{s^{(n)}\}_{n=1}^N$ via PCSS. Twenty-five total simulations were performed for each dataset configuration.} \label{fig:pcss_mse} 
\end{figure}

To evaluate the quality of PCSS for nodes that also have observed nodes as parents, we can make qualitative assessments. 
In particular, since there are only $K=1$ confounders, we can plot $x_j^{(n)}$ against $h_1^{(n)}$.
The visual trend we see from the resulting scatterplot corresponds to the desired conditional expectation $s_j^{(n)} = \E[X_j \mid H_1=h_1^{(n)}]$. 
By plotting the confounder sufficient statistics estimated from PCSS on the same plot for each dimension, we can qualitatively check for a matching trend. 
\cref{fig:pcss_syn_visual} provides such a plot for a particular random data simulation and select nodes that have at least one parent. We see that PCSS indeed matches the data trend well.
\begin{figure}[b]
\myfigure{width=\linewidth}{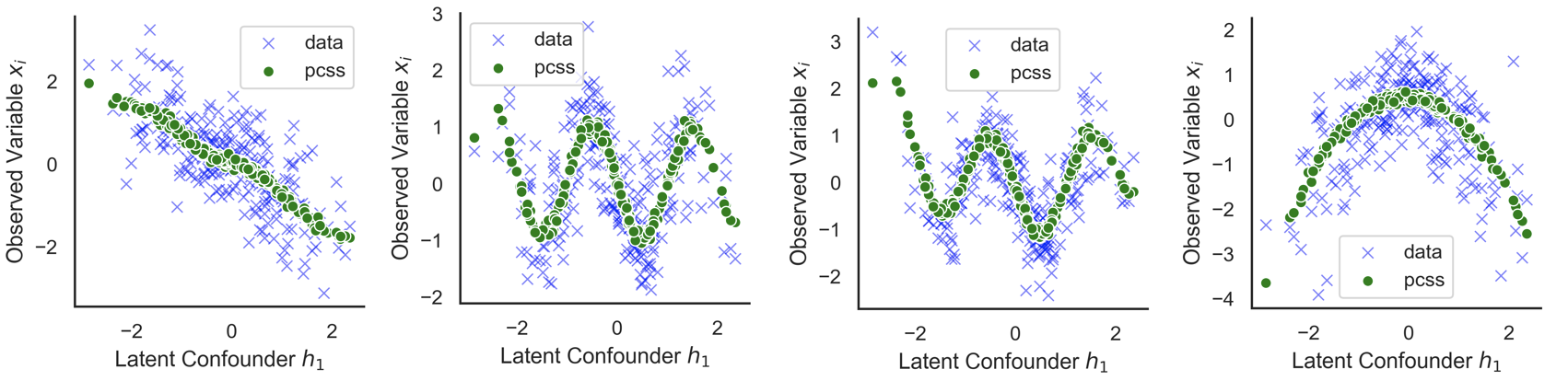}
\centering
\caption{Non-linear estimation of $E[X_i \mid H]$ via PCSS for nodes with at least one parent.} \label{fig:pcss_syn_visual}  
\end{figure}

\subsubsection{Evaluation of the DeCAMFounder Score} \label{sec:decam_perform}
We now evaluate how well the DeCAMFounder scores DAGs using noisy estimates of the confounder sufficient statistics.
For this evaluation, we separate the merits of the score function from the selected optimization procedure. 
In particular, since many optimization procedures for DAG estimation greedily optimize over parent sets, we focus on the task of effectively scoring parent sets. \\

\noindent \textbf{Description of the Parent Set Evaluation Tasks.}
In greedy procedures, changes to a parent set typically involve a single node addition (adding a node to the current parent set) or deletion (removing a node from the current parent set). 
In the presence of latent confounding, methods that do not account for confounders are susceptible to adding spurious edges. 
For example, suppose nodes $X_i$ and $X_j$ have a common hidden cause but no shared parents in $\Gtrue$ for some $1 \leq i,j \leq p, i \neq j$.
Then, a method that assumes no latent confounding might add an edge from $X_i$ to $X_j$ or vice versa with high probability, even when all the true parents of the node are included in the current parent set.

This observation leads to a natural evaluation procedure: consider the set of all nodes $C \coloneqq \{X_j: D_j \neq 0, 1 \leq j \leq p\}$ that have direct confounders (\raedit{see \cref{eq:cam_gen_model} for the definition of $D_j$}). 
Randomly select a node $X_{j'}$ from $C$ and randomly sample $T$ nodes $X_{r_1}, \ldots, X_{r_T}$ from $C$ that are not parents of $X_{j'}$. Then, we evaluate the methods in terms of scoring the following $T+1$ parent sets:
\begin{equation*}
	\begin{split}
	P_{\text{correct}} &= \Parents{\Gtrue}{j'} \\
	P_{t}  &=  \Parents{\Gtrue}{j'} \cup \{ X_{r_t} \} \quad \text{for} \quad \raedit{1 \leq t \leq T}. \\
	\end{split}
\end{equation*}
Again, we might expect that methods that do not correct for confounders score some $P_{t}$ higher than $P_{\text{correct}}$ due to spurious correlations created by the confounders. 
In our evaluation, we pick $T=100$. 
We call this evaluation procedure the \emph{Wrong Parent Addition Task}. 
We note that this task is similar to the "CauseEffectPairs" causality challenge introduced by \citet{cause_pairs}.

The second parent set evaluation task concerns the node deletion phase.
Here, we expect that linear methods may suffer by potentially removing true parent nodes. 
For example, if a parent node has low correlation with its child (e.g., as is the case with our sine and quadratic trends), then that parent might be removed when using a sparsity-inducing score function such as BIC or penalized log-likelihood with an $L_0$ or $L_1$ penalty. 

Similar to the Wrong Parent Addition task, we first randomly select a node $X_{j'}$ with at least one observed parent node, \raedit{where $1 \leq j' \leq p$}. 
For each node in the parent set of $X_{j'}$, let
\begin{equation*}
	\begin{split}
	P_{\text{correct}} &= \Parents{\Gtrue}{j'} 
    \\
	P_{t} &= \Parents{\Gtrue}{j'} \setminus \{ X_{i} \} \quad \text{for} \quad X_i \in \Parents{\Gtrue}{j'}. \\
	\end{split}
\end{equation*}
We would again like to understand if certain score functions favor incorrect $P_t$ parent sets when the parents have non-linear effects on the target. 
We call this evaluation procedure the \emph{Correct Parent Deletion} task.
These two tasks are illustrated in \cref{fig:parent_eval_vis}.\\
\begin{figure}
        \centering
        \begin{subfigure}[b]{0.25\textwidth}
            \centering
            \myfigure{width=\textwidth}{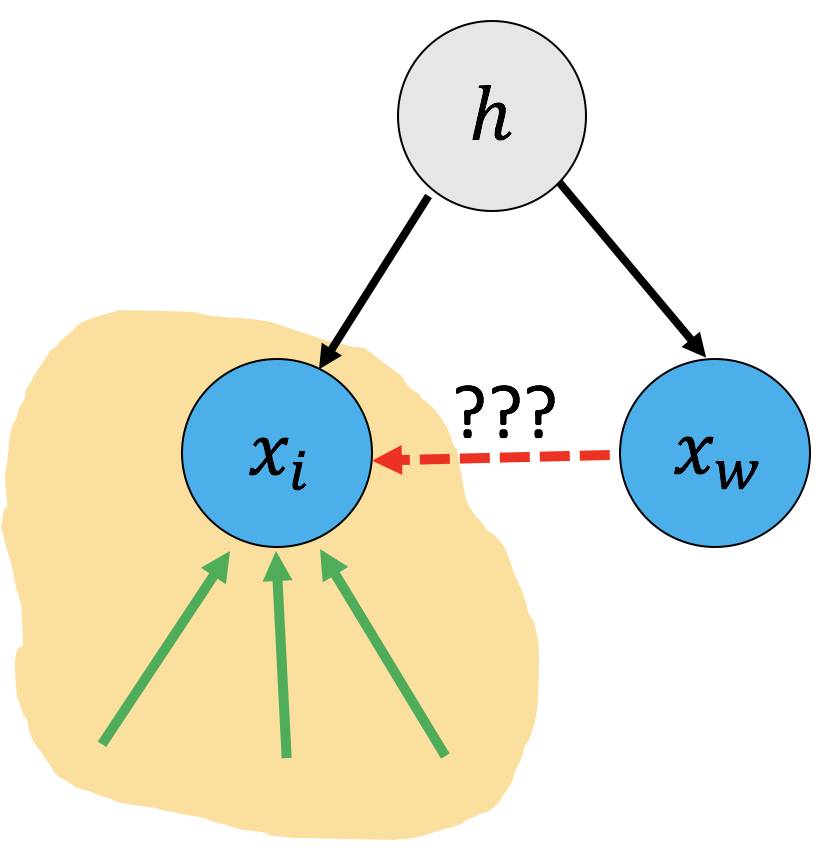}
            \caption[Network2]%
            {{\small Wrong Parent Addition}}    
            \label{fig:tf_corrs}
        \end{subfigure}
        \hspace{0.15\textwidth}
        \begin{subfigure}[b]{0.25\textwidth}  
            \centering 
            \myfigure{width=\textwidth}{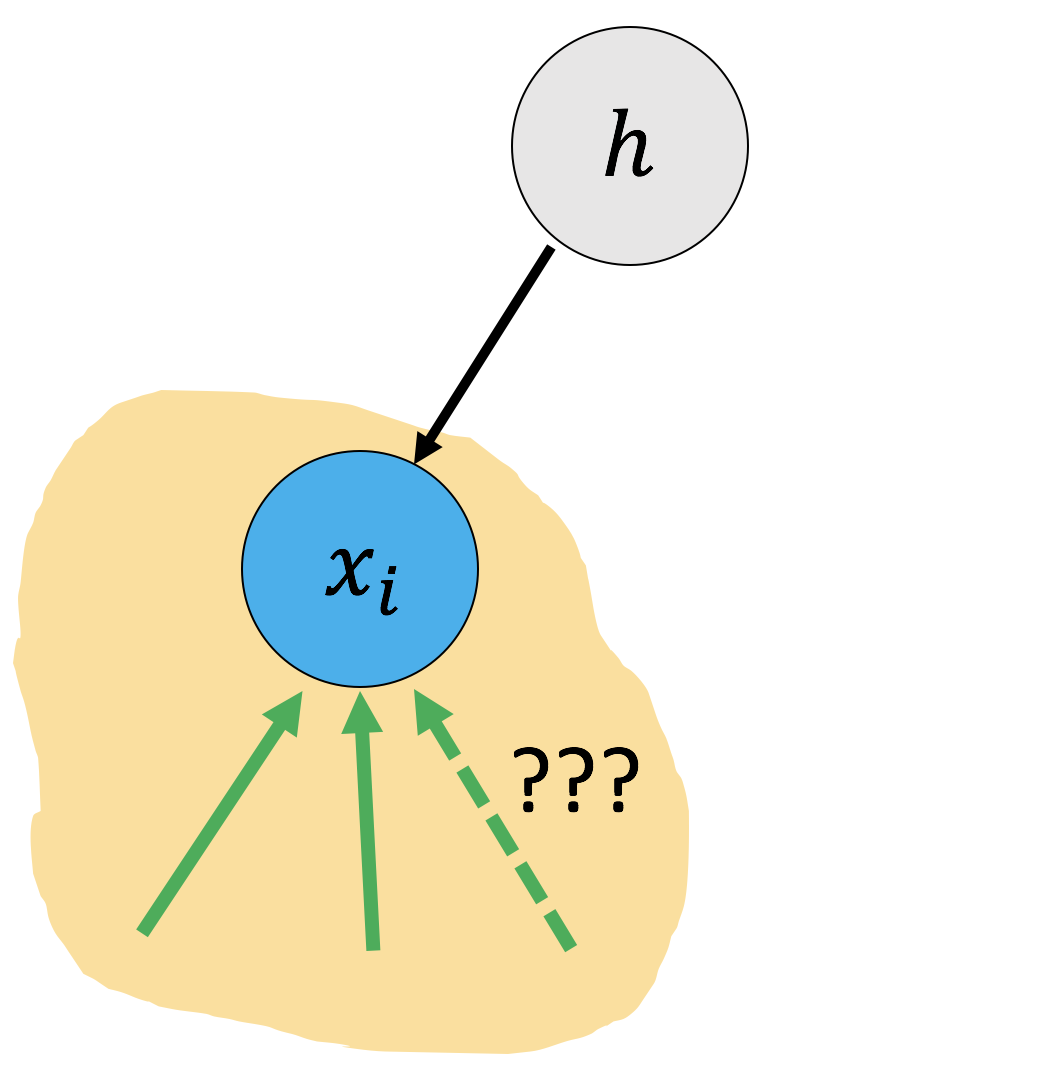}
            \caption[]%
            {{\small Correct Parent Deletion}}    
            \label{fig:ovarian_scree}
        \end{subfigure}
        \caption{Our parent set evaluation tasks. Solid arrows represent the set of true edges, and dotted arrows indicate a potential incorrect modifcation to the true parent set of a node.} \label{fig:parent_eval_vis}
\end{figure}

\noindent \textbf{Benchmark Methods.}
We now describe our benchmark scoring methods. 
The first three methods are linear, and use the BIC score with different estimates of the covariance matrix. 
In particular, the BIC score (under Gaussian errors) is a function of a target node $X_j$, parent set \raedit{$P_j \subset \{1, \cdots, p\} \setminus \{j\}$}, and estimated covariance matrix \raedit{$\hat{\Sigma} \in \R^{p \times p}$, $1 \leq j \leq p$}: 
\begin{equation} \label{eq:BIC}
	\text{BIC}(X_j, P_j, \hat{\Sigma}) 
    = 
    -\left(\frac{N}{2} \log \left( 2\pi \hat{\sigma}_{j \mid P_j}^2 \right) + \frac{N}{2} \right)  -.5\log(N)(|P_j| + 2),
\end{equation}
where the conditional noise variance (based on Schur complements) equals
\begin{equation*}
	\hat{\sigma}_{j \mid P_j}^2 = \hat{\Sigma}_{jj} -  \hat{\Sigma}_{j P_j} \hat{\Sigma}_{P_j P_j}^{-1} \hat{\Sigma}_{P_j j}.
\end{equation*}

\noindent The benchmarks are as follows:
\begin{enumerate}
	\item \textbf{\vanillabic:} scores each parent set by inputting the sample covariance matrix $\frac{1}{N} \bbX^T \bbX$ into \cref{eq:BIC}.
	
	\item \textbf{\lrpsbic:} scores each parent using the covariance matrix estimated from a low-rank plus sparse (LRPS) decomposition of the sample precision matrix; see \citet{causal_lrps}. We use the author's code in \texttt{R} to fit this decomposition, and pick the hyperparameters using cross-validation.
	\item \textbf{\pcssbic:} scores each parent using the covariance matrix $\frac{1}{N} (\bbX - \bbS)^T (\bbX - \bbS)$, where $\bbS$ is the output from \cref{algo:pcss_est}. 
    We pick one principal component and 3 principal components for the linear and non-linear case, respectively.
    This choice was based on visual inspection of the spectrum of the covariance matrix. See the Appendix. 
	
	\item \textbf{\cam:} scores each parent via \cref{eq:marg_like} but sets $\bbS_{:,j}$ equal to the zero matrix. 
    This method corresponds to fitting a standard causal additive model without correcting for the confounders.
	
	\item \textbf{\camobs:}  scores each parent by computing \cref{eq:gp_marg_like} but replacing $\bbS$ with the actual confounders $\bbH$. 
    This method corresponds to fitting a standard causal additive model when both $X$ and $H$ are observed, and is therefore not realizable in practice.
\end{enumerate}

For both parent set tasks, we fix $N=250$ and $p=500$ but vary the strength of confounding. 
For the Correct Parent Deletion task, we exclude the linear trend from the set of trends when generating non-linear data to focus on the particular problem of modeling non-linearities. 
Apart from that, the data in both settings are generated as specified at the beginning of this section. 
\cref{fig:non_linear_par_tasks} shows the non-linear SEM results. 
The remaining results are shown in \cref{fig:parent_set_addition} and \cref{fig:parent_set_removal}.  In these figures, the different methods are compared using the following metric: \\

\noindent \textbf{Prop. Times MLL Wrong $>$ MLL True}: out of the $T$ incorrect parent sets, what proportion  have scores larger than the true parent set. Lower is better here. \\

\cref{fig:non_linear_par_tasks} and \cref{fig:parent_set_addition} show that methods that account for confounders (i.e., \lrpsbic, \pcssbic, \camobs, \decamfound), place higher probability on the correct parent set for the Wrong Parent Addition task across different settings relative to those that do not model the confounders. 
In addition,  \cref{fig:non_linear_par_tasks} and \cref{fig:parent_set_removal} show that for the Correct Parent Deletion task, the linear methods suffer in the non-linear setting. 
Note that unlike \pcssbic, \lrpsbic~suffers even in the linear setting because it induces a very sparse graph, causing it to delete true parents.

Including \camobs~in our evaluation allows us to understand how parent set recovery is affected by estimation error from only having a finite amount of observational data. 
Interestingly, our method, which does not require knowing $H$, sometimes outperforms \camobs. 
This might be because our method leverages all $p$ observed nodes to estimate the confounding variation via PCA. 
\camobs, on the other hand, estimates the confounding variation one node at a time by regressing each observed node on $H$. 
Finally, all methods suffer when the confounding strength increases since the observed signal necessarily becomes weaker. 
We provide additional empirical results including a DAG evaluation scoring task in \cref{sec:add_syn}. 

\begin{figure}
        \centering
        \begin{subfigure}[b]{0.49\textwidth}
            \centering
            \myfigure{width=\textwidth}{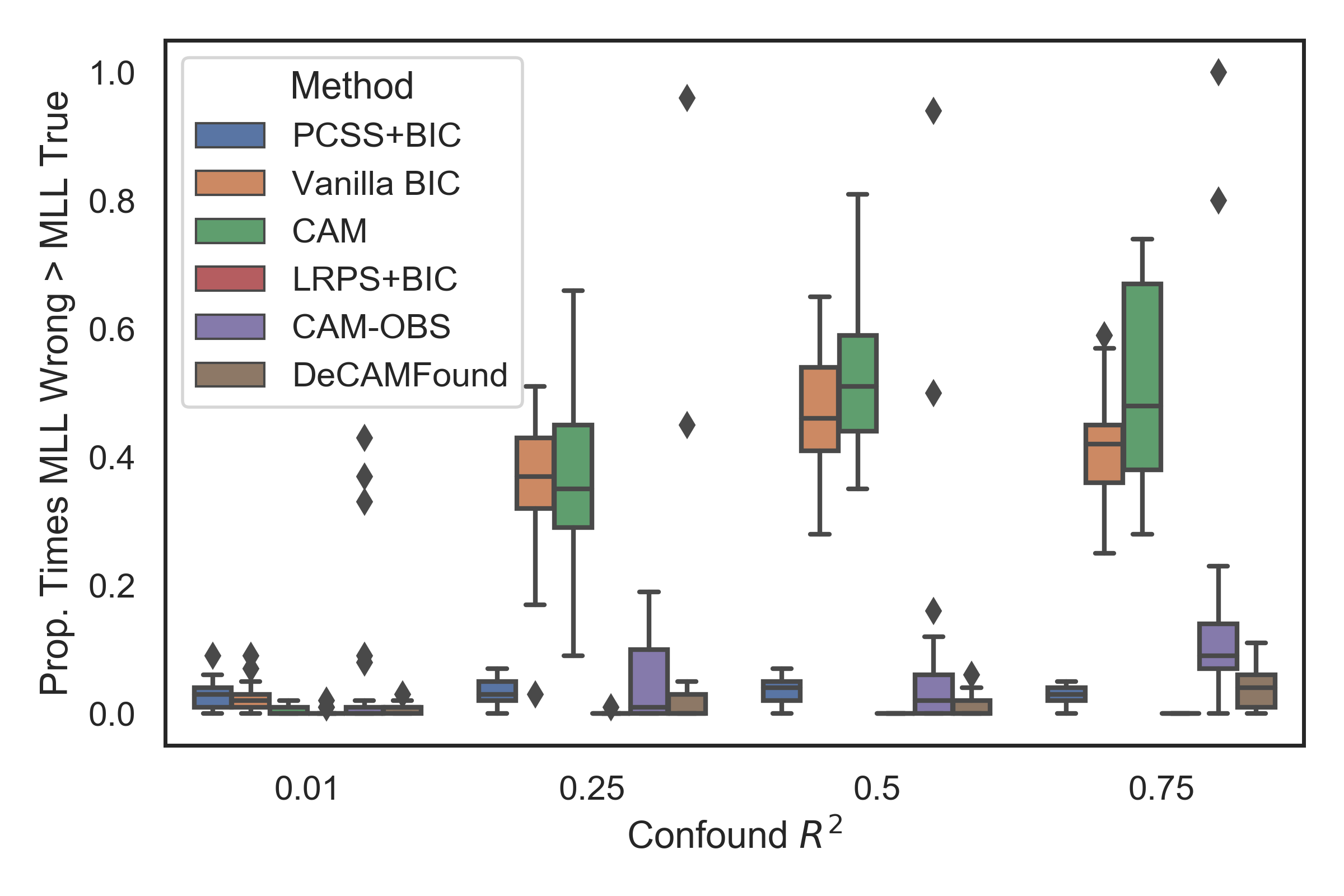}
            \caption[Network2]%
            {{\small Wrong Parent Addition Task}}    
        \end{subfigure}
        \begin{subfigure}[b]{0.49\textwidth}  
            \centering 
            \myfigure{width=\textwidth}{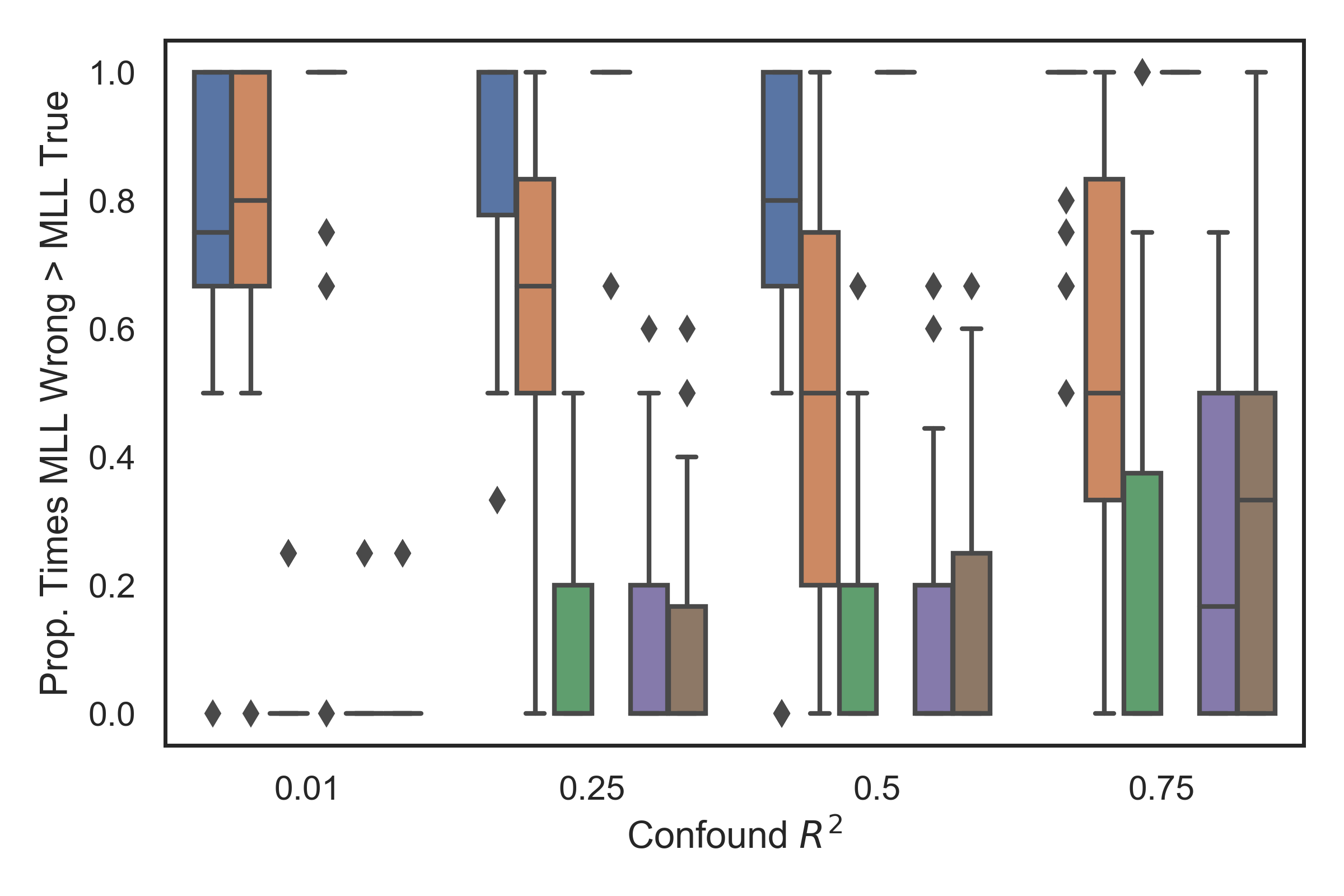}
            \caption[]%
            {{\small Correct Parent Deletion Task}}    
        \end{subfigure}
        \caption{Results for the Wrong Parent Addition and Correct Parent Deletion tasks (lower values on the y-axis are better for both tasks).  The data are generated according to a non-linear SEM. 25 total simulations per dataset configuration were performed.} \label{fig:non_linear_par_tasks}
\end{figure}

\subsection{Real Data: Ovarian Cancer Dataset} \label{sec:real_data}
We use the ovarian cancer RNA-seq dataset from \citet{causal_lrps} for our real-data analysis, which consists of $N=247$ human samples and $p=501$ genes. 
There are 15 total transcription factors (TFs), and the remaining 486 genes interact with at least one of these TFs. 
This dataset has a partial ground truth; the 15 TFs can only be causes and not effects of the 486 genes. 
The reference network for the 501 variables is obtained using Netbox, a software tool for performing network analysis using biological knowledge and community network-based approaches \citep{netbox}. 
A caveat, however, is that edges in this network might not have a precise causal interpretation. 

One metric used by \citet{causal_lrps} was a count of the number of directed edges from TFs to genes that agree with the edges in NetBox. 
The authors stated that  confounding can be expected in this dataset due to unobserved transcription factors or batch effects. 
We here take a different approach for the evaluation and instead \emph{explicitly} create confounders by design. 
In particular, we remove all 15 TFs from the dataset and assume that we only observed the remaining 486 genes. 

\begin{figure}
        \centering
        \begin{subfigure}[b]{0.48\textwidth}
            \centering
            \myfigure{width=\textwidth}{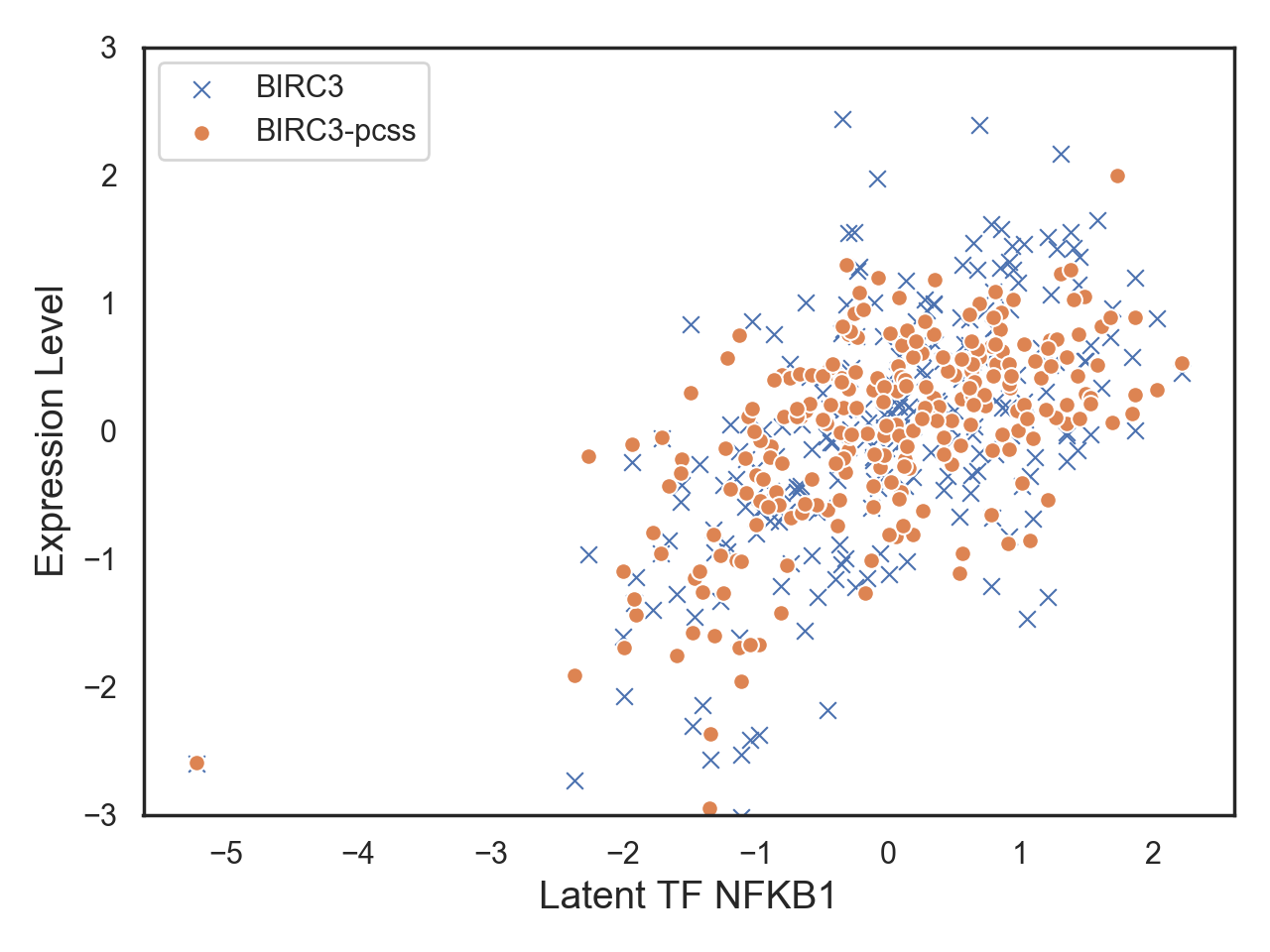}
        \end{subfigure}
        \hfill
        \begin{subfigure}[b]{0.48\textwidth}  
            \centering 
            \myfigure{width=\textwidth}{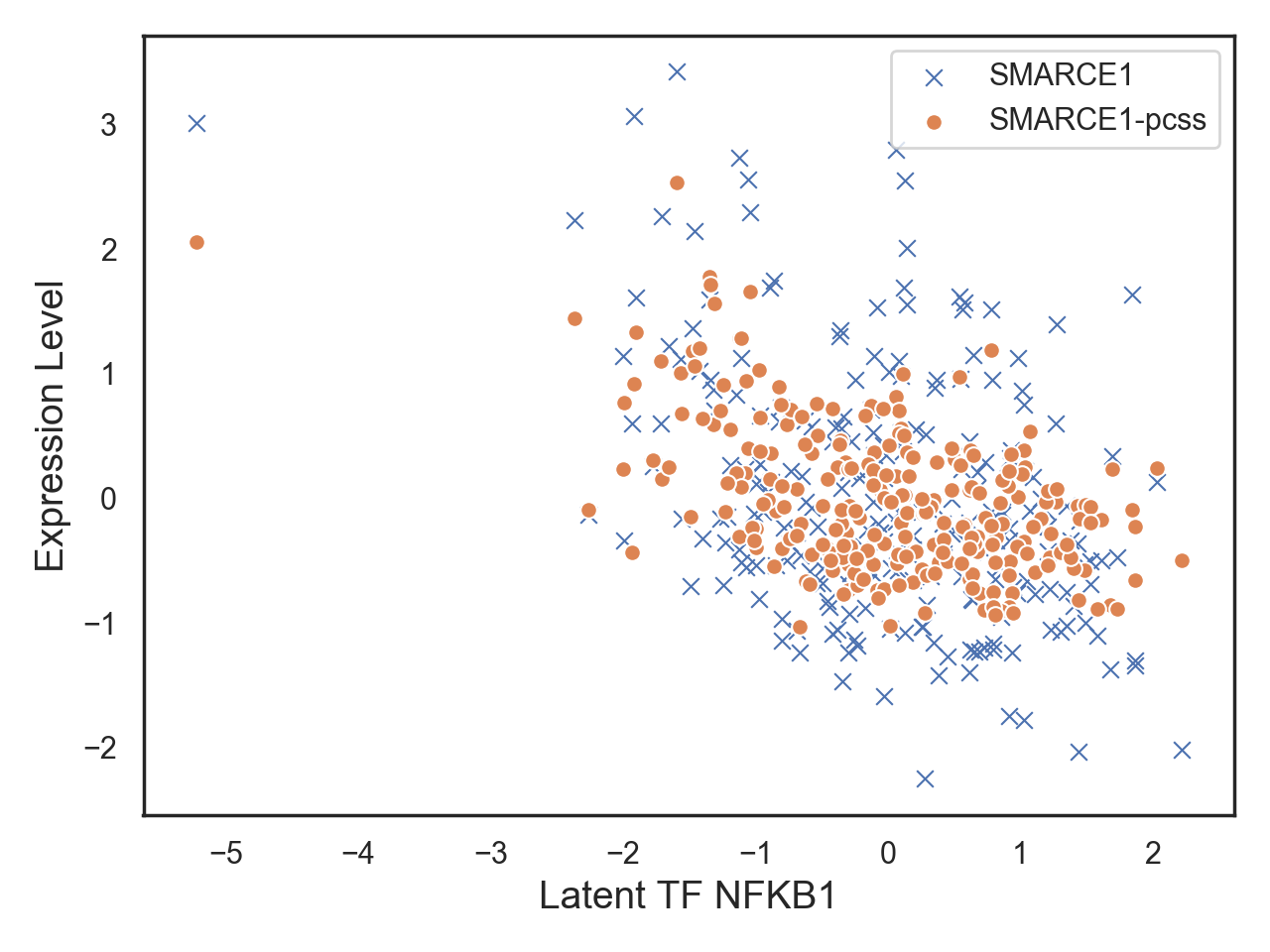}
        \end{subfigure}
        	\caption{The two genes most positively and negatively correlated with the transcription factor NFKB1, respectively. GENE-pcss refers to the total latent confounding variation estimated for that gene via PCSS. GENE refers to the observed values of the gene.} \label{fig:nfkb1_vs_pcss}
\end{figure}

\begin{figure}
        \centering
        \begin{subfigure}[b]{0.48\textwidth}
            \centering
            \myfigure{width=\textwidth}{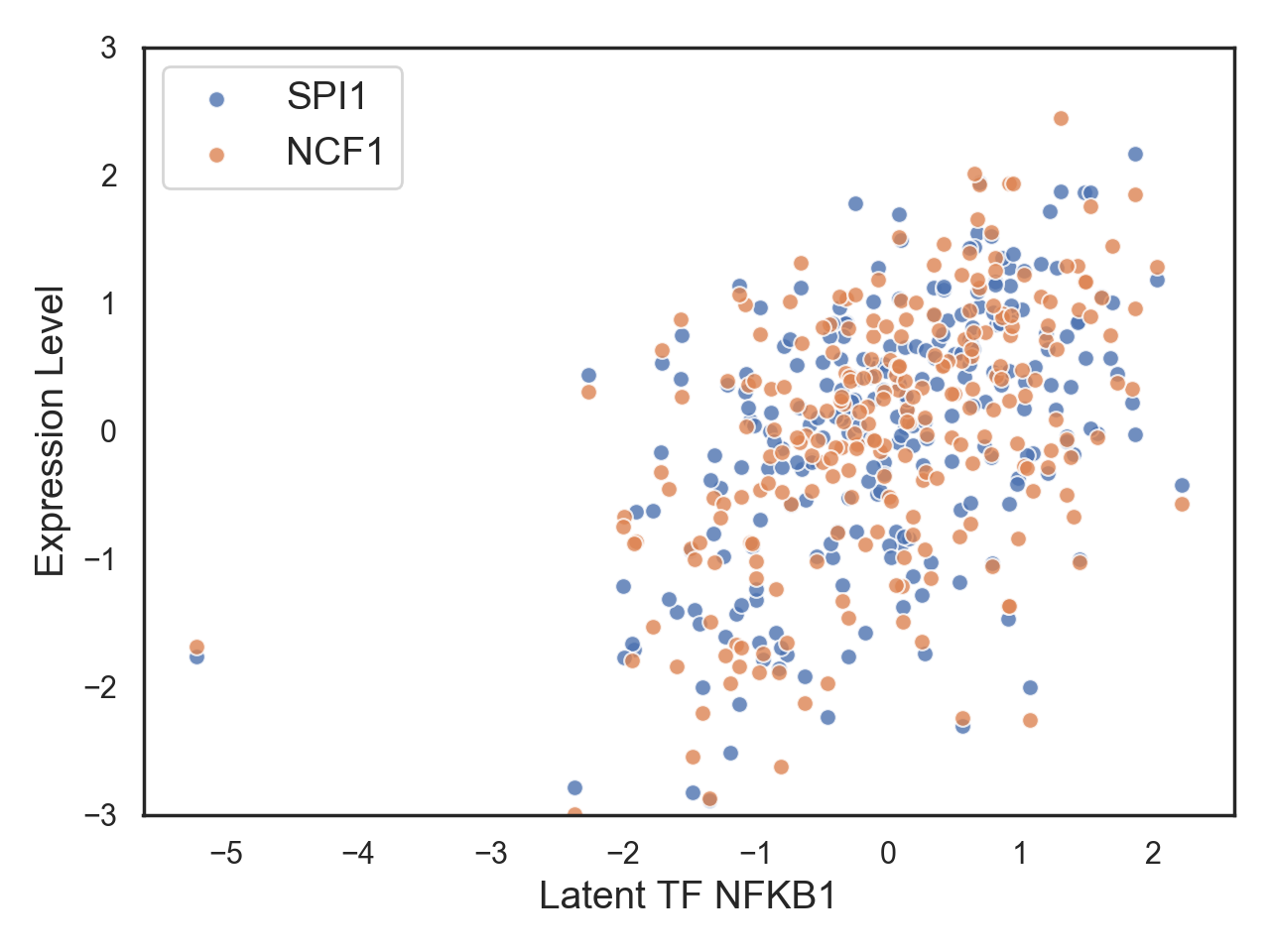}
        \end{subfigure}
        \hfill
        \begin{subfigure}[b]{0.48\textwidth}  
            \centering 
            \myfigure{width=\textwidth}{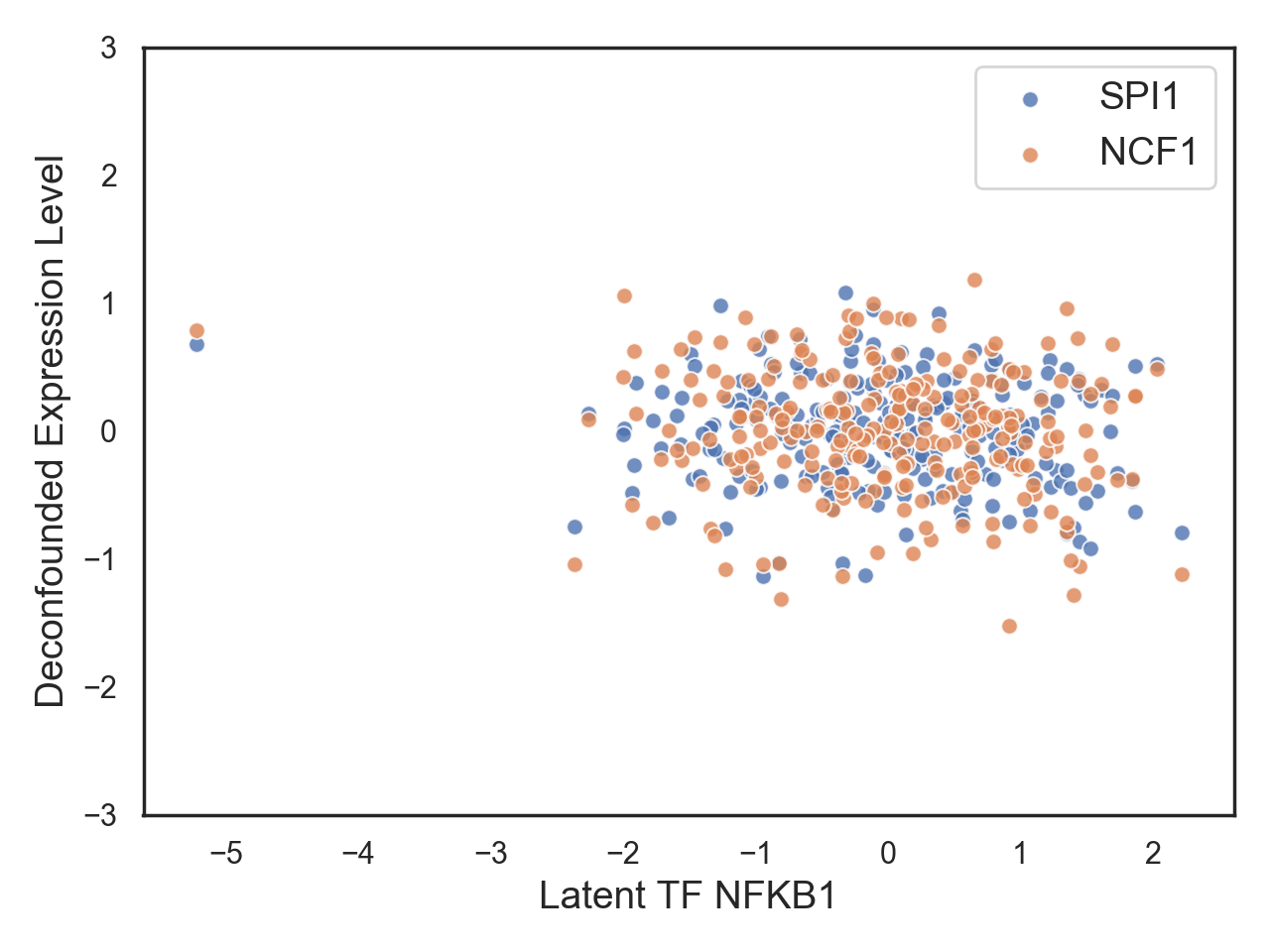}
        \end{subfigure}
        	\caption{Both genes have a correlation greater than 0.4 with NFKB1 (left hand plot). After subtracting out  the confounding variation estimated using PCSS for each gene (denoted as "deconfounded" expression level), the genes are no longer correlated with the unobserved transcription factor NFKB1.} \label{fig:remove_out_tf}
\end{figure}

\subsubsection{Estimating the Confounding Variation} \label{sec:real_data_pcss}

By removing the TFs, we can evaluate the methods in a similar fashion as the simulated experiments since we know the true values of the TFs. 
In the following, we first assess our ability to estimate the confounder sufficient statistics. 
Suppose that an observed gene $X_j$ is strongly correlated with one of the 15 latent TFs $T_k$, for some $1 \leq j \leq 486$ and $1 \leq k \leq 15$.
Then, $\E[X_j \mid T_k ] \approx \E[X_j \mid H] = S_j$. 
Since we know $T_k$, we can produce a similar plot as in \cref{fig:pcss_syn_visual}. 
To this end, we look at NFKB1, which is a transcription factor known to be associated with ovarian cancer \citep{nfkb1}. 
In \cref{fig:nfkb1_vs_pcss}, we look at the highest positive and negatively correlated genes with NFKB1, which are BIRC3 and SMARCE1 respectively. 
We see that the estimated confounding variation for each gene estimated from PCSS correlates well with the unobserved transcription factor NFKB1. 
This strong correlation suggests that PCSS estimates the variation explained by the unobserved confounders well in this dataset.

\begin{figure}[t]
\centering
\myfigure{width=.55\linewidth}{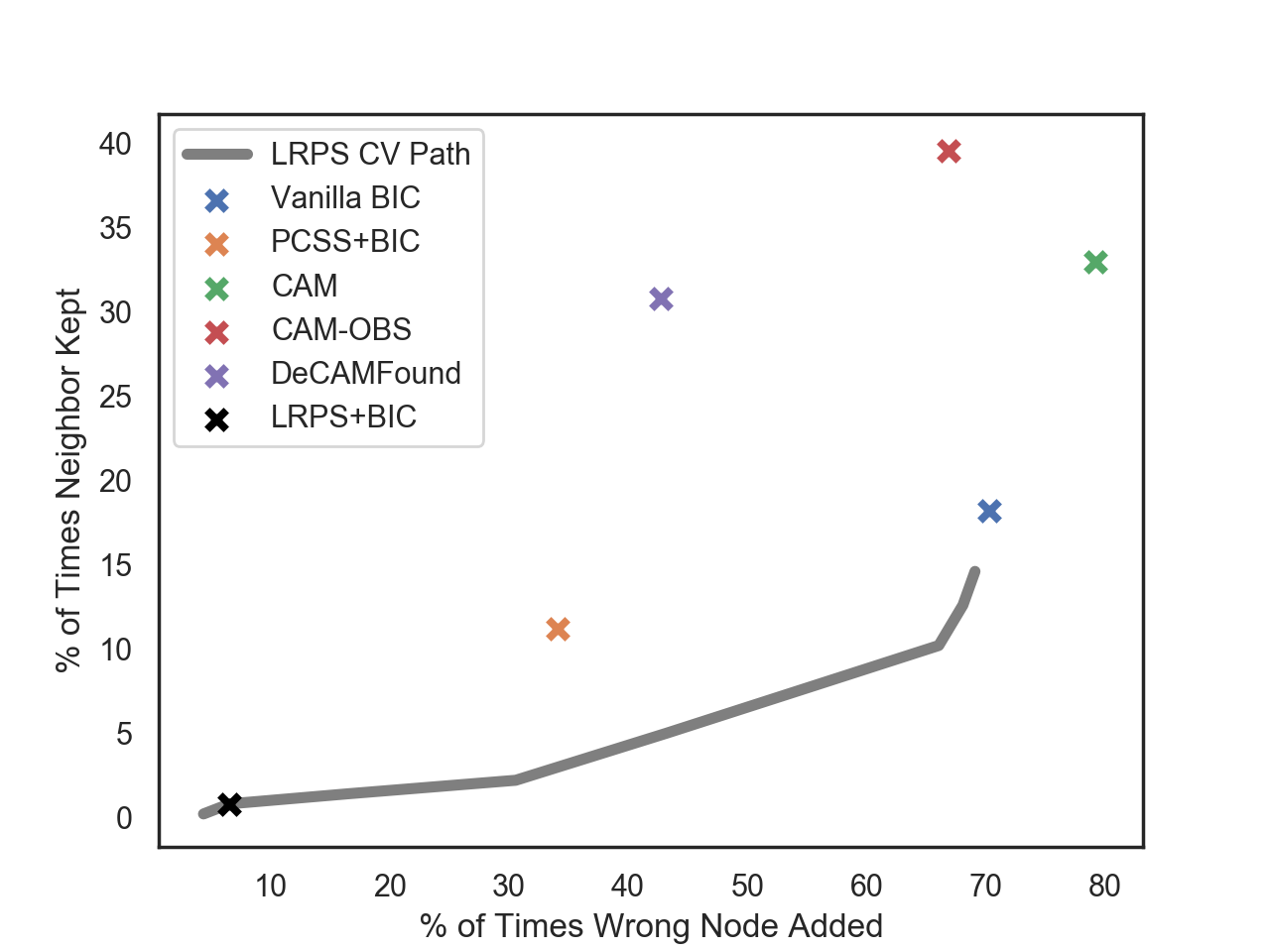}
\caption{The x-axis denotes the proportion of times a method scored the incorrect parent appended to the neighborhood set higher than just the neighborhood set  of a node (our proxy for the true parent set).
The y-axis denotes the proportion of times a method scored the full neighborhood set higher than the neighborhood set after removing out one of the neighbors. 
For LRPS, we provide the intermediate results for each covariance matrix outputted along its cross-validation path. 
The `x' for LRPS corresponds to the performance of the covariance matrix selected based on cross-validation. 
We compute the \emph{positive likelihood ratio} for each method which equals the ratio between our proxy for the true positive rate (y-axis) and false positive rate (x-axis). 
These ratios are as follows: \vanillabic=.26, \pcssbic=.33, \cam=.42, \camobs=.59, \decamfound= .72, \lrpsbic=.12.
}  
\label{fig:ovarian_results}
\end{figure} 

\subsubsection{Parent Recovery Performance} \label{sec:real_data_decam}

Unlike in the simulated dataset experiments, we do not know the true causal graph.
Thus, for the Wrong Parent Addition task and the Correct Parent Deletion task, we use the graph structure outputted by NetBox, as in \citet{causal_lrps}.
%
%
In particular, for a node $X_{j'}$, we let $P_{\text{correct}}$ be the neighbors of $X_{j'}$ in the graph outputted by NetBox.
For the Wrong Parent Addition task, we require a set of ``spurious neighbors'' of $X_{j'}$.
For this set, we take all variables $X_{r'}$ such that $X_{j'}$ and $X_{r'}$ are conditionally independent given $P_t$ and the TFs, but are conditionally dependent given only $P_{\text{correct}}$.
Then, we let $P_t = P_{\text{correct}} \cup \{ X_r \}$ for $X_r$ picked randomly from this set.

We focus on gene pairs $X_{j'}$, $X_{r'}$ with at least one strong TF confounder.
In particular, we require that each gene has correlation greater than 0.4 with one of the 15 removed TFs. 
This choice makes it more likely that one of the 15 left-out TFs is the source of confounding, rather than some other unobserved confounder which we have no control over.
This choice results in roughly 1000 parent sets to score. 
One such gene pair with a strong TF confounder is illustrated in \cref{fig:remove_out_tf}.
We see that subtracting the confounder sufficient statistics from each gene results in removing the confounding effect of the strongest transcription factor NFKB1.

For the Correct Parent Deletion task, we randomly sample 500 edges from the NetBox undirected graph. 
For each edge we randomly select one endpoint as $X_{j'}$ and define $P_t = P_{\text{correct}} \setminus \{ X_r \}$ for the other endpoint $X_r$.

The results for both tasks are summarized in \cref{fig:ovarian_results}. 
LRPS produces a very sparse graph, and hence almost never adds a wrong parent. 
However, it almost always removes true edges, which comes at the expense of statistical power. 
Our method has the highest power per unit of false positives. 
Rather surprisingly, \camobs, which explicitly uses the 15 held out transcription factors, selects wrong parents more often than our method.
This might be the result of having additional (pervasive) confounders beyond the 15 TFs we introduced by design, as raised by \citet{causal_lrps}. \\

\noindent \textbf{Conclusion.}
We \edit{have shown how to} identify the causal graph among the observed nodes in the setting of non-linear effects and pervasive confounders. We proposed the DeCAMFounder, which consistently estimates \edit{a correct ordering} of $\Gtrue$ using Gaussian processes. Since the DeCAMFounder  explicitly accounts for confounders and non-linear effects, we found improved performance on both simulated and real datasets relative to existing methods. An interesting \edit{future }direction is handling selection bias. In \citet{causal_lrps}, for example, the authors showed how to account for selection bias and pervasive confounders in the linear setting. It would be interesting to explore whether a similar idea could be used to extend the DeCAMFounder to handle selection bias. \\

\noindent \textbf{Data and Code Availability.} The simulated data, real data, and code used to generate our empirical results are available at \texttt{https://github.com/uhlerlab/decamfound}. \\

\noindent \textbf{Acknowledgements}
We are thankful to the two referees, associate editor, and editor for their very helpful comments which led to an improved manuscript.\\

\noindent \textbf{Funding.} The authors were partially supported by NSF (DMS-1651995) and TRIPODS program (DMS-2022448), ONR (N00014-22-1-2116), the United States Department of Energy (DOE) Office of Advanced Scientific Computing Research (ASCR) via the M2dt MMICC center (DE-SC0023187), the MIT-IBM Watson AI Lab, the Eric and Wendy Schmidt Center at the Broad Institute, and a Simons Investigator Award to C.~Uhler. \\

\noindent \textbf{Conflict of Interest.} None declared.



\bibliographystyle{rss}
\bibliography{references}

\section*{List of Figure Legends} \begin{description}
  \item[\cref{fig:exog_repar}]  The right-hand figure reparameterizes the model on the left such that the unobserved variable $H$ is a source in the graph.
The bold arrows represent the DAG $\Gtrue$ corresponding to the conditional distribution $\pr(X \mid H)$.
  
  \item[\cref{fig:pcss_mse}]  Maximum Mean-Squared Error (MSE) across all dimensions for estimating $\{s^{(n)}\}_{n=1}^N$ via PCSS. Twenty-five total simulations were performed for each dataset configuration.
  
  \item[\cref{fig:pcss_syn_visual}] Non-linear estimation of $E[X_i \mid H]$ via PCSS for nodes with at least one parent.

  \item[\cref{fig:parent_eval_vis}] Our parent set evaluation tasks. Solid arrows represent the set of true edges, and dotted arrows indicate a potential incorrect modifcation to the true parent set of a node.
  
  \item[\cref{fig:non_linear_par_tasks}] Results for the Wrong Parent Addition and Correct Parent Deletion tasks (lower values on the y-axis are better for both tasks).  The data are generated according to a non-linear SEM. 25 total simulations per dataset configuration were performed.
  
  \item[\cref{fig:nfkb1_vs_pcss}] The two genes most positively and negatively correlated with the transcription factor NFKB1, respectively. GENE-pcss refers to the total latent confounding variation estimated for that gene via PCSS. GENE refers to the observed values of the gene.
  
  \item[\cref{fig:remove_out_tf}] Both genes have a correlation greater than 0.4 with NFKB1 (left hand plot). After subtracting out  the confounding variation estimated using PCSS for each gene (denoted as "deconfounded" expression level), the genes are no longer correlated with the unobserved transcription factor NFKB1.
  
  \item[\cref{fig:ovarian_results}] The x-axis denotes the proportion of times a method scored the incorrect parent appended to the neighborhood set higher than just the neighborhood set  of a node (our proxy for the true parent set).
The y-axis denotes the proportion of times a method scored the full neighborhood set higher than the neighborhood set after removing out one of the neighbors. 
For LRPS, we provide the intermediate results for each covariance matrix outputted along its cross-validation path. 
The `x' for LRPS corresponds to the performance of the covariance matrix selected based on cross-validation. 
We compute the \emph{positive likelihood ratio} for each method which equals the ratio between our proxy for the true positive rate (y-axis) and false positive rate (x-axis). 
These ratios are as follows: \vanillabic=.26, \pcssbic=.33, \cam=.42, \camobs=.59, \decamfound= .72, \lrpsbic=.12.
  
\end{description}

\newpage

\appendix 

\section*{Supplementary Material}

\counterwithin{figure}{section}
\counterwithin{table}{section}

\section{Notation} \label{A:notation} In the following table, we provide an overview of the main notations used in the paper.

\begin{center}
\begin{tabular}{c|c|c}
    \multicolumn{3}{l}{\textbf{Random Variables}}
    \\
    \hhline{===}
    $X$ & Observed random vector, with dimension $p$ & \cref{sec:problem_setup}
    \\
    $H$ & Unobserved random vector, with dimension $K$ & \cref{sec:problem_setup}
    \\
    $\epsilon$ & Exogenous noise variables & \cref{sec:problem_setup}
    \\
    $S$ & Confounder sufficient statistics $\E[X \mid H]$ & \cref{sec:suff_stats}
    \\
    $U$ & Component of $X$ orthogonal to $\hilb_M$ & \cref{sec:poet_reduct}
    \\
    \multicolumn{3}{l}{}
    \\
    \multicolumn{3}{l}{\textbf{Distributions}}
    \\
    \hhline{===}
    $\Pobs$ & Marginal distribution over $X$ & \cref{sec:problem_setup}
    \\
    $\bbQ$ & Joint distribution over $X, H$ & \cref{sec:problem_setup}
    \\
    \multicolumn{3}{l}{}
    \\
    \multicolumn{3}{l}{\textbf{Graphs}}
    \\
    \hhline{===}
    $\Gtrue$ & Conditional DAG associated with $\bbQ(X \mid H)$ & \cref{sec:problem_setup}
    \\
    $\overline{\cG}$ & Full DAG over $X, H$ & \cref{sec:problem_setup}
    \\
    $\Parents{\cG}{j}$ & Parents of $X_j$ in $\cG$ &\cref{sec:problem_setup}
    \\
    $\Ancestors{\cG}{j}$ & Ancestors of $X_j$ in $\cG$ & \cref{sec:decam}
    \\
    $\dagset$ & Set of possible DAGs & \cref{sec:decam}
    \\
    \multicolumn{3}{l}{}
    \\
    \multicolumn{3}{l}{\textbf{Data}}
    \\
    \hhline{===}
    $N$ & Number of samples & \cref{sec:problem_setup}
    \\
    $p$ & Number of observed nodes & \cref{sec:problem_setup}
    \\
    $\bbX$ & Matrix of samples, $\in \R^{N \times p}$ & \cref{sec:poet_reduct}
    \\
    $\bbS$ & Matrix of confounder sufficient statistics, $\in \R^{N \times p}$  & \cref{sec:poet_reduct}
    \\
    $\bbU$ & $\in \R^{N \times p}$ & \cref{sec:poet_reduct}
    \\
    \multicolumn{3}{l}{}
    \\
    \multicolumn{3}{l}{\textbf{Function Spaces}}
    \\
    \hhline{===}
    $\hilb$ & Hilbert space containing $X_j$ & \cref{sec:poet_reduct}
    \\
    $M$ & Dimensionality of $\hilb$ & \cref{sec:poet_reduct}
    \\
    $\Phi$ & Feature map on $H$ & \cref{sec:poet_reduct}
    \\
    $\fkspace_j$ & RKHS used to model $f_{ij}$, $1 \leq j \leq p$ & \cref{sec:poet_reduct}
    \\
    $\rkspace_j$ & RKHS used to model $r_j$, $1 \leq j \leq p$ & \cref{sec:poet_reduct}
    \\
\end{tabular}
\end{center}

\section{Proofs} \label{A:proofs} 

\subsection{\edit{The canonical exogenous DAG}} \label{A:canon_exog_dag}

We provide a brief argument to verify that our direct definition of the canonical exogenous DAG is equal to the canonical DAG of the latent projection defined by \citet{exog_dag}.
Our direct definition allows us to avoid the concepts of simplicial complexes and facets that are necessary to define the latent projection; see \citet{exog_dag} for details.

\begin{nprop}
Let $\Gnat$ be a DAG over observed nodes $X$ and unobserved nodes $\bar{H}$.
Let $\cG_1$ be the canonical exogenous DAG for $\Gnat$.
Let $\cG'$ be the latent projection of $\cG$ onto $X$, and let $\cG_2$ be the canonical DAG associated with $\cG'$.
Then $\cG_1 = \cG_2$.
\end{nprop}

\begin{proof}
By definition of the latent projection, $\cG'$ contains $X_i \to X_j$ if and only if $X_i$ has a completely hidden path to $X_j$.
Furthermore, by definition of the canonical DAG, $\cG_2$ contains an edge $X_i \to X_j$ if and only if $\cG'$ contains the edge $X_i \to X_j$.
Thus, $\cG_1$ and $\cG_2$ have the same set of edges over $X$.

Next, $\cG'$ contains the bidirected face $C$ if and only if the nodes $X_C$ share a hidden common cause.
Furthermore, $\cG_2$ contains the node $H_k$ if and only if $C_k$ is a bidirected facet in $\cG'$, i.e., if $X_C$ is a maximal set of nodes with a hidden common cause.
In this case, $\cG_2$ has the edge $H_k$ to $X_i$ if and only if $i \in C$.
Thus, $\cG_1$ and $\cG_2$ have the same set of edges from $H$ to $X$.

By both constructions, there are no edges from $X$ to $H$ or between nodes in $H$, thus, $\cG_1$ and $\cG_2$ are equal.
\end{proof}

\subsection{Proof of \cref{lem:linear_s}}
\begin{proof}
By \cref{eq:linear_sem}, 
\begin{equation}
X = (I - B)^{-1} \epsilon + (I - B)^{-1} \Theta H.
\end{equation}
Hence, for all $1 \leq j \leq p$,
\begin{equation}
	\begin{split}
		X_j 
        &= 
        \sum_{i \in \Parents{\Gtrue}{j}} B_{ij} X_i + \Theta_i^T H 
        \\
		&= 
        \sum_{i \in \Parents{\Gtrue}{j}} B_{ij} [(I - B)^{-1} \epsilon + (I - B)^{-1} \Theta H]_i     + \Theta_j^T H.
	\end{split}
\end{equation}
Taking expectations with respect to $H$,
\begin{equation} \label{eq:linear_s}
	\begin{split}
		S_j &= \E[X_j \mid H] \\
		&= 
        \E \left[\sum_{i \in \Parents{\Gtrue}{j}} B_{ij} [(I - B)^{-1} \epsilon + (I - B)^{-1} \Theta H]_i    + \Theta_j^T H \mid H \right] 
        \\
		&= 
        \sum_{i \in \Parents{\Gtrue}{j}} \E \left[ B_{ij} [(I - B)^{-1} \epsilon + (I - B)^{-1} \Theta H]_i  \mid H \right] + \Theta_j^T H 
        \\
		&= 
        \sum_{i \in \Parents{\Gtrue}{j}} [(I - B)^{-1} \Theta H]_i   + \Theta_j^T H \qquad (\text{since $\epsilon \indep H$}),
	\end{split}
\end{equation}
where $[a]_i$ denotes the $i$th component of a vector $a$. 
Let $\Theta_{j, :}$ denote the $j$th row of $\Theta$. 
Recalling that $D_j = \Theta_{j, \cdot}^T\, H = S_j - r_j(S_1, \cdots, S_{j-1})$, 
\begin{equation} \label{eq:linear_confound_parents}
	\begin{split}
	    r_j(S_1, \cdots, S_{j-1}) 
        &=  
        \sum_{i \in \Parents{\Gtrue}{j}} B_{ij}  [(I - B)^{-1} \Theta H]_i. 
        \\
		&= \sum_{i \in \Parents{\Gtrue}{j}} B_{ij}  S_i \quad \text{(by \cref{eq:linear_s})}.
	\end{split}
\end{equation}
Therefore, the SCM in \cref{eq:linear_sem} can be re-written as
\begin{equation}
	(X_j - S_j) = \sum_{i \in \Parents{\Gtrue}{j}} B_{ij} (X_i - S_i) + \epsilon_j.
\end{equation}
%
\end{proof}

\subsection{Proof of \cref{thm:superdag}}

\begin{nlem}
For every $1 \leq j \leq p$, there exists an $f_j$ such that 
\begin{equation} \label{eq:fn_of_d}
	X_j = \epsilon_j + D_j + f_j(\epsilon_1, D_1, \cdots, \epsilon_{j-1}, D_{j-1})
\end{equation}
for the SEM in \cref{eq:cam_gen_model}.
\end{nlem}
\begin{proof}
The proof follows by inducting on the number of nodes in $\Gtrue$. 
For $p=1$, $X_1 = D_1 + \epsilon_2$, and \cref{eq:fn_of_d} trivially holds. 
For $p=2$, 
\begin{equation}
	\begin{split}
		X_2 
        &=  
        D_2 + \epsilon_2 + f_{12}(X_1) \\
		&= 
        D_2 + \epsilon_2 + f_{12}(D_1 + \epsilon_1).
	\end{split}
\end{equation}
The claim holds by setting $f_2(\epsilon_1, d_1) = f_{12}(d_1 + \epsilon_1)$. Suppose that \cref{eq:fn_of_d} holds for all DAGs with at most $p-1$ nodes. Then it suffices to prove that \cref{eq:fn_of_d} holds for all DAGs with $p$ nodes. For this, note that
\begin{equation}
	\begin{split}
			X_p 
            &= 
            D_p + \epsilon_p + \sum_{i \in \Parents{\Gtrue}{p}} f_{ij}(X_i) 
            \\
			&= 
            D_p + \epsilon_p + \sum_{i \in \Parents{\Gtrue}{p}} f_{ij}(\epsilon_i + D_i + f_j(\epsilon_1, D_1, \cdots, \epsilon_{i-1}, D_{i-1})),
	\end{split}
\end{equation}
where the last line follows from the inductive hypothesis. Thus, by setting
\begin{equation*}
	f_p(\epsilon_1, d_1, \cdots, \epsilon_{j-1}, d_{p-1}) = \sum_{i \in \Parents{\Gtrue}{p}} f_{ij}(\epsilon_i + d_i + f_j(\epsilon_1, d_1, \cdots, \epsilon_{i-1}, d_{i-1})),
\end{equation*}
the claim follows.
\end{proof}
\begin{ncor} \label{cor:suff_dir_confound} For an SEM in the form of \cref{eq:cam_gen_model}, it holds that 
\begin{equation*}
	\E[X_j \mid H] = \E[X_j \mid D_1, \cdots, D_j].
\end{equation*}
\end{ncor}
\noindent We prove \cref{thm:superdag} below using \cref{cor:suff_dir_confound}.
\begin{proof}
\noindent By \cref{eq:cam_gen_model}, it suffices to show that there exists an $r_j$ such that $D_j = S_j - r_j(S_1, \cdots, S_{j-1})$ for every $1 \leq j \leq p$. 
We prove this claim by inducting on the number of nodes in $\Gtrue$. 
For $p=1$, $X_1 = D_1 + \epsilon_1$. 
Since $H \indep \epsilon$, $S_1 = D_1$, and the claim holds by setting $R_1 = 0$. 
For $p=2$, $X_2 = \epsilon_2 + D_2 + f_{12}(X_1)$, where $f_{12}$ may equal $0$ if $X_1$ is not a parent of $X_2$. 
Then,
\begin{equation}
	\begin{split}
		S_2 &= \E[X_2 \mid H] 
        \\
		&= 
        \E[X_2 \mid D_1, D_2] \quad \text{(by \cref{cor:suff_dir_confound}}) 
        \\
		&= 
        \E[D_2 + \epsilon_2 + f_{12}(X_1) \mid D_1, D_2] 
        \\
		&= D_2 + \E[ f_{12}(X_1) \mid D_1] 
        \\
		&= D_2 + \E[ f_{12}(X_1) \mid S_1] \quad \text{(since $S_1 = D_1$)}.
	\end{split}
\end{equation}
Hence, $D_2 = S_2 - \E[ f_{12}(X_1) \mid S_1]$. 
The claim holds by setting $r_2(S_1) = \E[ f_{12}(X_1) \mid S_1]$. 
Suppose that $D_j = S_j - r_j(S_1, \ldots, S_{j-1})$ for all SEMs in the form of \cref{eq:cam_gen_model} with at most $p-1$ nodes. It suffices to show that that there exists an $r_p$ such that $d_p = s_p - r_p(s_1, \ldots, s_{p-1})$ for an arbitrary SEM in the form of \cref{eq:cam_gen_model} with $p$ nodes.

For this, consider the subgraph formed from $x_1, \ldots, x_{p-1}$. 
%
Since $X_p$ is a sink node, $\pr(X_1, \ldots, X_{p-1})$ factorizes according to a DAG. Hence, by the inductive hypothesis, there exists $\{r_j\}_{j=1}^{p-1}$ such that 
\begin{equation} \label{eq:induct_hyp}
	D_j = S_j - r_j(S_1, \ldots, S_{j-1}) \quad \forall j = 1, \ldots, p-1.
\end{equation}
Now, note that
\begin{equation}
	\begin{split}
		S_p 
        &= \E[X_p \mid H] 
        \\
		&= \E[X_p \mid D_1, D_2, \ldots, D_p] \quad \text{(by \cref{cor:suff_dir_confound}}) 
        \\
		&= \E \left[ D_p + \epsilon_p + \sum_{i \in \Parents{\Gtrue}{p}} f_{ij}(X_i) \mid D_1, D_2, \ldots, D_p \right] 
        \\
		&= D_p + \E \left[ \sum_{i \in \Parents{\Gtrue}{p}} f_{ij}(X_i) \mid D_1, D_2, \cdots, D_{p-1} \right] 
        \\
		&= D_p + \E \left[ \sum_{i \in \Parents{\Gtrue}{p}} f_{ij}(X_i) \mid \{ S_i - r_i(S_1, \cdots, S_{i-1})\}_{i=1}^{p-1} \right] \quad \text{(by \cref{eq:induct_hyp})}.
	\end{split}
\end{equation}
Thus, by setting 
\begin{equation} \label{eq:r_decomp}
r_p(s_1, \cdots, s_{p-1})
= 
\E \left[ \sum_{x_i \in \Parents{\Gtrue}{p}} f_{ij}(X_i) \mid \{s_i - r_i(s_1, \cdots, s_{i-1})\}_{i=1}^{p-1} \right],
\end{equation}
the result follows.
\end{proof}

\subsection{Proof of \cref{thm:pcss_convg}} \label{A:proof_spect_deconfound}
\raj{We use a subset of the technical assumptions from Section 3 of \citet{poet} to prove the result. 
We let $\Sigma_u = \cov(U) \in \R^{p \times p}$, where $U$ is defined in \cref{eq:poet_reduction}. 
In this section, we work with the joint distribution over $N$ samples.
In particular, we consider a sequence of random vectors $(X^{(1)}, H^{(1)}), \ldots, (X^{(N)}, H^{(N)})$, whose joint distribution is the product distribution over $N$ copies of $\pr(X, H)$.
\begin{nassum} \label{asum:3.2iii} (Assumption 2(c) from \citet{poet})
There exist constants $r_1, r_2 > 0$ and $b_1, b_2 > 0$, such that for any $q > 0$, \raedit{$1 \leq j \leq p$ and $1 \leq m \leq M$},
\begin{equation*}
    \pr(|U_j| > q) 
    \leq 
    \exp\{(-q / b_1)^{r_1}\}, 
    \qquad 
    \pr(|\phi_m(H)| > q) 
    \leq 
    \exp\{(-q / b_2)^{r_2}\}
\end{equation*}
\end{nassum}
As discussed in \citet{poet}, \cref{asum:3.2iii} requires exponential-type tails for the purposes of applying large deviation theory. \citet{poet} also states, in their working paper, that $\|\Sigma_u\|_1 < c_2$ in \cref{asum:3.2ii} can be weakened to just requiring that $\lambda_{max}(\Sigma_u) < c_2$ if the number of factors $J$ to use in PCA is known. 
The final assumption below bounds the strength of the latent factor loadings in $\Psi$, and the convergence of $[u^{(n)}]^T u^{(t)}$ towards its expectation.
\begin{nassum} \label{assum:3.4} (Assumption 4 from \citet{poet}) There exists a constant $C > 0$ such that for all $1 \leq j \leq p$, $1 \leq n, t \leq N$:
\begin{enumerate}
    \item $\max_{1 \leq m \leq M}  [\psi_j]_m < C$,
    \item $\E \left[p^{-1/2} \left([U^{(n)}]^T U^{(t)} - \E[[U^{(n)}]^T U^{(t)}] \right) \right]^4 < C$
    \item $\E\| p^{-1/2} \sum_{j=1}^p \Psi_j U_{j}^{(n)} \|^4 < M$.
\end{enumerate}
\end{nassum}
By Corollary 1 of \citet{poet}, it suffices to show that Assumption 2(a) and Assumption 3 in \citet{poet} hold. For $1 \leq j \leq p$,
\begin{equation*}
\begin{split}
    \E[U_j] &=  \E[X_j - \E[X_j \mid H]] \\
                &=\E[X_j] - \E[ \E[X_j \mid H]] \\
                &= 0.
\end{split}
\end{equation*}
By properties of conditional expectation, $U_j$ belongs to the orthogonal complement of $\hilb_M$. 
Hence, $\E[U_j \phi_m(H)] = 0$ for all $1 \leq j \leq p$. 
Thus, Assumption 2(a) holds. Since $\{(x^{(n)}, h^{(n)})\}_{n=1}^N$ are drawn iid from $\pr$, then $(X^{(n)}, H^{(n)})$ is independent of $(X^{(t)}, H^{(t)})$ when $n \neq t$. 
Hence, the strong-mixing coefficient defined in Equation 3.1 of \citet{poet} equals 0. Therefore, Assumption 3 in \cite{poet} indeed holds.
}

\subsection{Proof of \cref{eq:gp_gen_formula}}

We follow \citet{gp_nets} to compute the marginal likelihood. By \cref{thm:superdag}, the marginal likelihood decomposes as 
\begin{equation*}
    \begin{split}
        \pr(\bbX \mid \cG, \bbS, \sigma^2) 
        &= 
        \int \pr(\bbX \mid \cG, \bbS, \Omega_\cG, \sigma^2) d\pr(\Omega_\cG \mid \cG) 
        \\
        &= \int \prod_{j=1}^p \pr(\bbX_{:,j} - \bbS_{:,j} \mid \bbX_{\Parents{\cG}{j}}, \bbS, \{f_{ij} \}_{i \in \Parents{\cG}{j}}, r_j, \sigma^2) d\pr( \{f_{ij} \}_{i \in \Parents{\cG}{j}}, r_j) 
        \\
        &= \prod_{j=1}^p \int \pr(\bbX_{:,j} - \bbS_{:,j} \mid \bbX_{\Parents{\cG}{j}}, \bbS, \{f_{ij} \}_{i \in \Parents{\cG}{j}}, r_j, \sigma^2) d\pr( \{f_{ij} \}_{i \in \Parents{\cG}{j}}, r_j) 
        \\
        &= \prod_{j=1}^p \pr(\bbX_{:,j}- \bbS_{:,j} \mid \bbX_{\Parents{\cG}{j}}, \bbS, \sigma^2).
    \end{split}
\end{equation*}
Hence, 
\begin{equation*}
    \log \pr(\bbX \mid \cG, \bbS, \sigma^2) = \sum_{j=1}^p \log \pr(\bbX_{:,j}- \bbS_{:,j} \mid \bbX_{\Parents{\cG}{j}}, \bbS, \sigma^2).
\end{equation*}
The proof now follows from Equation 2.30 of \citet{gp_book}.

\subsection{\edit{Proof of \cref{thm:decam_consis}}} \label{A:proof_consis}
\noindent \textbf{Notation.}
Throughout we let $\| \ba \|_N^2 = \frac{1}{N} \sum_{n=1}^N (a^{(n)})^2$ denote the empirical norm of a vector $\ba \in \R^N$, and $\pr_N$ the empirical distribution of $\{(x^{(n)}, s^{(n)}) \}_{n=1}^N$. 
For simplicity of notation, we assume that, for all $1 \leq j \leq p$ and $i \in \Parents{\cG}{j}$, we have $\fkspace_{ij} = \fkspace_j$ for some $\fkspace_j$.
We denote the maximum likelihood estimate (MLE) of the noise variances and unknown functions in $\Omega_G$ when using the confounder sufficient statistics $\hat{s}$ estimated from \cref{algo:pcss_est} as 
\begin{equation} \label{eq:mle_emp_poet}
    \begin{split}
    & (\hat{\sigma}_j^\cG)^2 \coloneqq \left\| x_j -  \sum_{i \in \Parents{\cG}{j}} \hat{f}_{ij}^\cG(x_i) - \hat{s}_j + \hat{r}_j^\cG(\hat{s}_{\Parents{\cG}{j}}) \right\|_N^2 \\
    \{ \hat{f}_{ij}^\cG \}_{i \in  \Parents{\cG}{j}}, \hat{r}_j^\cG = &\argmin_{\ g_{ij} \in \fkspace_j, \ w_j \in \rkspace_j^\cG}  \left\| x_j - \sum_{i \in  \Parents{\cG}{j}}g_{ij}(x_i) - \hat{s}_j - w_j(\hat{s}_{\Parents{\cG}{j}})  \right \|^2_N
    \end{split}
\end{equation}
for $1 \leq j \leq p$. 
When the true confounder sufficient statistics $s$ are used instead of $\hat{s}$, we let
\begin{equation} \label{eq:mle_emp_exact}
    \begin{split}
    & (\bar{\sigma}_j^\cG)^2 \coloneqq \left\| x_j - \sum_{i \in \Parents{\cG}{j}} \bar{f}_{ij}^\cG(x_i) - s_j + \bar{r}_j^\cG(s_{\Parents{\cG}{j}}) \right\|_N^2 \\
    \{ \bar{f}_{ij}^\cG \}_{i \in  \Parents{\cG}{j}}, \bar{r}_j^\cG = &\argmin_{\ g_{ij} \in \fkspace_j, \ m_j \in \rkspace_j^\cG}  \left\| x_j - \sum_{i \in  \Parents{\cG}{j}} g_{ij}(x_i) - s_j - m_j(s_{\Parents{\cG}{j}})  \right\|^2_N.
    \end{split}
\end{equation}
denote the corresponding maximum likelihood estimates for $1 \leq j \leq p$. Finally, for ease of notation later on, we let
\begin{equation} \label{eq:def_h}
    \begin{split}
    h^\cG_j(x_{\Parents{\cG}{j}}, s_{\Parents{\cG}{j}}) &=  \sum_{i \in \Parents{\cG}{j}} f_{ij}^\cG(x_i) - s_j + r_j^\cG(s_{\Parents{\cG}{j}}) \\
     \hat{h}^\cG_j(x_{\Parents{\cG}{j}}, \hat{s}_{\Parents{\cG}{j}}) &=  \sum_{i \in \Parents{\cG}{j}} \hat{f}_{ij}^\cG(x_i) - \hat{s}_j + \hat{r}_j^\cG(\hat{s}_{\Parents{\cG}{j}}) \\
      \bar{h}^\cG_j(x_{\Parents{\cG}{j}}, s_{\Parents{\cG}{j}}) &= \sum_{i \in \Parents{\cG}{j}} \bar{f}_{ij}^\cG(x_i) - s_j + \bar{r}_j^\cG(s_{\Parents{\cG}{j}}).
    \end{split}
\end{equation}
for $1 \leq j \leq p$. To establish a uniform law of large numbers, we will make use of the following result from \citet{wainwright_2019}.
\begin{nlem} \label{lem:max_subgaus} (Exercise 2.12 in \citet{wainwright_2019}) Let $\{ X_i \}_{i=1}^n$ be a sequence of zero-mean random variables, each sub-Gaussian with parameter $\sigma$. Then, for $\delta > 0$,
\begin{equation*}
    \mathbb{P}\left(\max_{i=1, \cdots, n} |X_i| \geq 2\sigma\sqrt{\log 2 n} + \delta\right) \leq 2\exp\left(-\frac{\delta^2}{2\sigma^2}\right).
\end{equation*}

\end{nlem}

\noindent \textbf{Discussion of Additional Assumptions.} \begin{nassum} \label{assum:sparsity}
(sparsity condition) There exists a constant $\gamma < \infty$ such that $\Gtrue \in \dagset$ for all $p$, where $\dagset = \{\cG: |\Parents{\cG}{i}| \leq \gamma \}$.
\end{nassum}
\begin{nassum} (Function class behavior) \label{assum:complex} Let $\fkspace_j^\cG = \bigoplus_{i \in \Parents{\cG}{i}} \fkspace_i$, $1 \leq j \leq p$, and assume all constants $\{C_i\}_{i=1}^6$ below do not depend on $p$ and are less than $\infty$.
\begin{enumerate}
    \item There exists a constant $C_1$ such that for all $G \in \dagset$, $\mathrm{rank}(\fkspace_j) \leq C_1$ and $\mathrm{rank}(\rkspace_j^\cG) \leq C_1$ for all $1 \leq j \leq p$ and $G \in \dagset$.
    \item There exists a constant $C_2$ such that
    \begin{equation*}
        \max_{G \in \dagset}  H\left(u, \fkspace_j^\cG \bigoplus  \rkspace_j^\cG, L^1(\pr) \right) < C_2, \quad \max_{G \in \dagset}  H\left(u, \fkspace_j^\cG \bigoplus  \rkspace_j^\cG, L^4(\pr) \right) < C_2,
    \end{equation*}
    for all $1 \leq j \leq p$ and $u > 0$, where $H$ is the log-bracketing number of the function class with respect to the $L_1$ and $L_4$ norms, and $u > 0$ denotes the largest distance between any two functions in the bracketing of the function class.
    \item There exists a constant $C_3$ such that for $1 \leq j \leq p$ the following quantities are all bounded above by $C_3$:
    \begin{equation*}
        \E|X_j|^4, \quad  \E|S_j|^4, \quad \sup_{f \in \fkspace_j} \E|f(X_j)|^4, \quad  \max_{\cG \in \dagset} \sup_{r \in \rkspace_j^\cG} \E|r(S_{\Parents{\cG}{j}})|^4.
    \end{equation*}
    \item There exists constants $C_5$ and $C_6$ such that
    \begin{equation*}
        \max_{1 \leq j \leq p} 
        \max_{\cG \in \dagset} 
        \sup_{v_j^\cG \in \fkspace_j^\cG \bigoplus \rkspace_j^\cG} 
        \rho(m_j^\cG) 
        \leq C_4, \ \rho(M_j^\cG) = 2C_5\E_{\pr}[\exp(|M_j^\cG|/C_5) - 1 - |M_j^\cG|/C_5], 
    \end{equation*}
    where $M_j^\cG = (X_j - S_j - v_j^\cG(X_{\Parents{\cG}{j}}, S_{\Parents{\cG}{j}}))^2$.
    \item There exists a constant $C_5$ such that $\| \nabla r \|_2 \leq C_6$ for all $r \in \rkspace_j$, $1 \leq j \leq p$.
    \item There exists constants $0 < C_7, C_8 < \infty$ such that 
    \begin{equation*}
        \pr(||S_j| - \E[|S_j|]| > t) \leq C_7\exp(-C_8t^2),
    \end{equation*}
    for all $1 \leq j \leq p$.
\end{enumerate}
\end{nassum}
\begin{nassum} \label{assum:var_pos} There exists a constant $C_4$ such that
$\min_{\cG \in \dagset} \min_{1 \leq j \leq p}(\sigma_j^\cG)^2 \geq C_4 > 0$.
\end{nassum}
In \citet{cam}, the authors assume that the size of the neighborhood set of any observed node is uniformly bounded. Since the true parent set of a node is contained in its neighborhood set, \cref{assum:sparsity} is a weaker condition.
\cref{assum:gap} and \cref{assum:complex}(a) together imply that the true DAG still has the smallest negative log-likelihood, even in the misspecified setting when $\fkspace_j \subsetneq \fspace_j$. As shown in \citet{cam}, if $\fkspace_j $ is sufficiently close to $\fspace_j$ (i.e., by setting $C_1$ large enough), then assuming $\Gtrue$ still has the smallest negative log-likelihood over functions in $\fkspace_j$ is weak. \cref{assum:complex}(b) is the analogue of Assumption B3(ii) in \citet{cam}. 
Assumption B3(ii) in \citet{cam} is a smoothness assumption on the function class) used to prove that the log-bracketing number with respect to the $L^4$ norm is bounded. 
If we make the same smoothness assumption, then \cref{assum:complex}(b) holds by following the proof of Theorem 3 in \citet{cam}. 
\cref{assum:complex}(c) and \cref{assum:complex}(d) (which assumes $m_j^\cG$ has exponential moments) are also used in \citet{cam} to establish a uniform law of large numbers. \cref{assum:complex}(e) assumes that the function class $\rkspace_j$ has bounded gradients, and \cref{assum:complex}(f) assumes that the confounder sufficient statistics are sub-Gaussian. 
\cref{assum:var_pos} is a technical regularity assumption that ensures a log-likelihood exists (since the log of zero is undefined). \\

\noindent \textbf{Proof of \cref{thm:decam_consis}.} A key difference between our consistency proof and the proof in \citet{cam} is accounting for the additional estimation error of estimating the confounder sufficient statistics $s$ via PCA. 
Once we bound the error in empirical log-likelihood scores when using $\hat{s}$ instead of $s$, our proof closely follows the proof of Theorem 3 in  \citet{cam}. 
Before we can follow the proof technique in \citet{cam}, we show below that it suffices to prove the consistency of the MLE (a frequentist-based score) to prove consistency of the DeCAMFounder (a Bayesian-based score).  \\ 

\noindent \textit{Reduction to Proving the Consistency of the MLE.} The DeCAMFounder score of a DAG equals the sum of the marginal log-likelihood scores of each node. 
Under \cref{assum:complex}(a), each observed node belongs to a finite dimensional exponential family. 
Hence, by \citet{bic_consis}, the marginal log-likelihood of the parent set $\Parents{\cG}{j}$ for node $X_j$ converges to the negative BIC score of $\Parents{\cG}{j}$. 
The BIC score equals the sum of empirical log-likelihood score associated with the MLE parameters in \cref{eq:mle_emp_poet}, and a penalty on model complexity which equals $\log N$ times the number of parameters. 
When a DAG $\cG$ has a different partial ordering than $\Gtrue$, then the difference in expected log-likelihood between $\cG$ and $\Gtrue$ is non-zero. 
Hence, the empirical log-likelihood scores grows as $O(N)$ while the penalty grows as $O(\log N )$. 
Since the  log-likelihood score in BIC grows faster than the penalty, it suffices to show that minimizing the empirical negative log-likelihood of a DAG yields a DAG in $\Pi^*$ to prove order consistency of the DeCAMFounder. \\

\noindent \textit{Inequalities Sufficient for Proving Consistency.} As we show below, if there exists a sufficient "gap" $\Delta_N$ (to be defined shortly) in  empirical log-likelihood scores between the best incorrect DAG and the worst super-DAG of $\Gtrue$, then the DeCAMFounder converges to a DAG with a correct ordering asymptotically. 

First notice each DAG in $\Pi_*$ is a super-DAG of $\Gtrue$. 
Consequently, $\sigma^\cG_j = \sigma^{\Gtrue}_j$ for all $\cG \in \Pi_*$, $1 \leq j \leq p$. 
Suppose there exists a $\Delta_N$ such that,
\begin{equation} \label{eq:first_ineq}
    \sum_{j=1}^p \log\left( \hat{\sigma}_j^\cG \right) 
    \geq  
    \sum_{j=1}^p \log\left( \sigma_j^\cG \right) - \Delta_N, \ \forall~\cG \notin \Pi_*,
\end{equation}
\begin{equation} \label{eq:second_ineq}
    \sum_{j=1}^p \log\left( \hat{\sigma}_j^\cG \right) 
    \leq 
    \sum_{j=1}^p \log\left( \sigma_j^{\Gtrue} \right) + \Delta_N, \ \forall~\cG \in \Pi_*,
\end{equation}
where $\Delta_N/p = o_p(\xi_p)$. Then, for all $\cG \notin \Pi_*$ and $\cG^{\prime} \in \Pi_*$,
\begin{equation*}
    \begin{split}
      \sum_{j=1}^p \log \left( \hat{\sigma}_j^\cG \right) 
      - 
      \sum_{j=1}^p \log \left( \hat{\sigma}_j^{\cG^{\prime}} \right) 
      &\geq 
      \sum_{j=1}^p \log \left( \sigma_j^{\cG} \right) - \Delta_N 
      - 
      \sum_{j=1}^p \log \left( \hat{\sigma}_j^{\cG^{\prime}} \right) 
      \\
      &\geq \sum_{j=1}^p \log \left( \sigma_j^\cG \right) - \Delta_N - \sum_{j=1}^p \log \left( \sigma_j^{\Gtrue} \right) - \Delta_N 
      \\
      &\geq p\xi_p - 2\Delta_N 
      \\
      & = p(\xi_p - 2\Delta_N/p).
    \end{split}
\end{equation*}
Since $\Delta_N/p = o_p(\xi_p)$ and $\xi_p > 0$ by \cref{assum:gap}, with probability tending towards one, $\log(\hat{\sigma}_j^\cG) > \sum_{j=1}^p \log(\hat{\sigma}_j^{\cG^{\prime}})$. 
Hence, $\mathbb{P}(\hat{\cG} \in \Pi_*) \rightarrow 1$ as $p, N \rightarrow \infty$ as desired. \\

\noindent \textit{Definition of the Gap Parameter $\Delta_N$.}
\begin{equation}
	\begin{split}
	\Delta_N &= \sum_{j=1}^p \Delta_{N, j}, \quad \text{s.t.} \quad \Delta_{N, j} =  \max_{\cG \in \dagset} \nu_{N, j}^\cG + \eta_{N, j}^\cG, 
    \\
	& \nu_{N, j}^\cG 
    = 
    \left|\frac{1}{N}\sum_{n=1}^N 
    \left( x_j^{(n)} - \bar{h}^\cG_j \left( x_{\Parents{\cG}{j}}^{(n)}, {s}_{\Parents{\cG}{j}}^{(n)}\right) \right)^2 
    - 
    \left( x_j^{(n)} - \hat{h}^\cG_j \left( x_{\Parents{\cG}{j}}^{(n)}, \hat{s}_{\Parents{\cG}{j}}^{(n)} \right) \right)^2 \right| 
    \\
	&  \eta_{N, j}^\cG =  \sup_{v_j^\cG \in \fkspace_j^\cG \bigoplus \rkspace_j^\cG} \left|(\pr_N - \pr)(m_j^\cG)\right|,
	\end{split}
\end{equation}
where
\begin{equation} \label{eq:def_m}
	\begin{split}
	v_j^\cG &\in \fkspace_j^\cG \bigoplus \rkspace_j^\cG, 
    \\
    m_j^\cG \left( x_{\Parents{\cG}{j}}, s_{\Parents{\cG}{j}} \right) 
    &= 
    \left( x_j - s_j - v_j^\cG \left( x_{\Parents{\cG}{j}}, s_{\Parents{\cG}{j}} \right) \right)^2, 
    \\
    \fkspace_j^\cG 
    &= 
    \bigoplus_{i \in \Parents{\cG}{j}} \fkspace_i.
   \end{split}
\end{equation}

\noindent \textit{Proof of \cref{eq:first_ineq}.} By two applications of the triangle inequality,
\begin{equation*}
\begin{split}
    & \left( \sigma_j^\cG \right)^2 
    = 
    \min_{v_j^\cG \in \fkspace_j^\cG \bigoplus \rkspace_j^\cG} 
    \E \left[
    \left(X_j - S_j - v^\cG_j \left( X_{\Parents{\cG}{j}}, S_{\Parents{\cG}{j}} \right)\right)^2 \right] 
    \\
    &\leq 
    \min_{v_j^\cG \in \fkspace_j^\cG \bigoplus \rkspace_j^\cG} 
    \left\| x_j - s_j - v^\cG_j \left( x_{\Parents{\cG}{j}}, s_{\Parents{\cG}{j}} \right) \right\|_N^2 
    + 
    \sup_{v_j^\cG \in \fkspace_j^\cG \bigoplus \rkspace_j^\cG} 
    \left| (\pr_N - \pr)(m_j^\cG) \right| 
    \\
    &\leq 
    \min_{v_j^\cG \in \fkspace_j^\cG \bigoplus \rkspace_j^\cG} 
    \left\| x_j - \hat{s}_j - v^\cG_j \left( x_{\Parents{\cG}{j}}, \hat{s}_{\Parents{\cG}{j}} \right) \right\|_N^2 
    + 
    \nu_{N, j}^\cG + \eta_{N, j}^\cG 
    \\
    & = \left( \hat{\sigma}_j^\cG \right)^2 + \nu_{N, j}^\cG + \eta_{N, j}^\cG.
\end{split}
\end{equation*}
 \cref{eq:first_ineq} now follows from applying a Taylor series expansion and \cref{assum:var_pos}. \\

\noindent \textit{Proof of \cref{eq:second_ineq}.} 

\begin{equation}
\begin{split}
    (\hat{\sigma}_j^\cG)^2 
    &= 
    \min_{v_j^\cG \in  \fkspace_j^\cG \bigoplus \rkspace_j^\cG} 
    \left\| x_j - \hat{s}_j - v^\cG_j \left(x_{\Parents{\cG}{j}}, \hat{s}_{\Parents{\cG}{j}} \right) \right\|_N^2 
    \\
    &\leq 
    \min_{v_j^\cG \in \fkspace_j^\cG \bigoplus \rkspace_j^\cG} 
    \left\| x_j - s_j - v^\cG_j \left(x_{\Parents{\cG}{j}}, s_{\Parents{\cG}{j}} \right) \right\|_N^2 + \nu_{N, j}^\cG 
    \\
    &\leq 
    \min_{v_j^\cG \in   \fkspace_j^\cG \bigoplus \rkspace_j^\cG} 
    \E\left[\left( X_j - S_j - v^\cG_j \left(X_{\Parents{\cG}{j}}, S_{\Parents{\cG}{j}} \right) \right)^2 \right] 
    + 
    \nu_{N, j}^\cG + \eta_{N, j}^\cG 
    \\
    &=  
    (\sigma_j^\cG)^2 + \nu_{N, j}^\cG + \eta_{N, j}^\cG.
\end{split}
\end{equation}
 \cref{eq:second_ineq} now follows from applying a Taylor series expansion and \cref{assum:var_pos}. 
 To complete the proof, it suffices to show that both $\max_{\cG \in \dagset} \max_j \nu_{N, j}^\cG$ and $\max_{\cG \in \dagset} \max_j \eta_{N, j}^\cG$ are $O_p\left(\log(N)^{1/c} \sqrt{\frac{\log p}{N}} + \frac{N^{\frac{1}{4}}}{\sqrt{p}} \right)$. \\

\noindent \textit{Bound on $\nu_{N, j}^\cG$.} For $a \in \R^p$, define the residual sum of squares (RSS) functions as,
\begin{equation*}
	\begin{split}
		\widehat{RSS}^\cG_j(a) 
        &\coloneqq 
        \left\| x_j - \hat{h}_j^\cG \left( x_{\Parents{\cG}{j}}, a_{\Parents{\cG}{j}} \right)  \right\|_N^2 
        \\
		\overbar{RSS}^\cG_j(a) 
        &\coloneqq 
        \left\| x_j - \bar{h}_j^\cG \left( x_{\Parents{\cG}{j}}, a_{\Parents{\cG}{j}} \right)  \right\|_N^2.
	\end{split}
\end{equation*}
Then, by definition of $\hat{h}_j^\cG$ and $\bar{h}_j^\cG$ in \cref{eq:def_h}, 
\begin{equation}
	\widehat{RSS}_j^\cG(\hat{s}) \leq \overbar{RSS}^\cG_j(\hat{s}) \quad \text{and} \quad \widehat{RSS}^\cG_j(s) \geq \overbar{RSS}^\cG_j(s).
\end{equation}
\begin{equation}
    \begin{split}
        \widehat{RSS}^\cG_j(\hat{s}) - \overbar{RSS}^\cG_j(s) &= \underbrace{\left(\widehat{RSS}^\cG_j(\hat{s}) - \overbar{RSS}^\cG_j(\hat{s})\right)}_{\leq 0} + \overbar{RSS}^\cG_j(\hat{s}) - \overbar{RSS}^\cG_j(s) \\
        & \leq \overbar{RSS}^\cG_j(\hat{s}) - \overbar{RSS}^\cG_j(s) \\
        & \leq |\overbar{RSS}^\cG_j(\hat{s}) - \overbar{RSS}^\cG_j(s)|.
    \end{split}
\end{equation}
Similarly, 
\begin{equation}
    \begin{split}
        \overbar{RSS}^\cG_j(s) - \widehat{RSS}^\cG_j(\hat{s})  &= \underbrace{\left(\overbar{RSS}^\cG_j(s) - \widehat{RSS}^\cG_j(s)\right)}_{\leq 0} + \widehat{RSS}^\cG_j(s) - \widehat{RSS}^\cG_j(\hat{s}) \\
        & \leq \widehat{RSS}^\cG_j(s) - \widehat{RSS}^\cG_j(\hat{s}) \\
        & \leq |\widehat{RSS}^\cG_j(s) - \widehat{RSS}^\cG_j(\hat{s})|.
    \end{split}
\end{equation}
Hence, 
\begin{equation}
     \left| \widehat{RSS}^\cG_j(\hat{s}) - \overbar{RSS}^\cG_j(s) \right| 
     \leq 
     \left| \overbar{RSS}^\cG_j(\hat{s}) - \overbar{RSS}^\cG_j(s) \right| 
     + 
     \left| \widehat{RSS}^\cG_j(s) - \widehat{RSS}^\cG_j(\hat{s}) \right|.
\end{equation}
We bound $|\overbar{RSS}^\cG_j(\hat{s}) - \overbar{RSS}^\cG_j(s)|$ below. An identical bound works for $|\widehat{RSS}^\cG_j(s) - \widehat{RSS}^\cG_j(\hat{s})|$ by replacing $\overbar{RSS}^\cG$ with $\widehat{RSS}^\cG$ in the proof below. By expanding the square and canceling the common term, the difference  $|\overbar{RSS}^\cG_j(\hat{s}) - \overbar{RSS}^\cG_j(s)|$ equals the sum of two components:
\begin{equation}  \label{eq:sum_two_terms}
    \begin{split}
        |\overbar{RSS}_j^\cG(\hat{s}) - &\overbar{RSS}^\cG_j(s)| = 
        \\ 
        & \bigg| - \frac{1}{N} \sum_{n=1}^N 
        \left(
            \bar{h}_j^\cG \left(x_{\Parents{\cG}{j}}^{(n)}, \hat{s}_{\Parents{\cG}{j}}^{(n)}\right)^2 
            - 
            \bar{h}_j^\cG \left(x_{\Parents{\cG}{j}}^{(n)}, s_{\Parents{\cG}{j}}^{(n)} \right)^2 
        \right) 
        \\
        & \frac{-2}{N} \sum_{n=1}^N 
        x_j^{(n)} \left(
            \bar{h}_j^\cG \left(x_{\Parents{\cG}{j}}^{(n)}, \hat{s}_{\Parents{\cG}{j}}^{(n)} \right) 
            - 
            \bar{h}_j^\cG \left(x_{\Parents{\cG}{j}}^{(n)}, s_{\Parents{\cG}{j}}^{(n)} \right)
        \right) \bigg|. \\
    \end{split}
\end{equation}
We will bound the first term using a Taylor series expansion with respect to $\hat{s}_{\Parents{\cG}{j}}$, and apply \cref{thm:pcss_convg}. 
To this end, $\bar{h}_j^\cG \left( x_{\Parents{\cG}{j}}^{(n)}, \hat{s}_{\Parents{\cG}{j}}^{(n)} \right)^2$ equals
\begin{equation*}
        \bar{h}_j^\cG \left( x_{\Parents{\cG}{j}}^{(n)}, s_{\Parents{\cG}{j}}^{(n)} \right)^2 
        + 
        2 \left[ \bar{h}_j^\cG \left( x_{\Parents{\cG}{j}}^{(n)}, s_{\Parents{\cG}{j}}^{(n)} \right) \nabla_n \bar{r}_j^\cG \right]^T
        \left( \hat{s}_{\Parents{\cG}{j}}^{(n)} - s_{\Parents{\cG}{j}}^{(n)} \right) 
        + 
        O \left( \left\| \hat{s}_{\Parents{\cG}{j}}^{(n)} - s_{\Parents{\cG}{j}}^{(n)} \right\|_2^2 \right),
\end{equation*}
where $\nabla_n \bar{r}_j^\cG \coloneqq  \nabla_{s_{\Parents{\cG}{j}}^{(n)}} \bar{r}_j^\cG = \nabla_{ s_{\Parents{\cG}{j}}^{(n)}} \bar{h}_j^\cG(x_{\Parents{\cG}{j}}^{(n)}, \cdot)$. Hence,  
\begin{equation*}
	\bigg|\frac{1}{N} \sum_{n=1}^N 
    \left( \bar{h}_j^\cG \left(x_{\Parents{\cG}{j}}^{(n)}, \hat{s}_{\Parents{\cG}{j}}^{(n)}\right)^2 
    - 
    \bar{h}_j^\cG \left( x_{\Parents{\cG}{j}}^{(n)}, s_{\Parents{\cG}{j}}^{(n)} \right)^2 \right)\bigg|
\end{equation*}
equals
\begin{equation} \label{eq:cauchy}
\begin{split}
    & \left| \frac{1}{N} \sum_{n=1}^N 2 \left[ \bar{h}_j^\cG \left( x_{\Parents{\cG}{j}}^{(n)}, s_{\Parents{\cG}{j}}^{(n)} \right) \nabla_n \bar{r}_j^\cG \right]^T
    \left( \hat{s}_{\Parents{\cG}{j}}^{(n)} - s_{\Parents{\cG}{j}}^{(n)} \right) 
    + 
    O \left( \left\| \hat{s}_{\Parents{\cG}{j}}^{(n)} - s_{\Parents{\cG}{j}}^{(n)} \right\|_2^2 \right) \right| 
    \\
    &\leq 
    \frac{2}{N} \sum_{n=1}^N 
    \left| 
    \bar{h}_j^\cG \left( x_{\Parents{\cG}{j}}^{(n)}, s_{\Parents{\cG}{j}}^{(n)} \right) 
    \right| \left\| \nabla_n \bar{r}_j^\cG \right\|_2 
    \left\| \hat{s}_{\Parents{\cG}{j}}^{(n)} - s_{\Parents{\cG}{j}}^{(n)} \right\|_2 
    + 
    O \left( \left\| \hat{s}_{\Parents{\cG}{j}}^{(n)} - s_{\Parents{\cG}{j}}^{(n)} \right\|_2^2 \right),
\end{split}
\end{equation}
where the inequality follows from Cauchy-Schwartz. Hence, \cref{eq:cauchy}, \cref{thm:pcss_convg}, and \cref{assum:sparsity} imply that 
\begin{equation} \label{eq:term_one_modify}
	\max_{G \in \dagset} \max_{j \in [p]} 
    \left| 
    \frac{1}{N} \sum_{n=1}^N \bar{h}_j^\cG \left(x_{\Parents{\cG}{j}}^{(n)}, \hat{s}_{\Parents{\cG}{j}}^{(n)} \right)^2 
    - 
    \bar{h}_j^\cG \left( x_{\Parents{\cG}{j}}^{(n)}, s_{\Parents{\cG}{j}}^{(n)} \right)^2 
    \right|
\end{equation}
is 
\begin{equation} \label{eq:prod_bound}
O_p \left( \left[\log(N)^{1/c} \sqrt{\frac{\log p}{N}} + \frac{N^{\frac{1}{4}}}{\sqrt{p}} \right] \right)  
\max_{G \in \dagset} \max_{1 \leq j \leq p} 
\frac{1}{N} \sum_{n=1}^N 
\left| \bar{h}_j^\cG \left( x_{\Parents{\cG}{j}}^{(n)}, s_{\Parents{\cG}{j}}^{(n)} \right) \right|  
\left\| \nabla_n \bar{r}_j^\cG \right\|_2
\end{equation}
since $N = o(p^2)$ and $p = o(\exp(N))$. Since all functions in $\rkspace_j$ have gradients uniformly bounded by a constant (where the constant does not depend on $p$) by \cref{assum:complex}(e), it suffices to bound 
\begin{equation*}
    \begin{split}
      	 &\max_{G \in \dagset} \max_{1 \leq j \leq p}  
         \frac{1}{N} \sum_{n=1}^N 
         \left| \bar{h}_j^\cG \left( x_{\Parents{\cG}{j}}^{(n)}, s_{\Parents{\cG}{j}}^{(n)} \right) \right| 
         \\
      	 &\leq \max_{G \in \dagset} \max_{1 \leq j \leq p} \sup_{v_j^\cG \in \fkspace_j^\cG \bigoplus \rkspace_j^\cG}  
         \frac{1}{N} \sum_{n=1}^N 
         \left( \left|v_j^\cG \left( x_{\Parents{\cG}{j}}^{(n)}, s_{\Parents{\cG}{j}}^{(n)} \right) \right| 
         + 
         \left| s_j^{(n)} \right| \right) 
         \\
      	 &= \max_{G \in \dagset} \max_{1 \leq j \leq p} \sup_{v_j^\cG \in \fkspace_j^\cG \bigoplus \rkspace_j^\cG}  
         \frac{1}{N} \sum_{n=1}^N 
         \left| v_j^\cG \left( x_{\Parents{\cG}{j}}^{(n)}, s_{\Parents{\cG}{j}}^{(n)} \right) \right| 
         + 
         \max_{1 \leq j \leq p} \frac{1}{N} \sum_{n=1}^N \left| s_j^{(n)} \right|.
    \end{split}
\end{equation*}

By \cref{assum:sparsity}, \cref{assum:complex}(c), and Jensen's inequality, $\E|v_j^\cG|$ and $\E[s_j]$ are bounded above by $\sqrt{C_3}(\gamma + 1)$ and $\sqrt{C_3}$, respectively, for all $G \in \dagset$ and $1 \leq j \leq p$. 
Hence, by \cref{assum:complex}(f) and \cref{lem:max_subgaus},
\begin{equation} \label{eq:bounds}
\max_{1 \leq j \leq p} \frac{1}{N} \sum_{n=1}^N \left| s_j^{(n)} \right| 
= 
O_p\left(1 + \sqrt{\frac{\log p}{N}} \right).
\end{equation}
By \cref{assum:sparsity}, \cref{assum:complex}(b), \cref{assum:complex}(d), the uniform inequality used in the proof of Theorem 5.11 in \citet{geer_one}, and \cref{lem:max_subgaus},
\begin{equation} \label{eq:bound2}
\max_{\cG \in \dagset} 
\max_{1 \leq j \leq p} 
\sup_{v_j^\cG \in \fkspace_j^\cG \bigoplus \rkspace_j^\cG} 
\frac{1}{N} \sum_{n=1}^N 
\left| v_j^\cG \left( x_{\pa{x_j}{G}}^{(n)}, s_{\pa{x_j}{G}}^{(n)} \right) \right| 
= 
O_p \left(1 + \sqrt{\frac{\log p}{N}} \right).
\end{equation}
Since $p = o(\exp(N))$, \cref{eq:prod_bound}, \cref{eq:bounds}, and \cref{eq:bound2} imply that \cref{eq:term_one_modify} is 
\begin{equation} \label{eq:common_rate}
O_p\left( \left[\log(N)^{1/c} \sqrt{\frac{\log p}{N}} + \frac{N^{\frac{1}{4}}}{\sqrt{p}} \right] \right).
\end{equation}
By a nearly identical argument, the second term in \cref{eq:sum_two_terms} is also bounded by \cref{eq:common_rate}. Hence, $\max_{\cG \in \dagset} \max_j \nu_{N, j}^\cG$ is $O_p\left(\log(N)^{1/c} \sqrt{\frac{\log p}{N}} + \frac{N^{\frac{1}{4}}}{\sqrt{p}} \right)$. \\

\noindent \textit{Bound on $\eta_{N, j}^\cG$.}
Let the function class
\begin{equation*}
	\mathcal{L}_j^\cG \coloneqq \left\{ m_j^\cG: v_j^\cG \in \fkspace_j^\cG \bigoplus \rkspace_j^\cG \right\},
\end{equation*}
where $m_j^\cG$ is defined in \cref{eq:def_m} for $1 \leq j \leq p$. By the uniform inequality used to prove Theorem 5.11 in \citet{geer_one} and \cref{assum:complex}(d), there exists constants $K_1, K_2, K_3 > 0$ sufficiently large so that
\begin{equation} \label{eq:subgaus_rv}
    \pr\left( 
        \sup_{v_j^\cG \in \fkspace_j^\cG \bigoplus \rkspace_j^\cG} 
        \left| \frac{1}{\sqrt{N}} \sum_{n=1}^N m_j^\cG \left( x^{(n)}, s^{(n)} \right) - \E \left[ m_j^\cG \left( x_{\Parents{\cG}{j}}, s_{\Parents{\cG}{j}} \right) \right] \right| 
        \geq a \right) 
    \leq 
    K_1 \exp\left(\frac{-a^2}{K_2}R^2\right),
\end{equation}
where $a \geq \max(K_3\int_{0}^R H^{1/2}(u, \mathcal{L}_j^\cG, L^2) du , R)$ and $H(u, \mathcal{L}_j^\cG, L^2)$ is the log bracketing number of $\mathcal{L}_j^\cG$ with respect to the $L^2$ norm. 
By \cref{assum:sparsity}, the cardinality of $\dagset$ is upper bounded by $2^{\gamma} p^{\gamma + 1}$. Hence, by \cref{eq:subgaus_rv} and \cref{lem:max_subgaus},
\begin{equation} \label{eq:union_bound_l2}
	\begin{split}
    & \max_{1 \leq j \leq p} 
    \max_{G \in \dagset} 
    \sup_{v_j^\cG \in  \fkspace_j^\cG \bigoplus \rkspace_j^\cG} 
    \left| \frac{1}{N}\sum_{n=1}^N m_j^\cG \left( x^{(n)}, s^{(n)} \right) 
    - 
    \E \left[ m_j^\cG \left( x_{\Parents{\cG}{j}}, s_{\Parents{\cG}{j}} \right) \right] \right| 
    \\
	&= O_p\left((\log p)^{1/2}N^{-1/2}  \max_{1 \leq j \leq p} \max_{G \in \dagset} \int_{0}^R H^{1/2}(u, \mathcal{L}_j^\cG, L^2)du \right).
	\end{split}
\end{equation}
As shown in the proof of Theorem 3 in \citet{cam}, under \cref{assum:complex}(c), there exists a constant $K_4$ such that,
\begin{equation*}
\max_{j \in [p]} \max_{|\Parents{\cG}{j}| \leq \gamma } H^{1/2}(u, \mathcal{L}_j^\cG, L^2) \leq K_4 \max_{j \in [p]} \max_{G \in \dagset}  H\left(u, \left[\bigoplus_{i \in \Parents{\cG}{j}} \fkspace_i\right] \bigoplus  \rkspace_j^\cG, L^4(\pr) \right). 
\end{equation*}
Hence, \cref{assum:complex}(b) and \cref{eq:union_bound_l2} imply that $\max_{G \in \dagset} \max_j \eta_{N, j}^\cG$ is $O_p\left( \sqrt{\frac{\log p}{N}} \right)$ as desired.

\section{\edit{Behavior when $M = \infty$}} \label{A:m_inf_behav} 
\subsection{Limiting Behavior Under Misspecification}
Without loss of generality, assume $\{\phi_m \}_{m=1}^{\infty}$ is an orthonormal basis of $\hilb^*$ so that $\cov(\phi_i(H), \phi_j(H)) = I(i = j)$ (otherwise apply the Gram–Schmidt process). 
Then, 
\begin{equation*}
    \cov(S_i, S_j) = \sum_{m=1}^{\infty} \Psi_{jm}  \Psi_{im}, \quad \text{s.t. } S_k = \sum_{m=1}^{\infty} \Psi_{km} \phi_m(H) \quad (1 \leq k \leq p),
\end{equation*}
\raedit{for $1 \leq i,j \leq p$}. Let $\Psi_{j, :k}$ consist of the first $k$ coefficients and  $\Psi_{j, k:}$ the remaining coefficients for $S_j$. Then,
\begin{equation}
    \cov(S) = \Psi_{:k} \Psi_{:k}^T + \Psi_{k:} \Psi_{k:}^T.
\end{equation}
Suppose the matrix $\frac{1}{p}\cov(S)$ has $\tilde{M}$ non-vanishing eigenvalues as $p \rightarrow \infty$, where $\tilde{M}$ is a constant that does not depend on $p$. Through a change of basis, we may assume without loss of generality that the first $\tilde{M}$ basis functions correspond to the space spanned by the $\tilde{M}$ diverging eigenvectors of $\cov(S)$:
\begin{nassum} \label{assum:mod_eval} There exists constants $0 < K_1, K_2 < \infty$ that do not depend on $p$ such that
\begin{equation}
    \lambda_{\min} \left[ \frac{1}{p}  \Psi_{:\tilde{M}}^T \Psi_{:\tilde{M}} \right] > K_1, \quad \lambda_{\max} \left[ \frac{1}{p} \Psi_{:\tilde{M}}^T \Psi_{:\tilde{M}} \right] < K_2.
\end{equation}
\end{nassum}
Let $\tilde{u}^{(n)} = x^{(n)} -  \Psi_{:\tilde{M}}^T\Phi_{:\tilde{M}}(h^{(n)})$. Then, \raedit{for all $1 \leq j \leq p$, $1 \leq n \leq N$}, $\hat{s}_j^{(n)} \rightarrow \tilde{s}_j^{(n)} := \Psi_{j, :\tilde{M}}^T\Phi_{:\tilde{M}}(h^{(n)})$ \raedit{as $p, N \rightarrow \infty$} at the same rate in \cref{thm:pcss_convg} if we replace $u^{(n)}$ with $\tilde{u}^{(n)}$ in \cref{thm:pcss_convg}. Hence, when $M=\infty$, we can only recover the variation explained by the $\tilde{M}$ components of $\Phi(H)$. So $\tilde{S} \neq S$ in general. \\

\noindent \textbf{Behavior Under Smoothness.}
Suppose $\hilb^*$ satisfies the following smoothness condition: there exists a constant $C < \infty$ such that for any $q \in \hilb^*$,
\begin{equation}
    \sum_{m=T}^{\infty} |\alpha_m| \leq CT^{-\beta} \quad \text{s.t. } q(h) = \sum_{m=1}^{\infty} \alpha_m \phi_m(h),
\end{equation}
where $\beta > 0$ controls the smoothness of functions in $\hilb^*$. 
Then, $\E| S_j - \tilde{S}_j |^2 = O(\tilde{M}^{-\beta + 1 + \kappa})$ for any $\kappa > 0$, $1 \leq j \leq p$. 
Hence, if $\beta$ is large (i.e., when the function class is very smooth) or $\tilde{M}$ is large (i.e., when there exists many basis functions with diverging eigenvalues), then $S_j$ is close to $\tilde{S}_j$. 

\subsection{Modified Condition for Consistent Order Estimation} \label{sec:modified_gap}
Consider the following (modified) model gap parameter $\tilde{\xi}_p$:
\begin{equation} \label{eq:mod_gap}
    \tilde{\xi}_p 
    = 
    p^{-1} \min_{\cG \notin \Pi_*} \sum_{j=1}^p \left( \log \left( \tilde{\sigma}_j^\cG \right) - \log \left( \tilde{\sigma}_j^{\Gtrue} \right) \right),
\end{equation}
where
\begin{equation*}
    \begin{split}
    & (\tilde{\sigma}_j^\cG)^2 \coloneqq \E_{\pr} \left[ 
    \left(X_j - \sum_{i \in \Parents{\cG}{j}} \tilde{f}_{ij}^\cG(X_i) - \tilde{S}_j + \tilde{r}_j^\cG \left( \tilde{S}_{\Parents{\cG}{j}} \right) \right)^2 \right] 
    \\
    \{ \tilde{f}_{ij}^\cG \}_{i \in \Parents{\cG}{j}}, \tilde{r}_j^\cG 
    = &\argmin_{g_{ij} \in \fkspace_i, \ w_j \in \rkspace_j^G}
    \E_{\pr} \left[ 
    \left(X_j - \sum_{i \in \Parents{\cG}{j}} g_{ij}(X_i) - \tilde{S}_j - w_j \left(\tilde{S}_{\Parents{\cG}{j}} \right) \right)^2 
    \right].
    \end{split}
\end{equation*}
\cref{eq:mod_gap} has an analogous interpretation as $\xi_p$. 
However, \cref{eq:mod_gap} assumes that the gap between the negative log-likelihood score of the best incorrect DAG and correct DAG remains non-zero even in the misspecified setting when $\tilde{S}$ is used instead of $S$. 
As we show in the previous section, if $S_j$ is smooth, then $S_j$ is close to $\tilde{S}_j$. 
If $S_j$ is close to $\tilde{S}_j$, then \cref{eq:mod_gap} is not much stronger of a condition than assuming $\xi_p > 0$. 
If we replace $S$ with $\tilde{S}$ in the proof of \cref{thm:decam_consis}, and replace $\xi_p$ with $\tilde{x}_p$ in \cref{assum:gap}, then \cref{thm:decam_consis} remains true. \\

\noindent \textbf{Connection to \citet{cam}.} In \citet{cam}, the authors similarly use the gap parameter to handle misspecification, namely when a finite-dimensional function space is used for fitting instead of a non-parametric, infinite-dimensional space (where the true functions belong in general). 
As the authors discuss after Equation 15, there is a tradeoff between identifiability (having a larger gap by using a richer function space) and variance. 
They argue that this difference can be much different than the bias-variance tradeoff. 
From this perspective, $\tilde{\xi}_p$ is a specific case of the general idea proposed in \citet{cam} to handle misspecification.  

\section{Score Function Details} \label{A:kernel_details} 

In our experiments, we used an RBF kernel for $k_{j}$ and $k_{ij}$, $1 \leq j \leq p$. We implemented the DeCAMFounder using the Gaussian process package \texttt{GPyTorch} \citep{black_box} to fit the kernel hyperparameters (i.e., by maximizing the log marginal likelihood via gradient ascent). We used a total of $100$ iterations using the \texttt{Adam} optimizer with a learning rate of $0.01$. See the \texttt{scores.py} file in the "decamfound" folder in the Github repository \texttt{https://github.com/uhlerlab/decamfound} for our \texttt{python} code.

\section{Generating Simulated Data} \label{A:sim_data_generate} 

\newcommand{\parents}{{\textrm{par}}}
\newcommand{\confound}{{\textrm{confound}}}
\newcommand{\var}{{\textrm{Var}}}

For node $X_j$, \raedit{$1 \leq j \leq p$}, with sampled trend types $\{f_{ij}\}_{i \in \Parents{\cG}{j}}$ and $\{g_{kj}\}_{k=1}^K$ and weights $\{\theta_{ij}\}_{i \in \Parents{\cG}{j}}$ and $\{\theta_{kj}^{'}\}_{k=1}^K$, let 
\[
O_j = \sum_{i \in \Parents{\Gtrue}{j}} \theta_{ij} f_{ij}(X_i) 
\quad \text{and} \quad
C_j = \sum_{k=1}^K \theta'_{kj} g_{kj}(H_k)
\]
represent the (unnormalized) variation explained by the observed and confounder nodes, respectively. 
We wish to find normalization constants $c_{\parents,j}$ and $c_{\confound,j}$ so that 
\[
\cov(c_{\parents,j} O_j + c_{\confound,j} C_j) = 1 - \sigma_{\text{noise}}^2
\]
and
\[
\cov(c_{\confound,j} C_j) = \sigma_\confound^2.
\]

Since it may be difficult to analytically solve for these normalization constants, we find them inductively using a Monte Carlo approach. 
In particular, suppose we have computed the normalization constants $c_{\parents,i}$ and $c_{\confound,j}$ for $i < j$.
Then, we take many Monte Carlo samples (we use 10,000 in the experiments) from the marginal distribution over $H, X_1, X_2, \ldots, X_{j-1}$.
This allows us to estimate $\cov(C_j)$ and solve for $c_{\confound,j}$.
Similarly, we may use these samples to estimate $\cov(O_j, C_j)$, which along with the value of $c_{\confound,j}$ and the estimate of $\var(C_j)$, allows us to solve for $c_{\parents,j}$ in the first equation. 
There are several edge cases depending on if $x_j$ has no observed and/or confounder parents:
\begin{enumerate}
	\item If $X_j$ has no observed or confounder parents (i.e., is a source), then set $c_{\parents,j} = 0$ and $c_{\confound,j} = 0$. Set the noise variance for $x_j$ equal to 1.
	
	\item  If $X_j$ has no confounder parents but at least one observed node parent, then set the signal variance for $x_j$ equal to $\sigma_{\text{signal}}^2 + \sigma_\confound^2$.
	
	\item  If $X_j$ has at least one confounder parents but no observed node parents, then set the noise variance for $X_j$ equal to $\sigma_{\text{noise}}^2 + \sigma_{\text{signal}}^2$.
\end{enumerate}
See the \texttt{generate\_synthetic\_data.py} file in our Github repository for the code.

\clearpage
\section{Sensitivity Analysis} \label{A:sensitivity} \raj{In \cref{sec:experiments}, we considered $K=1$ unobserved confounders. Here we perform a sensitivity analysis with respect to $K$, and the number of factors selected for the DeCAMFounder. To this end, we generate data as in \cref{sec:experiments} but now vary $K \in \{5, 10, 15\}$. For each choice of $K$ and confounding strength, we generate 25 datasets.  The scree plots for each dataset configuration in \cref{fig:scree_plot_sens} illustrate how both $K$ and the confounding strength affect the spectrum of the sample covariance matrix. When $K$ or the confounding strength increases, the proportion of large eigenvalues increases, indicating that a larger number of latent factors $J$ must be selected. When the confounding $R^2$ is strong (i.e., greater than or equal to 0.50), then the "elbow" in the scree plots are pronounced, and occur at 10, 20, and 30 for $K=5, 10$, and 15, respectively. This pronounced "elbow" occurs because stronger confounders make the eigenvalues corresponding to the latent factors more spiked. When the confounding $R^2$ is less than 0.5, then it is harder to estimate the number of latent factors since the eigenvalues are less spiked.} \\

\begin{figure}
        \centering
        \begin{subfigure}[b]{0.45\textwidth}
            \centering
            \myfigure{width=\textwidth}{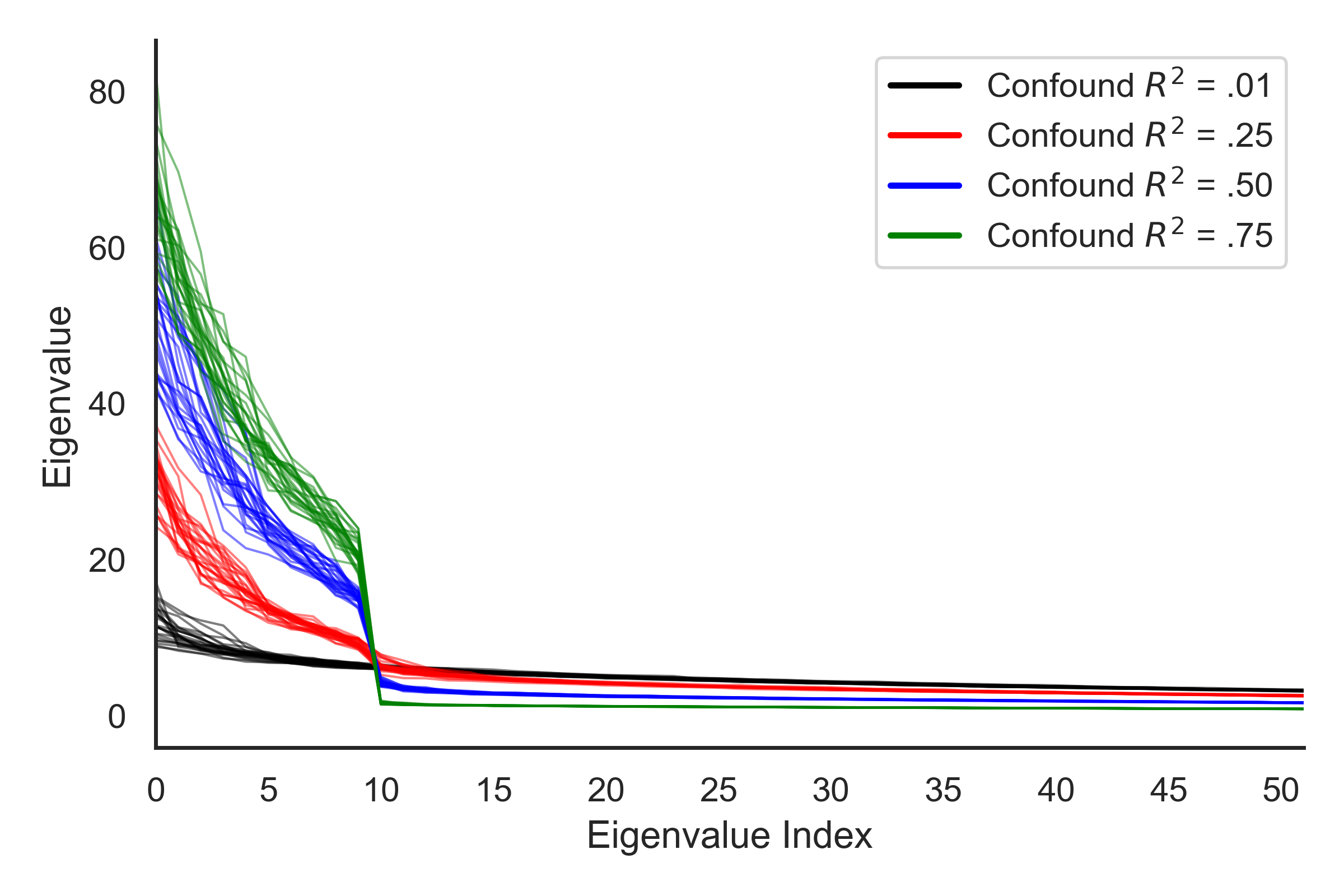}
            \caption{\small $K=5$} \label{fig:scree_plot_sens5} 
        \end{subfigure}
        \hfill
        \begin{subfigure}[b]{0.45\textwidth}  
            \centering 
            \myfigure{width=\textwidth}{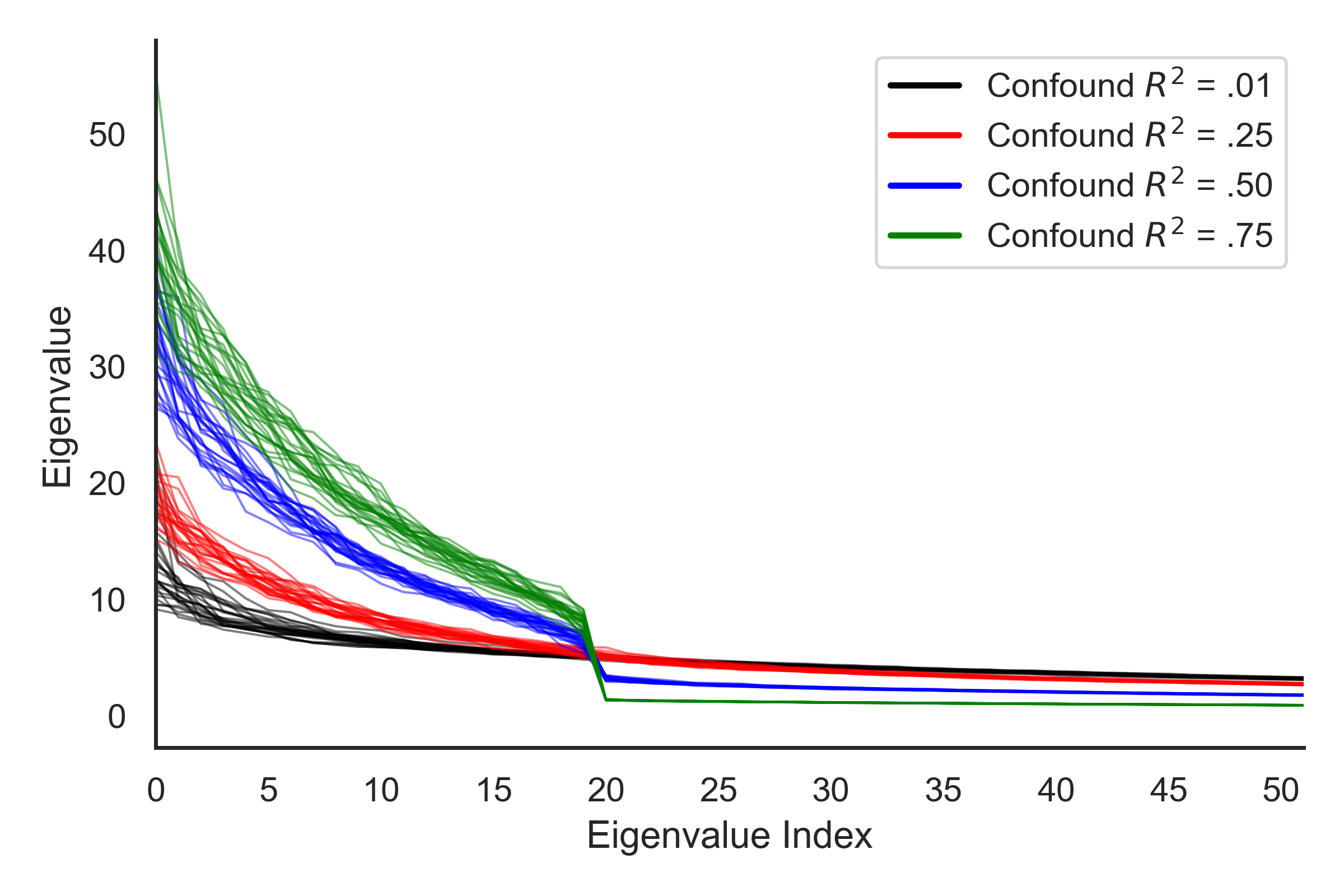}
           \caption{\small$K=10$} \label{fig:scree_plot_sens10}  
        \end{subfigure}
        \hfill
        \begin{subfigure}[b]{0.45\textwidth}  
            \centering 
            \myfigure{width=\textwidth}{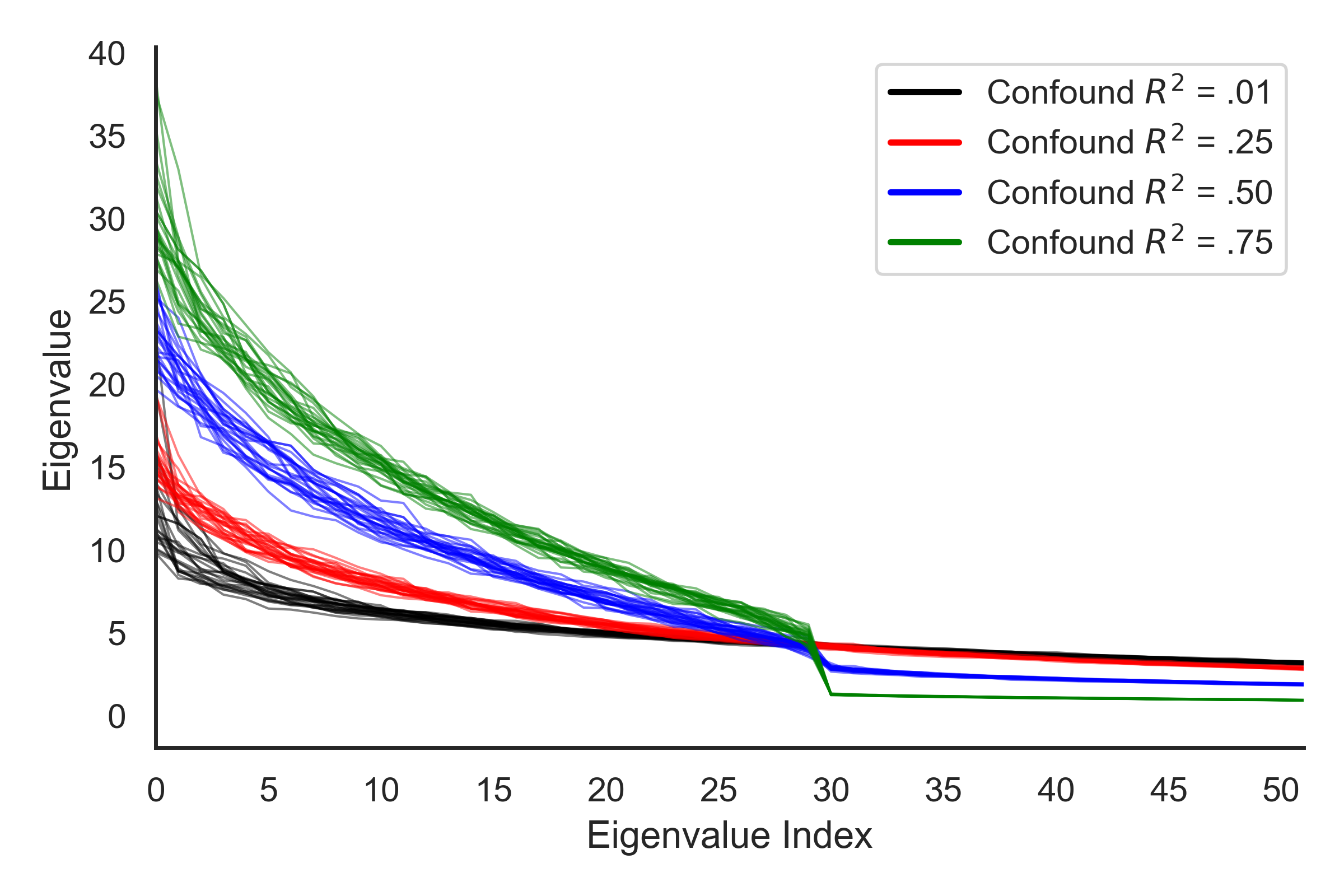}
           \caption{\small $K=15$} \label{fig:scree_plot_sens15}   
        \end{subfigure}
        	\caption{Scree plots for each choice of $K$ and confounding strength. There are 25 lines for "Confound $R^2$" since we generated 25 datasets for each configuration of $K$ and confounder $R^2$.} \label{fig:scree_plot_sens}
\end{figure}

\noindent \raj{\textbf{Sensitivity to the number of unobserved confounders.} Based on \cref{fig:scree_plot_sens}, the true number of latent factors appears to equal 10, 20, and 30 for $K=5, 10$, and 15, respectively. \cref{tab:sens_cor_rank} illustrates how the results vary for the DeCAMFounder as $K$ increases. In particular, we consider the two tasks described in \cref{sec:experiments}: the Wrong Parent Addition Task and the Correct Parent Deletion task. Instead of showing the distribution of the Prop. Times MLL Wrong $>$ MLL True metric across each dataset configuration (e.g., as in \cref{fig:non_linear_par_tasks}), we compute the average of this metric in \cref{tab:sens_cor_rank} so that we can compare across different choices of $K$. To focus on the effect of $K$ on performance, we set the number of latent factors equal to the correct number of factors. We see that the results in  \cref{tab:sens_cor_rank} do not change significantly  relative to the $K=1$ results in \cref{fig:non_linear_par_tasks}. Hence, the performance of the DeCAMFounder appears to be driven by the total variation explained by the confounders (which determine the "spikiness" of the eigenvalues) rather than the number of unobserved confounders.} \\

\begin{figure}
    \centering
    \caption{Performance sensitivity to the number of unobserved confounders for the DeCAMFounder}
\label{tab:sens_cor_rank}
\resizebox{\textwidth}{!}{%
\begin{tabular}{c c c c}
\toprule
 \# Confounders ($K$) &  Confound $R^2$ & Wrong Parent Addition & Correct Parent Deletion \\
\midrule
                         5 &           0.01 &                                              0.10\% &                                              3.90\% \\
                          &           0.25 &                                              4.80\% &                                              9.50\% \\
                          &           0.50 &                                              6.50\% &                                             15.50\% \\
                          &           0.75 &                                             12.90\% &                                             48.00\% \\
                        10 &           0.01 &                                              0.40\% &                                              8.60\% \\
                          &           0.25 &                                              2.80\% &                                             11.70\% \\
                          &           0.50 &                                              3.20\% &                                             19.80\% \\
                          &           0.75 &                                              4.00\% &                                             44.00\% \\
                        15 &           0.01 &                                              0.60\% &                                              0.60\% \\
                          &           0.25 &                                              2.70\% &                                              6.80\% \\
                          &           0.50 &                                              1.30\% &                                             24.30\% \\
                          &           0.75 &                                              4.00\% &                                             46.90\% \\
\bottomrule
\end{tabular}
}
\end{figure}

\begin{figure}[!t]
    \centering
        \caption{Performance sensitivity to the number of latent factors selected for the DeCAMFounder}
\label{tab:sens_rank}
\resizebox{\textwidth}{!}{%
\begin{tabular}{c c c c c}
\toprule
 \# Confounders ($K$) &  Confound $R^2$ & Bias Type? & Wrong Parent Addition & Correct Parent Deletion \\
\midrule
                 5 &           0.01 &      Under &                 0.20\% &                   4.70\% \\
                  &             &       Over &                 0.10\% &                   3.90\% \\
                  &           0.25 &      Under &                 7.60\% &                   7.00\% \\
                  &             &       Over &                 3.90\% &                  10.00\% \\
                10 &           0.01 &      Under &                 0.40\% &                   7.40\% \\
                  &             &       Over &                 0.40\% &                   8.60\% \\
                  &           0.25 &      Under &                 2.50\% &                  10.90\% \\
                  &             &       Over &                 2.10\% &                  14.20\% \\
                15 &           0.01 &      Under &                 1.00\% &                   0.60\% \\
                  &             &       Over &                 0.70\% &                   0.60\% \\
                  &           0.25 &      Under &                 2.80\% &                   5.10\% \\
                  &             &       Over &                 2.80\% &                   5.80\% \\
\bottomrule
\end{tabular}
}
\end{figure}

\noindent \raj{\textbf{Sensitivity to the number of latent factors selected.} When the confounding strength equals 0.01 and 0.25, estimating the number of latent factors is more challenging since the factors are weaker; see \cref{fig:scree_plot_sens}. In \cref{tab:sens_rank}, we consider two types of biases: underestimating or overestimating the number of latent factors. For underestimation, we set the number of latent factors as 5 fewer than the correct number of latent factors (denoted as "Under"). For overestimation, we set the number of latent factors as 5 more than the correct number of latent factors (denoted as "Over"). The values in \cref{tab:sens_rank} are similar to those in \cref{tab:sens_cor_rank}, suggesting that the DeCAMFounder is robust to the number of latent factors selected. One reason for this robustness is that when the confounding $R^2$ is small (which results in weak latent factors), then the signal strength of the observed variables is larger. Hence, the DAG learning task is easier, allowing for a larger margin of error when estimating the number of latent factors to select in the DeCAMFounder.}

\section{Additional Figures and Experiments} \label{A:add_exp} 
\subsection{Synthetic Data Experiments} \label{sec:add_syn}
In \cref{sec:experiments}, we reported on the proportion of times the incorrect parent set was selected over the true parent set. We report on the following additional metric to understand how confidently wrong (or correct) each method is:
\begin{enumerate}
	\item \textbf{Log Odds (Wrong vs. True)}: each score equals the log marginal likelihood (technically an approximation for BIC) of the parent set. Assuming a uniform prior over the set of all parent sets,  the log odds (LO) between the wrong and true parent set reduces into the difference between scores:
	\begin{equation}
		\begin{split} \label{eq:app_log_odds}
		\text{LO}_i &= \log \frac{\pr(X \mid P_t) \pr( P_{\text{t}} )}{\pr(X \mid P_{\text{correct}}) \pr( P_{\text{correct}} ) } \\
						&= \text{score}(P_t) - \text{score}(P_{\text{correct}}).
		\end{split}
	\end{equation}			
	
	We report $\max_{1 \leq t \in T} \text{LO}_t$ for the Wrong Parent Addition task in \cref{fig:parent_set_addition} and for the Correct Parent Deletion task in  \cref{fig:parent_set_removal}. A higher value is worse (i.e., it indicates that a method places higher confidence in an incorrect parent set than the true parent set).
\end{enumerate}

%
\begin{sidewaysfigure}[h]
        \centering
        \begin{subfigure}[b]{0.495\textwidth}
            \centering
            \myfigure{width=\textwidth}{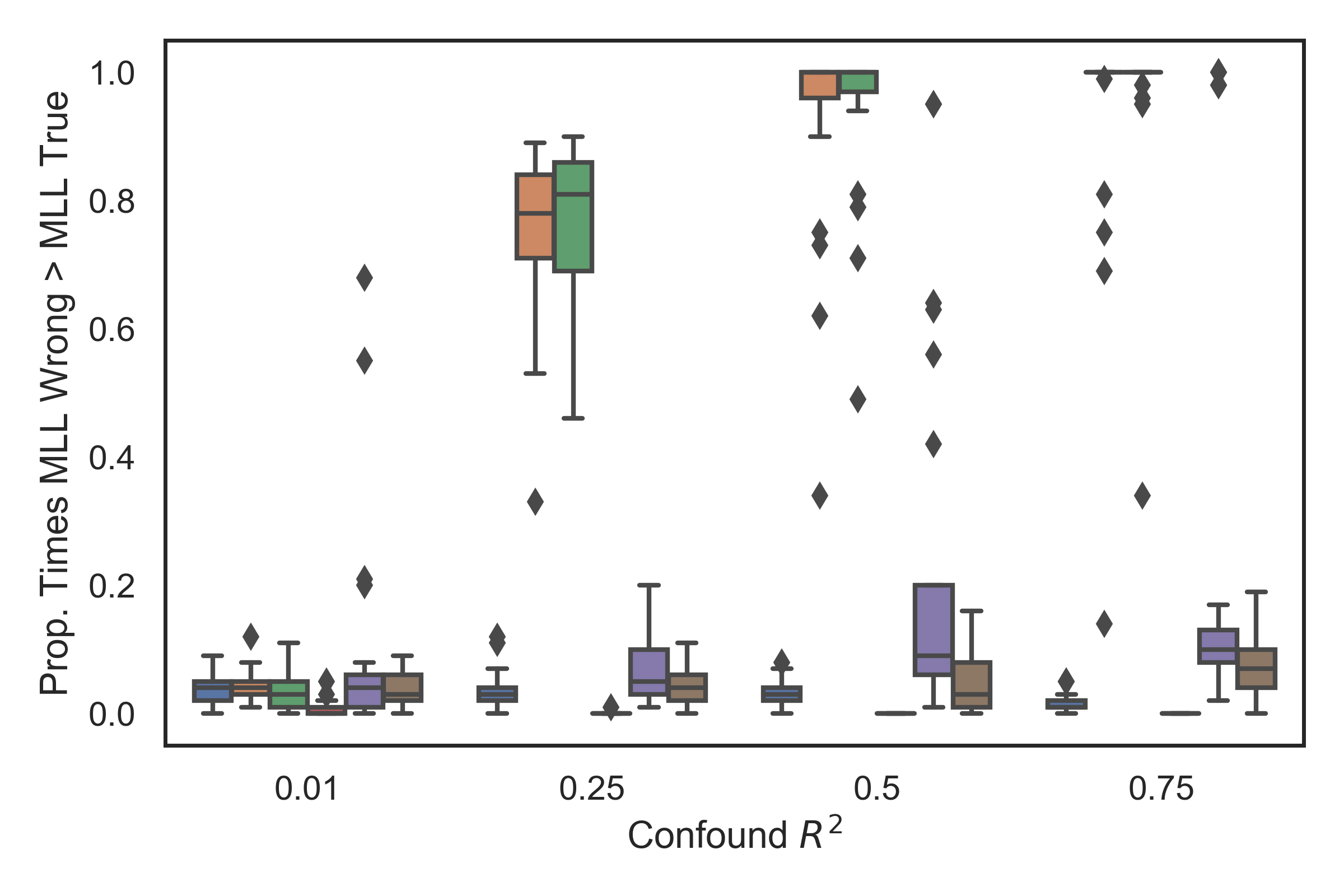}
            \caption[Network2]%
            {{\small Linear SEM}}    
            \label{fig:parent_set_linear_prop}
        \end{subfigure}
        \hfill
        \begin{subfigure}[b]{0.495\textwidth}  
            \centering 
            \myfigure{width=\textwidth}{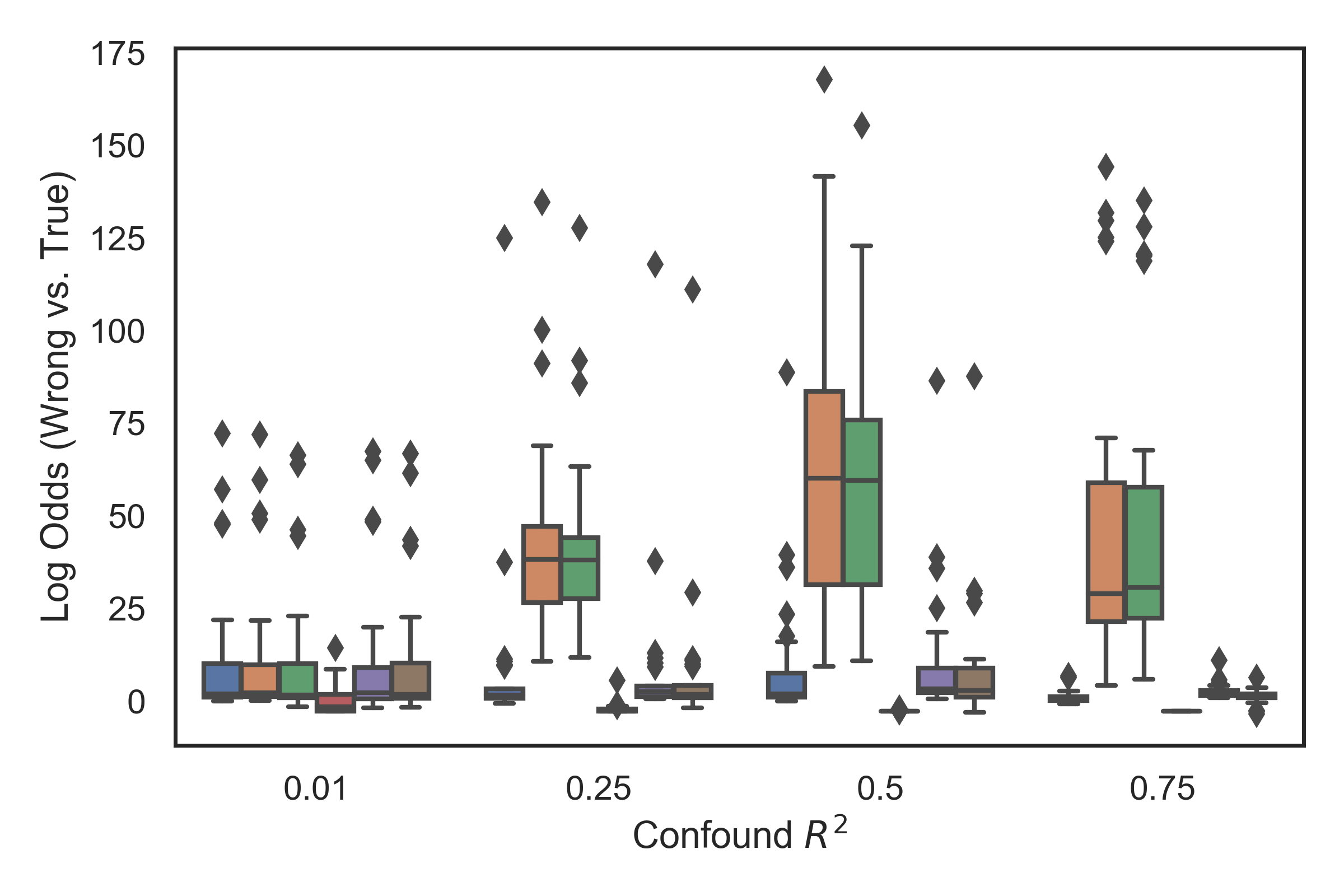}
            \caption[]%
            {{\small Linear SEM}}    
            \label{fig:parent_set_linear_mag}
        \end{subfigure}
        \vskip\baselineskip
        \begin{subfigure}[b]{0.495\textwidth}   
            \centering 
            \myfigure{width=\textwidth}{./figures/parent_set_non_linear_prop_metric.png}
            \caption[]%
            {{\small Non-Linear SEM}}    
            \label{fig:parent_set_non_linear_prop}
        \end{subfigure}
        \hfill
        \begin{subfigure}[b]{0.495\textwidth}   
            \centering 
            \myfigure{width=\textwidth}{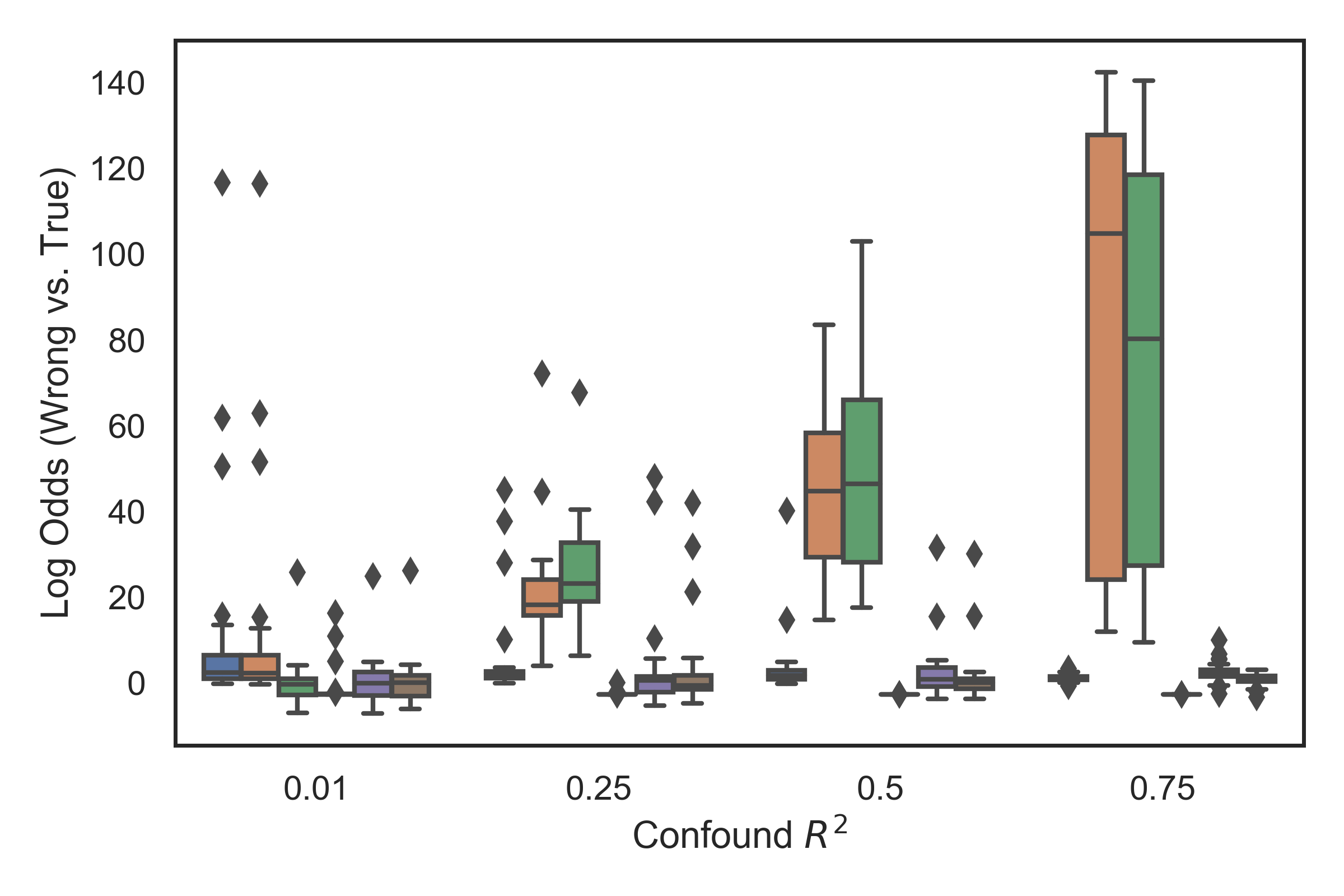}
            \caption[b]%
            {{\small Non-Linear SEM}}    
            \label{fig:parent_set_non_linear_mag}
        \end{subfigure}
        \caption{Results for the Wrong Parent Addition task.  25 total simulations per dataset configuration were performed. See \cref{sec:decam_perform} and \cref{sec:add_syn} for a description of the performance metrics.} \label{fig:parent_set_addition}
\end{sidewaysfigure}
%


\begin{sidewaysfigure}[h]
        \centering
        \begin{subfigure}[h]{0.495\textwidth}
            \centering
            \myfigure{width=\textwidth}{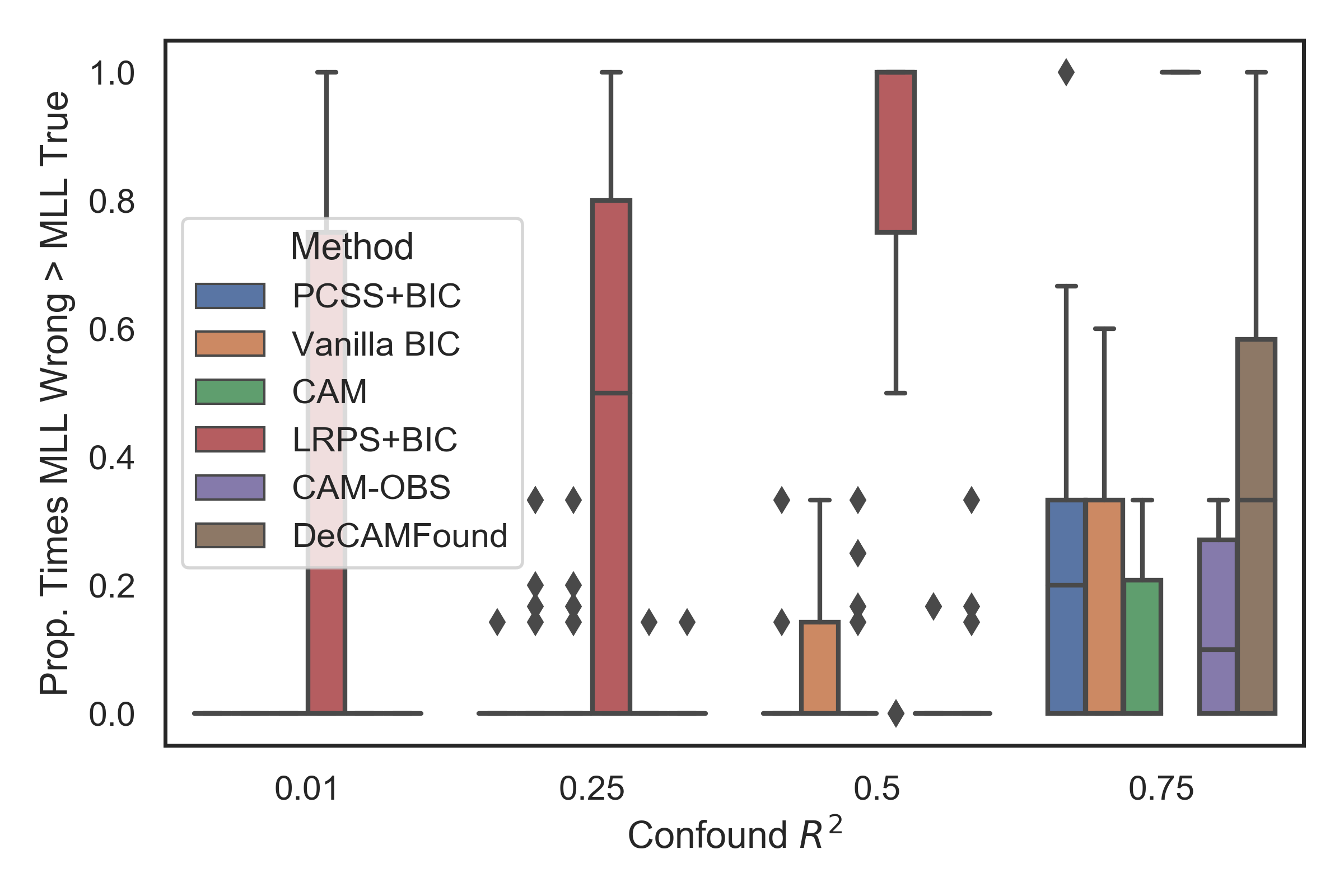}
            \caption[Network2]%
            {{\small Linear SEM}}    
            \label{fig:remove_parent_set_linear_prop}
        \end{subfigure}
        \hfill
        \begin{subfigure}[h]{0.495\textwidth}  
            \centering 
            \myfigure{width=\textwidth}{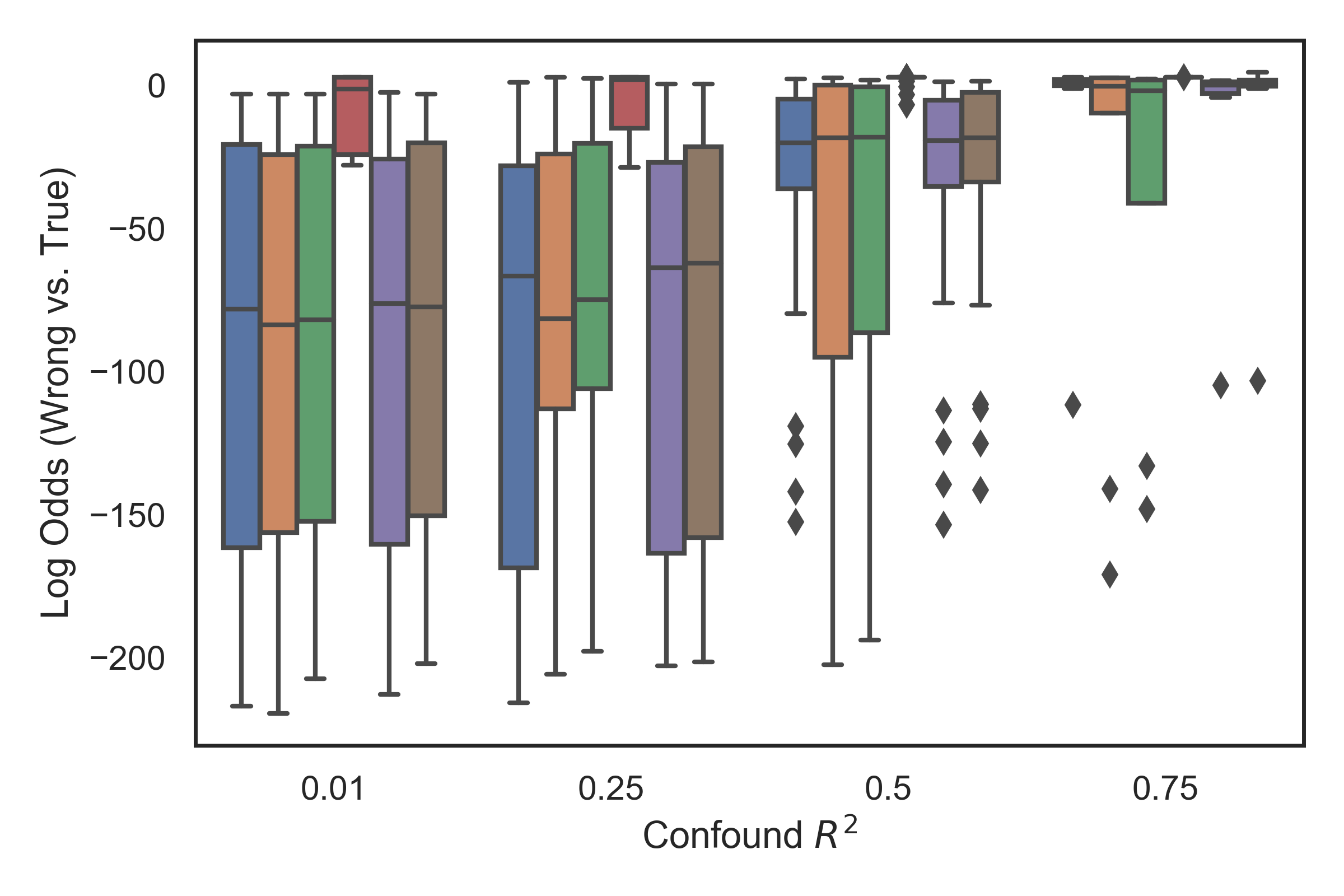}
            \caption[]%
            {{\small Linear SEM}}    
            \label{fig:remove_parent_set_linear_mag}
        \end{subfigure}
        \vskip\baselineskip
        \begin{subfigure}[h]{0.495\textwidth}   
            \centering 
            \myfigure{width=\textwidth}{./figures/remove_parent_set_non_linear_prop_metric.png}
            \caption[]%
            {{\small Non-Linear SEM}}    
            \label{fig:remove_parent_set_non_linear_prop}
        \end{subfigure}
        \hfill
        \begin{subfigure}[h]{0.495\textwidth}   
            \centering 
            \myfigure{width=\textwidth}{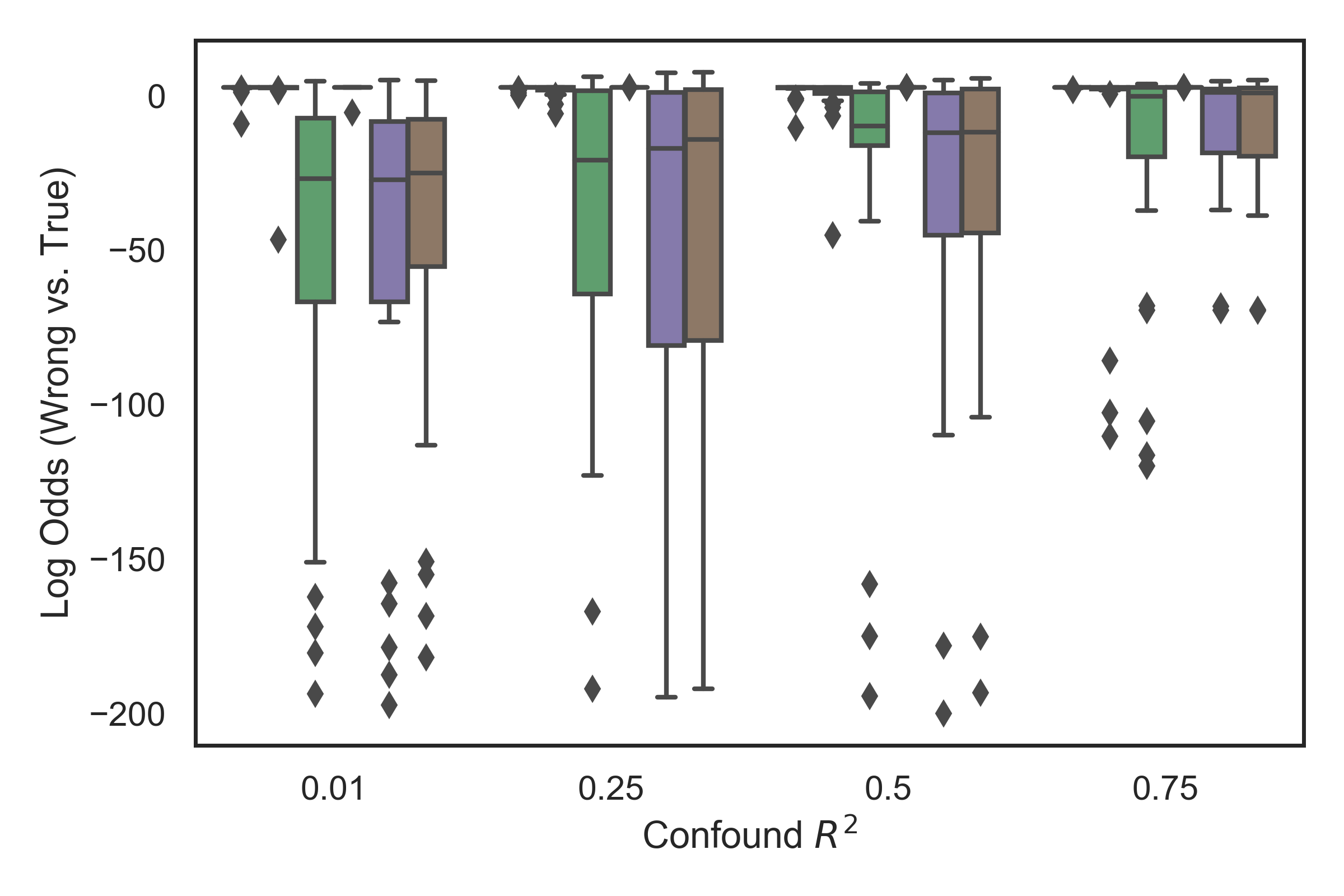}
            \caption[h]%
            {{\small Non-Linear SEM}}    
            \label{fig:remove_parent_set_non_linear_mag}
        \end{subfigure}
        \caption{Results for the Correct Parent Deletion task.  25 total simulations per dataset configuration were performed. See \cref{sec:decam_perform} and \cref{sec:add_syn} for a description of the performance metrics.} \label{fig:parent_set_removal}
\end{sidewaysfigure}

\noindent \textbf{Scoring Candidate DAG Results.} We score a candidate set of $T=100$ incorrect DAGs $\mathcal{G}$, built from randomly adding or deleting edges multiple times, starting from the true DAG (i.e., the natural extension of our two parent set evaluation tasks). Since scoring a single graph takes $O(pN^3)$ for the non-linear methods,  we consider fewer settings (i.e., fix the confounding variance to be equal to the signal variance), and only do 10 total simulations instead of 25. These results are shown in \cref{fig:dag_syn_results}. We report on two metrics:
\begin{enumerate}
	\item \textbf{Avg. Posterior SHD:} equals $\sum_{t=1}^T \text{SHD}(G_t, \Gtrue) \pr(G \mid X)$, where SHD denotes the structural hamming distance to the true graph, and $\pr(G \mid X)$ equals the posterior probability of a graph computed from  renormalizing the log marginal likelihood scores. Lower is better.
	\item \textbf{SHD Between MAP and True DAG:} reports the SHD from the true DAG for the 
	  highest scoring graph (i.e., the maximum a posteriori estimate). Lower is better.
\end{enumerate}
\cref{fig:dag_syn_results} shows that linear methods, even if they account for confounding, suffer in the non-linear setting. This advantage of modeling non-linearities agrees with the results, for example, in \citet{cam}. \cref{fig:dag_syn_results} also shows that, as in the Correct Parent Deletion Task, LRPS suffers even in the linear setting  because the induced undirected graph is very sparse (and hence often favors deleting true edges). Finally, \cref{fig:dag_syn_results} shows that both CAM-OBS (trivially) and our method are robust to both non-linearities and confounding. 

\begin{figure}[h]
        \centering
        \begin{subfigure}[b]{\textwidth}
            \centering
            \myfigure{width=\textwidth}{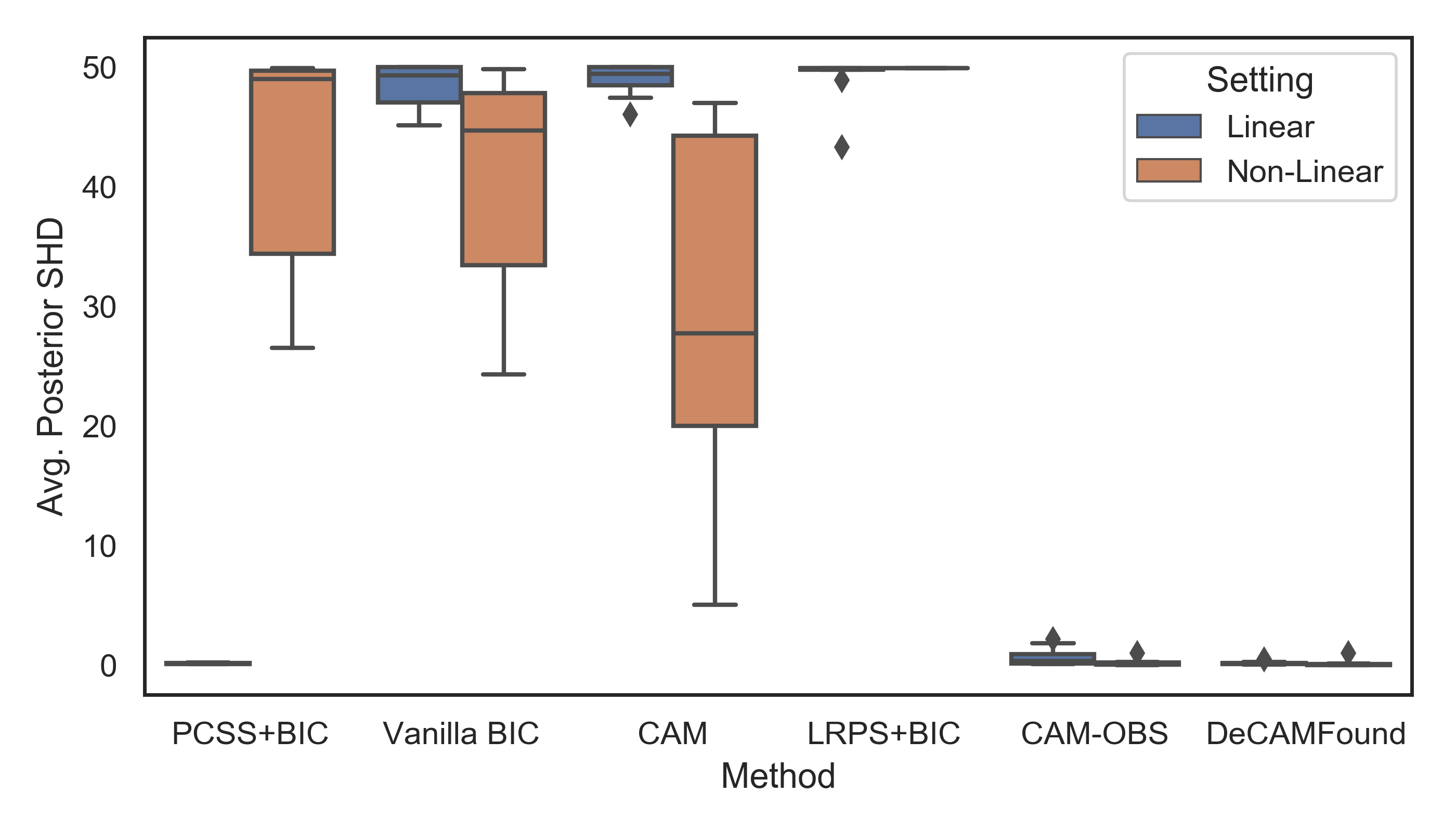}
            \caption[Network2]%
            {{\small SHD}}    
            \label{fig:dag_shd}
        \end{subfigure}
        \vfill
        \begin{subfigure}[b]{\textwidth}  
            \centering 
            \myfigure{width=\textwidth}{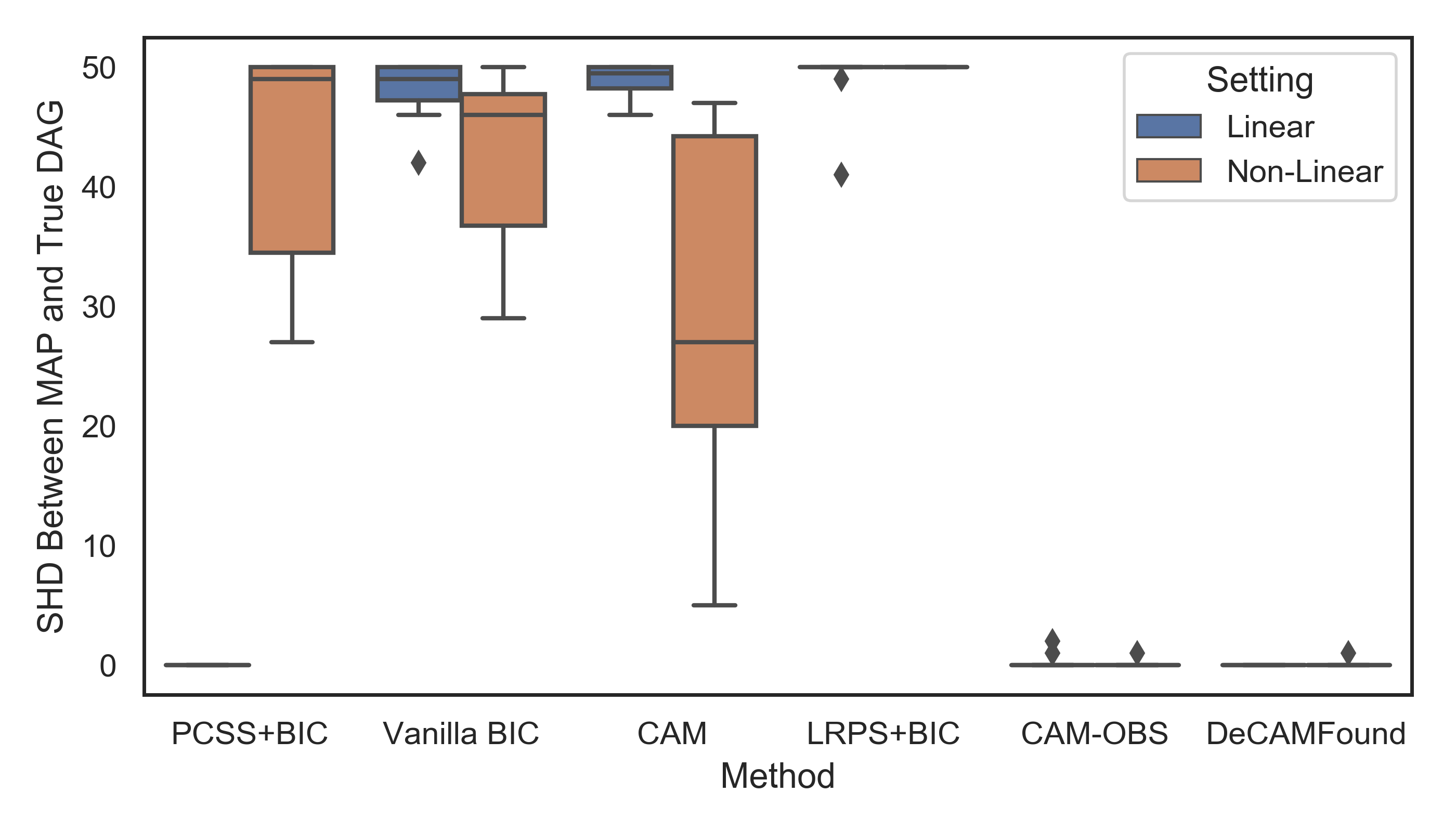}
            \caption[]%
            {{\small Average SHD}}    
            \label{fig:dag_mag}
        \end{subfigure}
     \caption{Results for the candidate DAG scoring task.  10 total simulations per dataset configuration (i.e., linear / non-linear) were performed.}  \label{fig:dag_syn_results}
\end{figure}

\subsection{Ovarian Cancer Dataset}

\noindent \textbf{Scree Plot.} Analyzing the spectrum of the data matrix consisting of the 486 observed genes in \cref{fig:scree_plot} shows that there are about 7 spiked eigenvalues. \\

\noindent \textbf{Pervasive TF Gene Correlations.} 7/15 TFs have edges with more than 75 genes according to NetBox (i.e., these TFs might play the role of pervasive confounders). We summarize the absolute value of the correlations between these 7 TFs and the 486 genes in \cref{fig:tf_corrs}. \\

\noindent \textbf{Removing the Effect of a Latent TF via PCSS.} See \cref{fig:real_deconfound_steps}.


%
\begin{figure}
            \centering 
            \myfigure{width=0.85\textwidth}{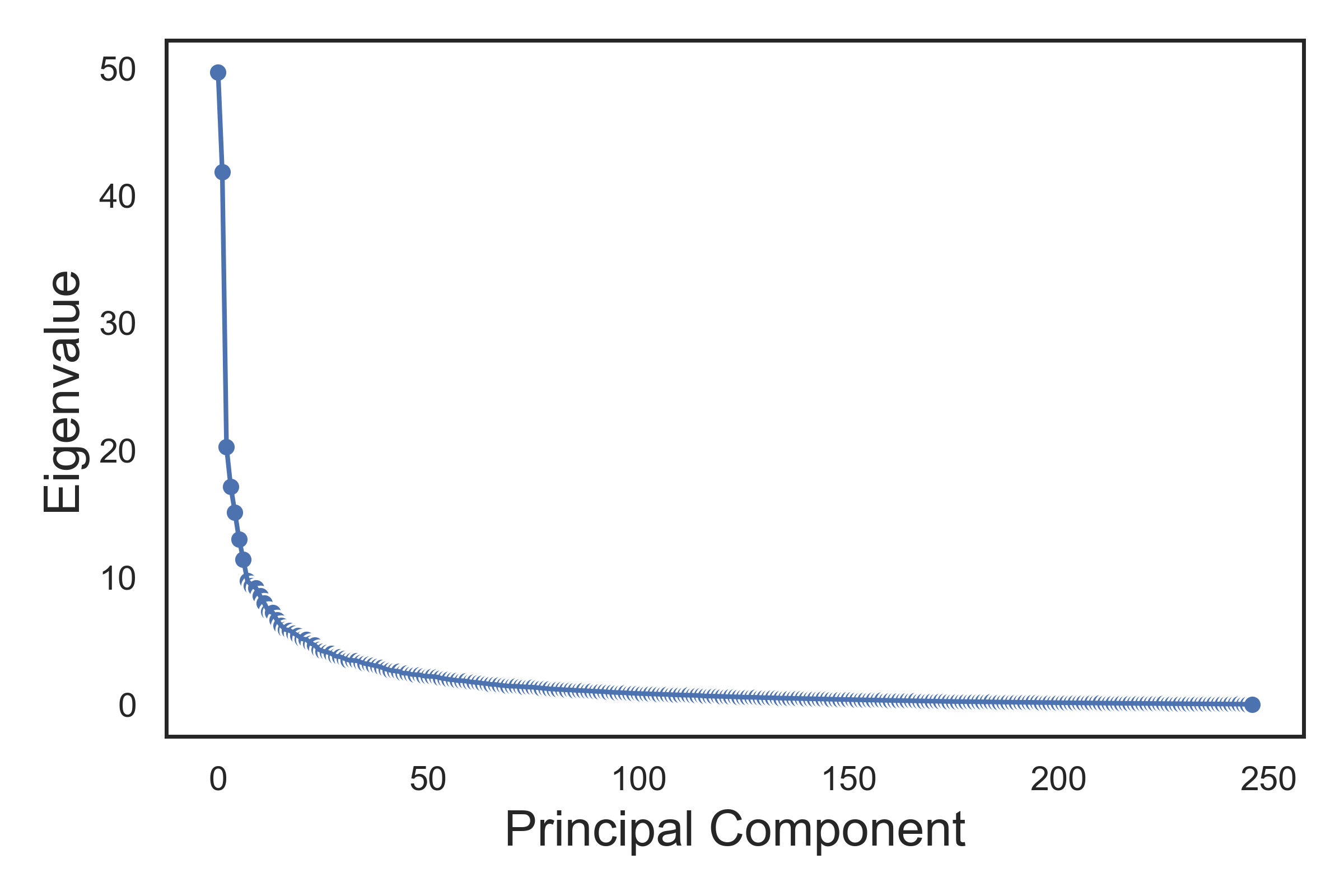}
        	\caption{PCA scree plot when the input data matrix consists of the 486 observed genes. Based on this scree plot, we select $K=7$ components for the spectral methods.} \label{fig:scree_plot}
\end{figure}

\begin{figure}[h]
\centering
\myfigure{width=0.85\textwidth}{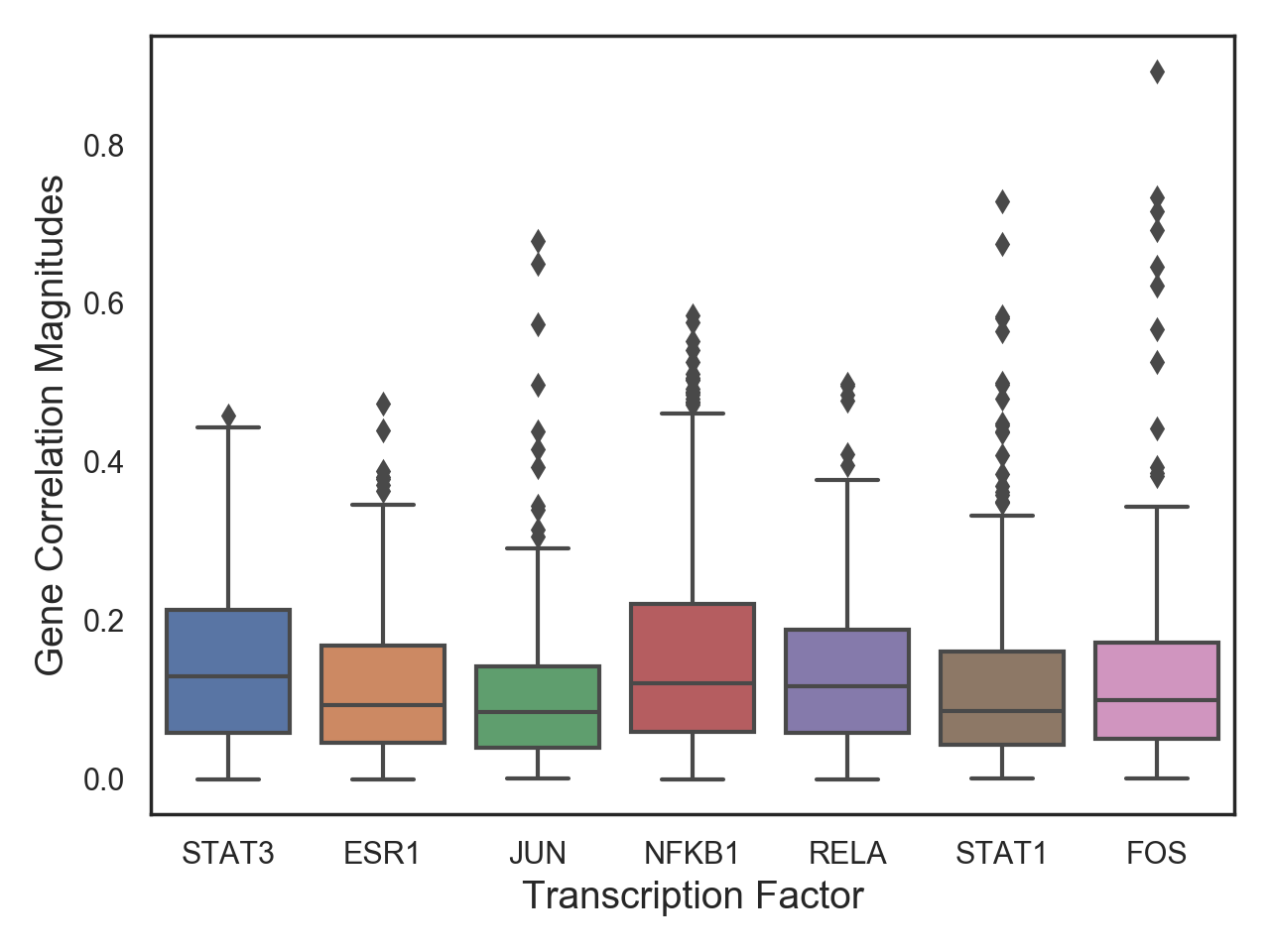} 
\caption{Absolute value of correlations between TFs and the 486 observed genes. Only TFs with at least 75 edges are shown.} \label{fig:tf_corrs}
\end{figure}


\begin{figure}
        \centering
        \begin{subfigure}[b]{0.9\textwidth}
            \centering
            \myfigure{width=0.9\textwidth}{./figures/ovarian/NFKB1_BIRC3.png}
        \end{subfigure}
        \hfill
        \begin{subfigure}[b]{0.9\textwidth}  
            \centering 
            \myfigure{width=0.9\textwidth}{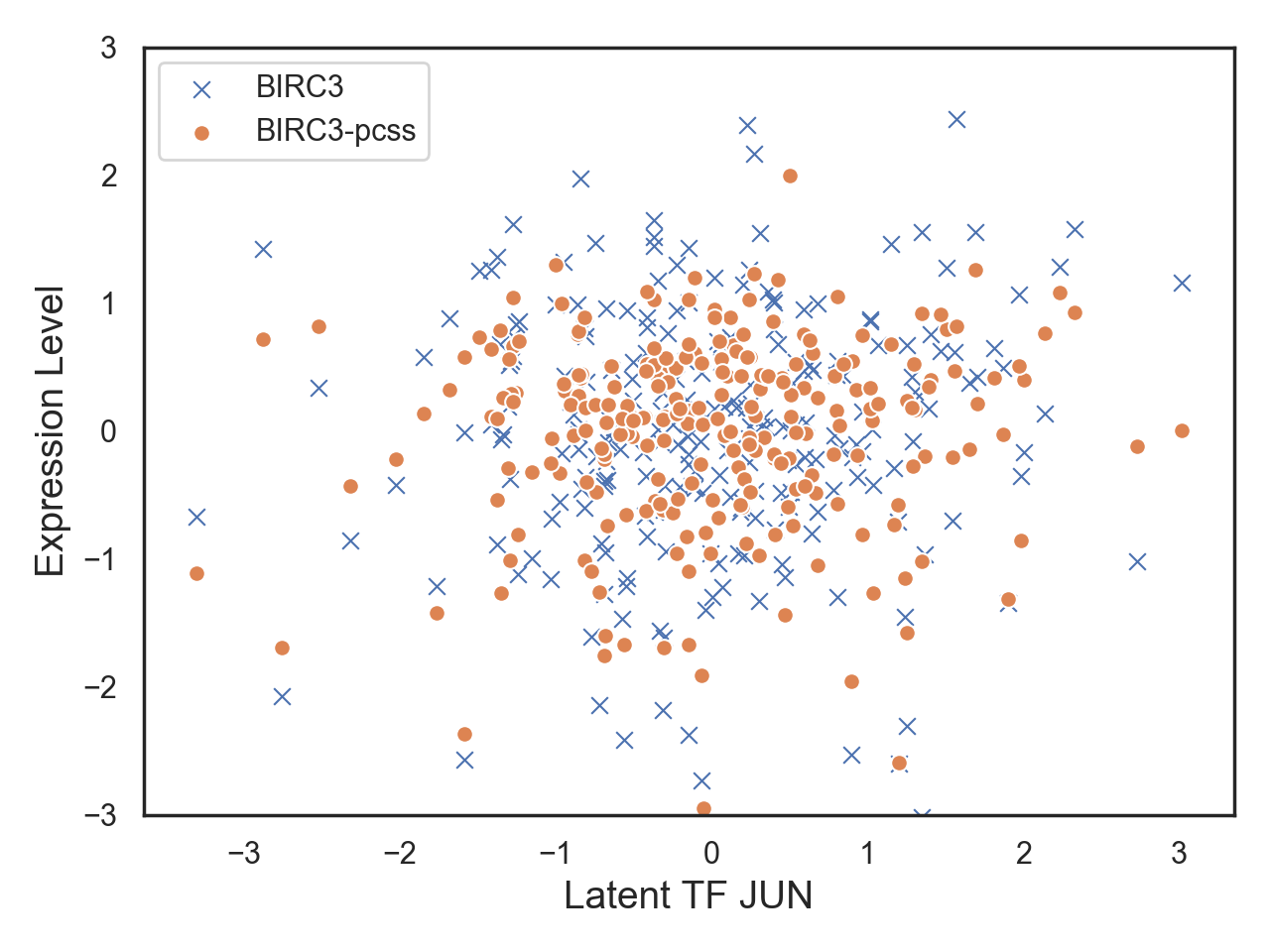}
        \end{subfigure}
        	\caption{Out of the 15 latent TFs, BIRC3 has the highest absolute correlation with NFKB1 and the smallest absolue correlation with JUN. BIRC3-pcss refers to the total latent confounding variation estimated for that gene via PCSS. BIRC3 refers to the actual observed values of the gene. BIRC3-pcss is correlated with NFKB1 (which is correlated with BIRC3) and not correlated with JUN.} \label{fig:gene_vs_pcss}
\end{figure}

\begin{figure}[h]
\centering
\myfigure{width=\linewidth}{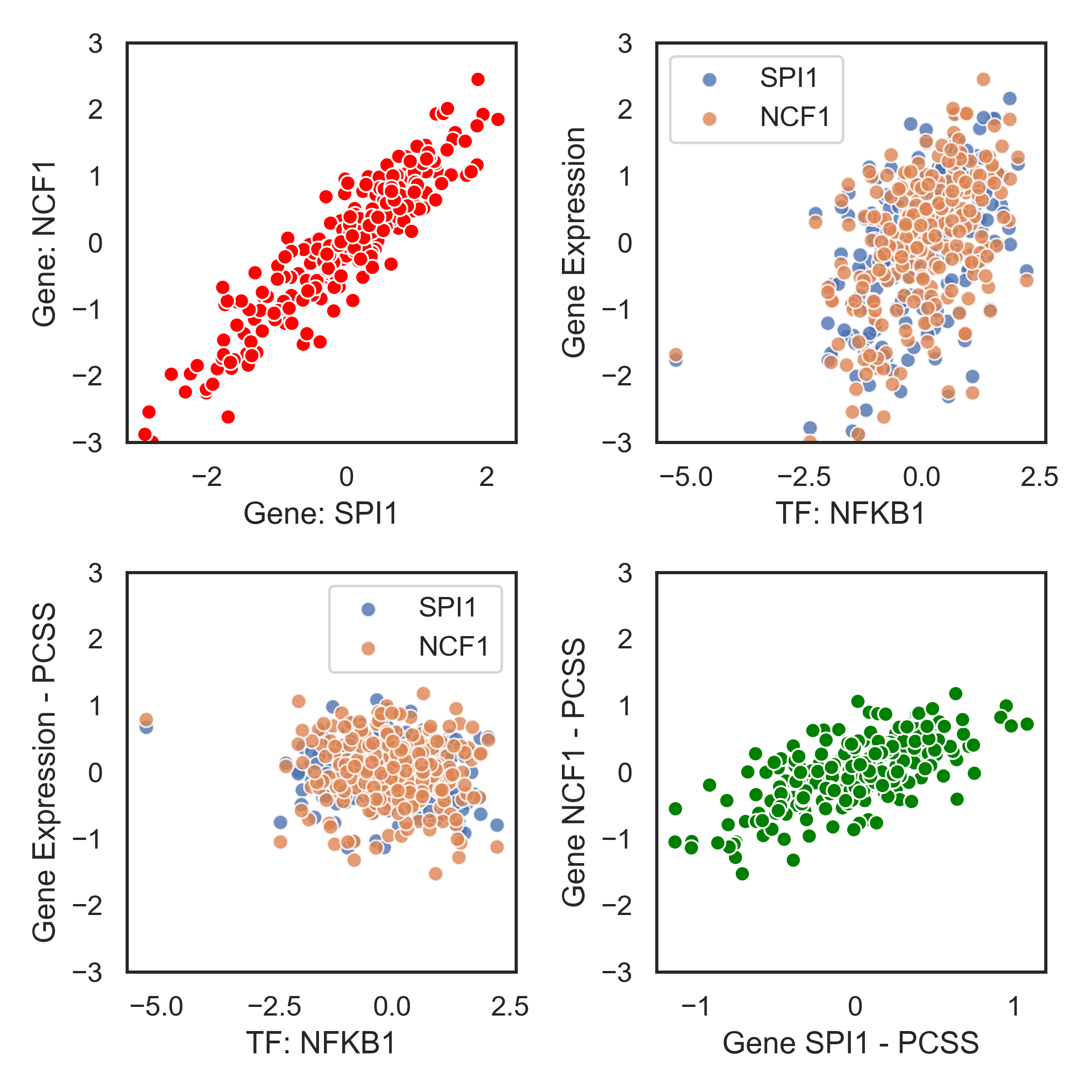}
\caption{Top left: scatterplot of two genes that are conditionally independent given each parents' gene neighborhood sets and TFs but dependent when removing the TFs. Top right: scatter plot of each gene for the TF that has the highest correlation with both genes. Bottom left: correlation with the transcription factor after removing the estimated confounder sufficient statistics from each gene. Bottom right: weaker correlation after removing the  confounder sufficient statistics from each gene. Since both genes are still marginally dependent given the TFs without conditioning on the parent sets, the genes are still correlated in the bottom right figure.}  \label{fig:real_deconfound_steps}
\end{figure}

\end{document}